%% file: main.tex
\documentclass[manuscript,screen,nonacm]{acmart}
\input{latex/packages} 				%
\input{latex/thm-env}  				%
\input{latex/hyperref-config}       %
\input{latex/macro-english}         %
\input{latex/macro-tikz}			%
\usepackage{latex/fbarrows}         %
\input{latex/macro-beluga}          %

\input{latex/macro-symbols} 		%

\input{latex/macro}   				%

\begin{document}
	\input{latex/metadata} %
	{
    \let\clearpage\relax

   	\include{sections/newintro}

	\include{sections/axioms}
	\include{sections/uniqueness}

	\include{sections/ccc}

	\include{sections/proved-lts}

	\include{sections/complementarity}

	\include{sections/relations-pccsk}
	\include{sections/key-based}
   	\include{sections/newconclusion}

	}
	
	\bibliographystyle{ACM-Reference-Format}
	\bibliography{bib/bib}
	
	\appendix
	{
	\let\clearpage\relax
	}

\end{document}

%% file: latex/packages.tex
\usepackage{mathtools} %
\usepackage[inline]{enumitem} %
\usepackage{xspace} %
\usepackage{multirow} %
\usepackage{tabularx} %
\usepackage[caption=false]{subfig} %
\usepackage{multicol} %
\usepackage{ebproof} %
\usepackage{csquotes} %
\usepackage{tikz} %
\usepackage{diagbox} %
\usepackage{thm-restate}

\usepackage{cancel} %

\usepackage{graphicx}
\usepackage[export]{adjustbox}

\usepackage{siunitx}

\usepackage{ushort} %

\usepackage[most]{tcolorbox} %
\definecolor{Gray}{gray}{0.85}
\tcbset{ %
	halign=center,
	colbacktitle=black!15!white, 
	colframe=black!15!white, 
	coltitle=black, 
	colback=white, 
	titlerule=0pt, 
	boxsep=1pt,
	left=1pt,
	right=1pt,
	top=1pt,
	bottom=1pt, 
	middle=0pt,
	lower separated=false,
	sidebyside gap=1pt,
	every box on layer 2/.style={  %
		sidebyside=false,
		colbacktitle=black!5!white,
		colframe=black!5!white,
	} 
}

%% file: latex/thm-env.tex
\AtEndPreamble{%
	\theoremstyle{acmdefinition}
	\newtheorem{notation}[theorem]{Notation}
	\newtheorem{remark}[theorem]{Remark}
	\newtheorem{discussion}[theorem]{Beluga's Breach}%
}

%% file: latex/hyperref-config.tex
\makeatletter
	\providecommand*{\toclevel@theorem}{0}
	\providecommand*{\toclevel@proposition}{0}
	\providecommand*{\toclevel@definition}{0}
	\providecommand*{\toclevel@notation}{0}
	\providecommand*{\toclevel@example}{0}
	\providecommand*{\toclevel@description}{0}
	\providecommand*{\toclevel@lemma}{0}
	\providecommand*{\toclevel@remark}{0}
	\providecommand*{\toclevel@conjecture}{0}
\makeatother

\newcommand{\itemref}[1]{%
	(\ref{#1})\@\xspace%
}

\hypersetup{hypertexnames=false}

%% file: latex/macro-english.tex
\newcommand*{\eg}{e.g.\@,\xspace}
\newcommand*{\st}{s.t.\xspace}

\newcommand*{\ie}{i.e.\@,\xspace}

\newcommand*{\aka}{a.k.a.\@\xspace}
\newcommand*{\resp}{resp.\@\xspace}

\makeatletter
\newcommand*{\etc}{%
	\@ifnextchar{.}%
	{etc}
	{etc.\@\xspace}%
}
\makeatother

%% file: latex/macro-tikz.tex
\usetikzlibrary{arrows.meta}
\usetikzlibrary{calc}
\usetikzlibrary{decorations.pathmorphing}
\usetikzlibrary{shapes.geometric}
\usetikzlibrary{backgrounds}
\usetikzlibrary{arrows.meta}
\usetikzlibrary{decorations.pathreplacing}
\usetikzlibrary{shapes.multipart}
\usetikzlibrary{positioning}
\usetikzlibrary{tikzmark} %
\usetikzlibrary{arrows}  %
\usetikzlibrary{fit}
\usetikzlibrary{hobby}

\tikzset{
	sys/.style = {
		line width=.5mm,
		inner sep=0pt,
		outer sep=0pt,
		rounded corners=0.5cm,
		draw = tsdep,
		minimum size=1.5cm,
		text width=1.6cm,
		text centered
	},
	spl/.style = {
		rectangle split,
		rectangle split horizontal,
		rectangle split parts=2,
		rectangle split part fill={sdep!20,ind!20},
	},
	sya/.style= {
		thick,
		font = {\LARGE},	
	},
	rel/.style = {
		line width=1mm, %
	},
	adj/.style = {
		line width=1mm, %
		dotted,
	},
	cpt/.style = {
	},
	u/.style = {yshift=#1\pgflinewidth},
	s/.style = {xshift=#1\pgflinewidth},
}

\tikzstyle{f} = [-> %
] %
\tikzstyle{b} = [decorate, %
decoration={snake, amplitude=-1pt, segment length=1.5mm, pre=lineto, post length=2pt, pre length=1pt},%
-{Straight Barb[scale=0.8]} %
] %
\tikzstyle{fb} = [arrows={- angle 45},%
-{>[sep=-1pt]>}%
] %

\tikzstyle{pf} = [{Bar[]}-> %
] %
\tikzstyle{pb} = [decorate, %
decoration={snake, amplitude=-1pt, segment length=1.5mm, pre=lineto, post length=2pt, pre length=1pt},%
{Bar[]}-{Straight Barb[scale=0.8]} %
] %
\tikzstyle{pfb} = [arrows={- angle 45},%
{Bar}-{>[sep=-1pt]>}%
] %

%% file: latex/macro-beluga.tex
\newcommand*{\emacsfont}{\fontfamily{lmtt}\selectfont}

\definecolor{belugapurple}{RGB}{154,0,214}
\definecolor{belugared}{RGB}{225, 0, 0}
\definecolor{belugagreen}{RGB}{31,147,15}
\definecolor{belugablue}{RGB}{0,0,238}
\definecolor{belugapink}{RGB}{237,103,166}

\newcommand\delimone{{\color{belugapurple}\emacsfont\bfseries LF}\color{belugagreen}\aftergroup:}
\newcommand\delimtwo{{\color{belugapurple}\emacsfont\bfseries inductive}\color{belugagreen}\aftergroup:}
\newcommand\delimthree{{\color{belugapurple}\emacsfont\bfseries schema}\color{belugagreen}\aftergroup=}
\newcommand\delimfour{{\color{belugapurple}\emacsfont\bfseries rec}\color{belugablue}\aftergroup:}
\newcommand\delimfive{{/}\slshape\aftergroup/}
\newcommand\delimaux{{\color{belugapink}.}}
\newcommand\delimsix{\color{belugapink}--\aftergroup\delimaux}

\lstdefinelanguage{Beluga}
{
	morekeywords=[1]{mlam,fn,case,of,total,in,type,impossible,let,and,ctype},
	keywordstyle=[1]\color{belugapurple}\emacsfont\bfseries,
	morecomment=[l]{\%},
	morecomment=[s]{\%\{}{\}\%},
	commentstyle=\color{belugared},
	sensitive=true,
}

\lstdefinestyle{belugastyle}{
	language=Beluga,
	basicstyle=\ttfamily\small,
	columns=flexible,
	keepspaces=true,
	showstringspaces=false,
	breaklines=true,
	breakatwhitespace=true,
	rangeprefix=\%\%\ ,
	rangesuffix=\ \%\%,
	includerangemarker=false,
	moredelim=**[is][\delimone]{LF}{:},
	moredelim=**[is][\delimtwo]{inductive}{:},
	moredelim=**[is][\delimthree]{schema}{=},
	moredelim=**[is][\delimfour]{rec}{:},
	moredelim=**[is][\delimfive]{/}{/},
	moredelim=**[is][\delimsix]{--}{.},
	literate={→}{{$\rightarrow$}}{1}
	{⊢}{{$\vdash$}}{1}
	{⇒}{{$\Rightarrow$}}{1}
}

\lstdefinestyle{belugastyleframes}{
	language=Beluga,
	basicstyle=\ttfamily\small,
	columns=flexible,
	keepspaces=true,
	showstringspaces=false,
	breaklines=true,
	breakatwhitespace=true,
	rangeprefix=\%\%\ ,
	rangesuffix=\ \%\%,
	includerangemarker=false,
	moredelim=**[is][\delimone]{LF}{:},
	moredelim=**[is][\delimtwo]{inductive}{:},
	moredelim=**[is][\delimthree]{schema}{=},
	moredelim=**[is][\delimfour]{rec}{:},
	moredelim=**[is][\delimfive]{/}{/},
	moredelim=**[is][\delimsix]{--}{.},
	frame=tblr,
	captionpos=b,
	literate={→}{{$\rightarrow$}}{1}
	{⊢}{{$\vdash$}}{1}
	{⇒}{{$\Rightarrow$}}{1}
}

\newcommand{\bellogo}{%
	\includegraphics[height=2ex, valign=m]%
		{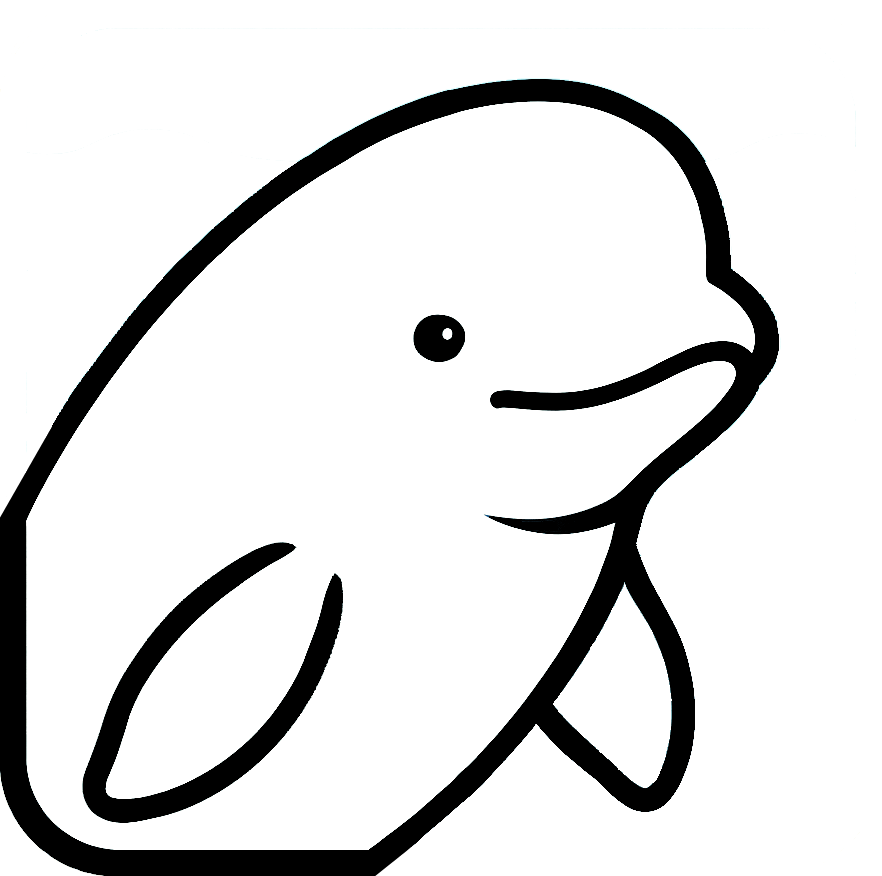}%
}

\newcommand{\repourl}{https://github.com/CinRC/Formalising-Independence-and-Causality-in-Reversible-Concurrent-Calculi}
\newcommand{\latestcommit}{805756a}%

\makeatletter
	\providecommand*{\hrefurl}{\hyper@normalise\hrefurl@}
	\providecommand*{\hrefurl@}[2]{\hyper@linkurl{#2}{#1}}
\makeatother

\ExplSyntaxOn
\NewDocumentCommand{\bel}{s m o o o} 
{%
	\IfValueTF{#4}%
	{%
		\IfValueTF{#3}%
		{%
			\IfBooleanTF#1%
			{%
				\hrefurl{\repourl/blob/\latestcommit/#2\#L#4-L#5}{\emacsfont{#3}\,\bellogo}%
			}%
			{%
				\hrefurl{\repourl/blob/\latestcommit/#2\#L#4-L#5}{#3\,\bellogo}%
			}
		}%
		{%
			\hrefurl{\repourl/blob/\latestcommit/#2\#L#4-L#5}{\bellogo}%
		}
	}%
	{%
		\IfValueTF{#3}%
		{%
						\IfBooleanTF#1%
			{%
				\hrefurl{\repourl/blob/\latestcommit/#2}{\emacsfont{#3}\,\bellogo}%
			}%
			{%
				\hrefurl{\repourl/blob/\latestcommit/#2}{#3\,\bellogo}%
			}		
		}%
		{%
			\hrefurl{\repourl/blob/\latestcommit/#2}{\bellogo}%
		}%
	}
}

\NewDocumentCommand{\tbel}{m o o}{%
	\IfValueTF{#2}%
	{%
		\bel*{#1}[#1][#2][#3]
	}
	{%
		\bel*{#1}[#1]
	}
} 
\ExplSyntaxOff

\DeclareSIUnit[quantity-product = {}]\kloc{\text{k LOC}}
\newcommand{\bkloc}[1]{\qty{#1}{\kloc}} %

%% file: latex/macro-symbols.tex
\newcommand{\axiom}[1]{{\color{black}\bfseries#1}}
\newcommand{\SP}{\hyperref[ax:sp]{\axiom{SP}}\xspace}
\newcommand{\BTI}{\hyperref[ax:bti]{\axiom{BTI}}\xspace}
\newcommand{\WF}{\hyperref[ax:wf]{\axiom{WF}}\xspace}
\newcommand{\PCI}{\hyperref[ax:pci]{\axiom{PCI}}\xspace}
\newcommand{\ID}{\hyperref[ax:id]{\axiom{ID}}\xspace}
\newcommand{\IRE}{\hyperref[ax:ire]{\axiom{IRE}}\xspace}
\newcommand{\CIRE}{\hyperref[ax:cire]{\axiom{CIRE}}\xspace}
\newcommand{\RPI}{\hyperref[ax:rpi]{\axiom{RPI}}\xspace}
\newcommand{\NRE}{\hyperref[def-NRE]{\axiom{NRE}}\xspace}
\newcommand{\BLD}{\hyperref[def-BLD]{\axiom{BLD}}\xspace}
\newcommand{\ED}{\hyperref[def-ED]{\axiom{ED}}\xspace}
\newcommand{\FLD}{\hyperref[def-FLD]{\axiom{FLD}}\xspace}
\newcommand{\PL}{\hyperref[ax:pl]{\axiom{PL}}\xspace}
\newcommand{\UT}{\hyperref[def-UT]{\axiom{UT}}\xspace}

\newcommand{\procset}{\ensuremath{\mathbb{P}}}  %
\newcommand{\kprocset}{\ensuremath{\mathbb{X}}} %

\newcommand{\names}{\ensuremath{\mathsf{N}}}
\newcommand{\labelset}{\ensuremath{\mathsf{L}}}
\newcommand{\keyset}{\ensuremath{\mathsf{K}}}

\newcommand{\proofset}{\ensuremath{\mathsf{P}}}

\newcommand{\Proc}{\mathsf{Proc}} %
\newcommand{\Lab}{\mathsf{Lab}}
\newcommand{\union}{\cup}
\newcommand{\inter}{\cap}

\newcommand{\Par}{\mid}

\newcommand{\prov}[1]{{#1}^\dagger}
\newcommand{\base}[1]{{#1}^\circ}

\newcommand{\rev}[1]{\ushortw{#1}} %
\newcommand{\Rev}[1]{\ensuremath{\rev{#1}}} %
\newcommand{\tRev}[1]{\Rev{\text{#1}}}

\newcommand{\myol}[2][3]{{}\mkern#1mu\overline{\mkern-#1mu#2}} %

\DeclareSymbolFont{symbolsC}{U}{txsyc}{m}{n}
\SetSymbolFont{symbolsC}{bold}{U}{txsyc}{bx}{n}
\DeclareFontSubstitution{U}{txsyc}{m}{n}
\DeclareMathSymbol{\opentimes}{\mathrel}{symbolsC}{93}

\newcommand{\equ}{\sim}							%
\newcommand{\nequ}{\not\sim}					%
\newcommand{\kequ}{\equ_{\keyset}}				%
\newcommand{\lequ}{\equ_{\Lab}}					%
\newcommand{\ind}{\mathrel{\iota}}				%
\newcommand{\nind}{\mathrel{\myol[1]{\iota}}}	%
\newcommand{\lind}{\mathrel{\Vert}}				%
\newcommand{\cind}{\mathrel{\mathsf{ci}}}		%
\newcommand{\ncind}{\mathrel{\myol[1]{\mathsf{ci}}}}  %

\newcommand{\conc}{\mathrel{\Diamond}}			%
\newcommand{\cf}{\mathrel{\#}} 					%
\newcommand{\ncf}{\mathrel{\myol[1]{\#}}}		%
\newcommand{\conn}{\mathrel{\curlyvee}} 		%
\newcommand{\dep}{\mathrel{\varkappa%
}}												%
\newcommand{\tdep}{\mathrel{\varkappa^*}}		%
\newcommand{\ldep}{\opentimes}					%
\newcommand{\cfi}{\mathrel{\cf_i}} 				%
\newcommand{\ncfi}{\mathrel{\ncf_i}} 				%
\newcommand{\cfg}{\mathrel{\cf_g}}              %

\newcommand{\eveqt}{\equ}
\newcommand{\sqeqt}{\equ}

\newcommand{\Cf}{\cf}
\newcommand{\Cfi}{\cfi}

\newcommand{\Cfg}{\cfg}

\newcommand{\sdep}{\ldep}
\newcommand{\coind}{\cind}
\newcommand{\ci}{\cind}
\newcommand{\co}{\conc}

\newcommand{\evkeqt}{\kequ}
\newcommand{\evleqt}{\lequ}

\newcommand{\cte}{\sharp} %

\newcommand{\es}{\varepsilon}

\newcommand{\nil}{\mathbf{0}}

\newcommand{\evtkof}[2]{\evop_{#1}(#2)}

\newcommand{\len}[1]{|#1|}

\newcommand{\srcof}[1]{\mathsf{src}(#1)}
\newcommand{\tgtof}[1]{\mathsf{tgt}(#1)}
\newcommand{\lblof}[1]{\mathsf{lbl}(#1)}
\newcommand{\ip}{\prec}

\makeatletter
\DeclareFontEncoding{LS2}{}{\@noaccents}
\makeatother
\DeclareFontSubstitution{LS2}{stix}{m}{n}
\DeclareSymbolFont{largesymbolsstix}{LS2}{stixex}{m}{n}
\DeclareMathDelimiter{\lBrace}{\mathopen} {largesymbolsstix}{"E8}{largesymbolsstix}{"0E}
\DeclareMathDelimiter{\rBrace}{\mathclose}{largesymbolsstix}{"E9}{largesymbolsstix}{"0F}

\DeclareFontEncoding{LS1}{}{}
\DeclareFontSubstitution{LS1}{stix}{m}{n}
\DeclareSymbolFont{stixsymbols}{LS1}{stixscr}{m}{n}
\SetSymbolFont{stixsymbols}{bold}{LS1}{stixscr}{b}{n}

\DeclareMathOperator{\ordop}{ord}

\newcommand{\card}[1]{\ensuremath{|#1|}}

\newcommand{\orig}[1]{O_{#1}} %
\newcommand{\bs}{\backslash}

\newcommand{\Left}{\mathrm{L}}        %
\newcommand{\Right}{\mathrm{R}}       %
\newcommand{\Dir}{\mathrm{d}}         %
\newcommand{\OpDir}[1]{\overline{#1}} %

\renewcommand{\L}{\Left}              
\newcommand{\R}{\Right}
\newcommand{\D}{\Dir}
\newcommand{\OD}{\OpDir{\D}}

\newcommand{\lmidr}{{\mid_{\R}}}
\newcommand{\lmidl}{{\mid_{\L}}}
\newcommand{\lmidd}{{\mid_{\D}}}
\newcommand{\lmidod}{{\mid_{\OD}}}

\newcommand{\lplusr}{{+_{\R}}}
\newcommand{\lplusl}{{+_{\L}}}
\newcommand{\lplusd}{{+_{\D}}}
\newcommand{\lplusod}{{+_{\OD}}}

\ushortCreate()[.05em](1.2){ushortwt} %

\newcommand{\rlmidr}{{\ushortwt{\mid_{\R}}}}
\newcommand{\rlmidl}{{\ushortwt{\mid_{\L}}}}

\newcommand{\rlplusr}{{\Rev{+_{\R}}}}
\newcommand{\rlplusl}{{\Rev{+_{\L}}}}

\newcommand{\bpair}[2]{\langle #1 , #2 \rangle} %
\newcommand{\cpair}[2]{\bpair{\lmidl #1}{\lmidr #2}} %

\newcommand{\kplabelset}{\labelset_{\keyset}^{\proofset}}
\newcommand{\klabelset}{\labelset_{\keyset}}

\newcommand{\out}[1]{\overline{#1}}

\newcommand{\setst}{ \mid }

\newcommand{\ccs}{\ensuremath{\mathsf{CCS}}\xspace}
\newcommand{\ccsk}{\ensuremath{\mathsf{CCS}{\keyset}}\xspace}

\newcommand{\pccs}{\ensuremath{\mathsf{CCS}^{\proofset}}\xspace}
\newcommand{\pccsk}{\ensuremath{\mathsf{CCS}{\keyset}^{\proofset}}\xspace}

\DeclareMathOperator{\keysop}{\mathsf{keys}}
\DeclareMathOperator{\keyop}{\mathsf{key}}
\DeclareMathOperator{\evop}{\mathsf{ev}}

\DeclareMathOperator{\lablop}{\ell} %
\DeclareMathSymbol{\kayop}{\mathalpha}{stixsymbols}{"6B} %
\DeclareMathSymbol{\ekayop}{\mathalpha}{stixsymbols}{"62} %

\newcommand{\kay}[1]{\kayop(#1)}

\newcommand{\ekay}[1]{\labl{#1}[\kay{#1}]}

\newcommand{\key}[1]{\keyop(#1)}
\newcommand{\keys}[1]{\keysop(#1)}
\newcommand{\ev}[1]{\evop(#1)}
\newcommand{\ord}[1]{\ordop(#1)}
\newcommand{\labl}[1]{\lablop(#1)}

%% file: latex/macro.tex
\definecolor{ind} {RGB}{43,131,186}
\colorlet  {conc}{ind!30}
\definecolor{sdep}{RGB}{215,25,28}
\colorlet  {tsdep}{sdep!20}
\definecolor{conf}{RGB}{171,221,164}
\definecolor{ord} {RGB}{253,174,97}

\newcommand{\nhphantom}[1]{\sbox0{#1}\hspace{-\the\wd0}}

%% file: latex/metadata.tex
\title[Concurrency, Causality and Conflict via Independence in Reversible Calculi]{%
	Concurrency, Causality and Conflict
	via Independence %
	in Reversible Calculi%
}

\author{Clément Aubert}
\authornote{Supported by the \href{https://www.nsf.gov}{National Science Foundation} under Grant No.: \href{https://www.nsf.gov/awardsearch/showAward?AWD_ID=2242786}{2242786}.}
\email{caubert@augusta.edu}
\orcid{0000-0001-6346-3043}
\author{Gabriele Cecilia}
\authornotemark[1]
\email{gcecilia@augusta.edu}
\orcid{0009-0007-7797-5008}
\affiliation{%
	\institution{Augusta University}
	\city{Augusta}
	\state{Georgia}
	\country{United States of America}
}

\author{Iain Phillips}
\email{i.phillips@imperial.ac.uk}
\orcid{0000-0001-5013-5876}
\affiliation{%
	\institution{Imperial College London}
	\city{London}
	\country{England}
}

\author{Irek Ulidowski}
\authornote{Partial support of the AY2024 International PI Invitation Program, IAR Nagoya University.}
\email{ulidowski@agh.edu.pl}
\orcid{0000-0002-3834-2036}
\affiliation{%
	\institution{AGH University of Kraków}
	\city{Kraków}
	\country{Poland},
}
\affiliation{%
	\institution{University of Leicester}
	\city{Leicester}
	\country{England}
}

\renewcommand{\shortauthors}{Aubert et al.}

\begin{abstract}
	Among the different ways of approaching the semantics of process calculi, true-concurrency models stand out for their ability to highlight subtle interplays between events.
	At their heart lie three crucial relations: concurrency, causality and conflict.
	This paper shows that reversibility, when endowed with a notion of independence, provides a rich tooling to study and characterise these true-concurrency relations.
	First, we prove that systems admitting pre-reversibility (\ie that can be extended with an independence relation satisfying some basic axioms) have a unique notion of independence, events, concurrency, causality and conflict.
	We then analyse the relationship between independence and the true-concurrency relations, establishing novel independence-based characterisations of causality and conflict.
	Our second series of contributions revolves around two concrete process calculi and two syntactic notions defined on their transition labels, namely independence and dependence;
	we prove that they partition connected transitions and characterise elegantly concurrency on adjacent transitions.
	This part of our development was machine-checked using the proof assistant Beluga.
	Last, we study how the key mechanism commonly used in reversible process algebra can be used as a proxy to retrieve causality and core independence on past events.
	We conclude by discussing how our results extend beyond reversible systems. %

\end{abstract}

\ccsdesc[500]{Computing methodologies~Concurrent programming languages}
\ccsdesc[500]{Theory of computation~Process calculi}
\ccsdesc[500]{Theory of computation~Interactive proof systems}
\ccsdesc[500]{Theory of computation~Logic and verification}

\keywords{%
	Labelled Transition Systems,
	Concurrency,
	Reversibility,
	Dependence,
	Axiomatic Approach to Reversibility,
	Formalisation,
	Machine-Checked Proof,
	Beluga
}

\maketitle

%% file: sections/newintro.tex
\section{Introduction}

\subsection*{Background: Reversibility as a Novel Approach to True-Concurrency Relations}

	Concurrent systems, where multiple actions may execute simultaneously or interact with each other, are inherently complex: their behaviour is shaped not only by which events (\ie atomic actions) occur, but also by how they relate.
	Semantic models provide a useful framework for describing such executions and reasoning about their properties, while abstracting away from implementation details and providing a basis for formal specification and analysis.
	They have often adopted %
the interleaving style, where concurrent executions are represented as alternative sequences of events~\cite{TOA1992}---however, this perspective makes it harder to observe and study subtle interplays, or lack thereof, between events.
A refined approach, the \emph{true-concurrency semantics}, aims at capturing and making explicit, among other relations, the causal dependencies between events~\cite[Section 1.2]{Glabbeek2001}.

At the heart of %
true-concurrency models lie three fundamental relations between events: \emph{concurrency}, relating events that may occur together; \emph{causality}, describing dependency among events; and \emph{conflict}, expressing incompatible events.
These \enquote{true-concurrency relations} can either be given directly as part of the model or derived, sometimes one from another.
For example, prime event structures require causality as a primitive parameter~\cite[Definition 1.3.4]{Winskel1986}, while stable configuration structures derive the notion of causality from a universal quantification over all execution paths, and then derive concurrency from the absence of causality~\cite[Deﬁnition 5.6]{Glabbeek2001}.

Different approaches have aimed at capturing directly the true-concurrency relations in process calculi such as the Calculus of Communicating Systems (CCS~\cite{milner80lncs}), or more generally in labelled transition systems (LTSes).
In those widely studied formalisms for modelling and analysing concurrent systems~\cite{BAETEN2005131}, it was first observed that these relations can be defined operationally as properties quantified over execution paths~\cite{NPW81}: \eg events are in conflict if they cannot both appear in any computation~\cite[p.~412]{BoudolC88Rex}.
An alternative method uses \emph{proof labels}: in short, it enhances transition labels to define \emph{independence} or, dually, \emph{dependence}, that allows one to retrieve true-concurrency relations between events from local relations on transitions~\cite{BC88,BC94,DeganoGP03}.

The advent of reversibility, \ie the ability to undo previously executed computations, opened a series of original investigations on those questions.
It was observed that bisimulations intimately related to true-concurrency relations defined by semantics~\cite{Bednarczyk1991} or logical means~\cite{10.1007/3-540-48340-3_32} could be characterised by reversible calculi at the syntactical level~\cite{Aubert2020b,AubertPU26,BEM25}.
Those results followed the intuition that reversibility provides a natural environment to study causality~\cite{PU07a,Ulidowski2014}: indeed, the \emph{causal-consistent reversibility} approach~\cite{DK04,DBLP:conf/rc/GluckLMMPUV23} states that an event can be reversed \emph{only after all the events it caused have been undone}, grounding causality as a fundamental notion in reversible calculi.

Our contributions revolve around fleshing out concretely the intuition that reversible calculi are well-suited to studying true-concurrency relations.
We work \enquote{both ways to the middle} by first leveraging the axiomatic theory for reversible computation~\cite{LPU24} to provide characterisations of concurrency, causality and conflict in \emph{pre-reversible} labelled transition systems satisfying sometimes additional axioms.
We then instantiate the axiomatic approach to two reversible process calculi and refine our results by leveraging the structure made apparent by proof labels and keys.
We briefly introduce the mechanisms at play behind those process calculi in the following example, before listing our main contributions and detailing how the paper is structured.

\subsection*{Building Intuitions with an Example}
Consider actions $a$ and $b$, and the following \ccs processes: $a\Par b$, where the actions can be executed concurrently, and $a.b + b.a$, where $a$ executes first followed by $b$, or vice versa.
These processes are equivalent under interleaving semantics ($a$ and $b$ are executed in some order), but they are not under true-concurrency semantics: \(a\), \(b\) are concurrent in $a \Par b$ but not in $a.b + b.a$, where $a$ causes $b$ or $b$ causes $a$.
However, this cannot be retrieved from the execution traces alone, and so \emph{something} must be added.

The \emph{proof-label} approach used by \pccs~\cite{BC88,BC94,DeganoGP03} annotates transition labels with syntactic information identifying the sub-process responsible for the transition.
The execution of \(a\) followed by \(b\) in $a\Par b$ is written in this system as
\[a \Par b \pr{f}[\lmidl a] \nil \Par b \pr{f}[\lmidr b] \nil \Par \nil\]
where \(\lmidl\), \(\lmidr\) in the transition labels show that \(a\), \(b\) originate from the left and the right side of the $\Par$ parallel operator respectively---which is the condition to declare them \emph{independent}.
The corresponding execution in $a.b + b.a$ then becomes
\[a.b + b.a \pr{f}[\lplusl a] b \pr{f}[~b] \nil\]
and one can establish by complementarity that those two transitions are not independent~\cite{DeganoGP03}, differentiating them from the previous transitions labelled \(\lmidl a\), \(\lmidr b\).

This difference in causal structure between our two example processes can also be observed in the reversible setting: \ccsk~\cite{PU07}, a reversible version of \ccs, enables the representation of reversible systems~\cite{Aman2020} by marking each executed action with a unique \emph{communication key}.
This key is recorded both in the label of the transition and in the syntax of the processes.
An action can later be undone by removing its corresponding key, but only in a causal-consistent manner.
In our example, the unique keys \(k\) and \(l\) are attached to $a$ and $b$ respectively\footnote{The actual keys are irrelevant: what matters is that they are different in this case, and attached to the same label.}, obtaining: 
\begin{align*}
	a\Par b \pr{f}[a[k]] a[k]\Par b \pr{f}[b[l]]  a[k]\Par b[l] && \text{ and } && a.b +b.a \pr{f}[a[k]] a[k].b + b.a \pr{f}[b[l]] a[k].b[l]+ b.a\text{.}
\end{align*}
The resulting processes, $a[k]\Par b[l]$ and $a[k].b[l]+ b.a$, have different backward behaviours: while the former can reverse $a[k]$, the latter cannot, since $b[l]$ guards (\ie was caused by) $a[k]$ in the backward direction and needs to be undone first to preserve causal consistency.
These \enquote{undoing} steps (called \emph{backward transitions}) are represented by
\begin{align*}
	a[k]\Par b[l]\pr{b}[a[k]] a \Par b[l] && \text{ and } && a[k].b[l]+ b.a \pr{b}[b][l] a[k].b + b.a\text{.}
\end{align*}
Thus, the ability of the former process to reverse $a[k]$ and $b[l]$ in any order indicates that \(a\) and \(b\) are concurrent.
Conversely, the impossibility to reverse these actions in both orders in the latter process is evidence of causal ordering between them.
This hints at the ability for backward systems to capture true-concurrency relations natively.

However, both the proof-label-based and the reversible-based approaches have shortcomings.
Proof-label-based approaches typically \emph{define} true-concurrency relations using only independence or its complement, dependence.
This can lead to conflating conflict with absence of independence~\cite[p.~442]{BC88}, or to defining concurrency from absence of dependence~\cite[p.~312]{DeganoGP03}, but such definitions are only accurate in particular settings.
For instance, absence of independence coincides with conflict for coinitial transitions, but not for composable transitions, where it represents causality (\autoref{lem-adj non-ind vs ccc}).
Reversible-based approaches require to observe all transitions in both directions to detect causality and do not seem to have a native mechanism to detect conflicts.
Therefore, a natural idea is to inject proof labels into reversible concurrent systems~\cite{aubert2023c,BEM24,BEM25}, combining the approaches to obtain the best of both worlds.
This is precisely what \pccsk does, with the added benefit of recording information via keys in the processes themselves~\cite{APU25,aubert2023c}.

\subsection*{Contributions Overview: Capturing True-Concurrency Relations for Reversible Systems}

We now highlight some of the results from our efforts to capture concurrency, causality and conflict in reversible systems.
We begin our study in the general setting of the axiomatic theory for reversible computation, where the only assumption for a labelled transition system is that an independence relation is provided.
Our first contribution is to show that this relation, as well as the notion of events, is unique if some expected properties are required (\autoref{thm-uniqueness}). We additionally introduce a new label-based definition of events, that is proven to be equivalent to the original independence-based definition of events (\autoref{prop:eveqt evleqt}).

Next, we analyse under which conditions the true-concurrency relations can be represented in terms of independence or its complement. 
We prove that, if some basic reversibility properties are satisfied, the true-concurrency relations for adjacent (\ie coinitial, cofinal or composable) forward events partition independence and its complement (Lemmas~\ref{lem-coind ind adj} and \ref{lem-adj non-ind vs ccc}).
Furthermore, for non-adjacent events, we provide new existentially quantified characterisations of causality and conflict, in terms of \emph{immediate predecessor} and \emph{initial conflict} (\autoref{prop-chain ip} and \autoref{thm-conf is conf dep}); these lead to new, independence-based methods to check causality and conflict for transitions (Propositions~\ref{prop:leq chain} and \ref{prop:cf leq}).

We then turn our attention to the concrete LTSes of \pccsk and \ccsk and prove that they are in bijection (\autoref{lem-bijection}).
\pccsk provides a convenient way of defining independence and dependence from transition labels only (\autoref{def-ind-dep-pccsk}); consequently, we can establish their complementarity instead of having to assume it (\autoref{thm-complementarity}).
We further prove that this notion of independence fits the requirements of the axiomatic approach (Theorems~\ref{thm-axioms-hold-pccsk} and \ref{thm-axioms-hold-ccsk}) and instantiate the independence-based characterisation of true-concurrency relations from~\autoref{sec:CCC}.
Additionally, we show how the order on keys stored in a process can be used to 
reconstruct \enquote{statically} causal ordering and core independence (Theorems~\ref{thm-ordering-event-key} and \ref{thm-independent-event-key}). %

Last but not least, one of our highlighted contributions is the formalisation, in the proof assistant Beluga~\cite{PientkaD10}, of some of our development.
We present it briefly below, before concluding our introduction with an outline of the paper and a detailed account of the improvements over previous publications.

\subsection*{Formalisation and Paper Organisations}

\paragraph{Beluga Formalisation}
The formalisation covers the definitions, examples and proofs of Sections~\ref{sec:proved-lts}, \ref{sec:ind-complem}, \ref{ssec:true-conc-rel-ccskp} and \ref{ssec:proof-axioms-ccsk}: in particular, it includes the definitions of \pccsk and \ccsk, the bijection between their LTSes, the complementarity of their independence and dependence relations and the proof that the axiomatic theory holds for the two LTSes.
Machine-checking our proofs allowed to certify their correctness, identify and resolve issues, and sometimes prompted changes to the proof strategy adopted.
It also allowed us to streamline some tedious arguments and to omit low-level details from the presentation.
To the best of our knowledge, very few results about reversible calculi have been formalised~\cite{MalettoR24,PaoliniPR15,Nicolo2026}, and even fewer about \emph{concurrent} reversible systems~\cite{Davalos2026,Cec25}.

Beluga is a proof assistant designed to specify and reason about formal systems.
Its architecture is based on two levels: a data level, built on the logical framework LF~\cite{Harper93}, which is used to represent the syntax, inference rules and derivation trees of the object language (in our case, process algebras); and a computation level, which supports the analysis and manipulation of LF data through pattern matching and recursive functions.
Beluga employs higher-order abstract syntax (HOAS)~\cite{pfenning88pldi} to represent object-level binding constructs using binders of LF's meta-language, relieving the user from explicitly managing $\alpha$-equivalence and capture-avoiding substitutions. %
This makes Beluga particularly well-suited for encoding process calculi such as \ccsk or \pccsk, where restriction binds names. %
In contrast, the mechanisation of generic LTSes, as used by the axiomatic theory, does not involve names or binding constructs, and therefore does not benefit from HOAS.
Moreover, compared to other proof assistants, Beluga provides fewer facilities for conveniently reasoning about abstract mathematical structures, such as support for equivalence classes---needed to reason about properties of events.
For these reasons, the machine-checking of Sections~\ref{sec:axiomatic}, \ref{sec:CCC}, \ref{ssec:true-concu-relation-pccsk} and \ref{sec:key-based} was omitted.

Formalised definitions, lemmas, theorems, etc. are hyperlinked to the accompanying Beluga development, hosted on \href{\repourl}{GitHub}; links are marked by a \bellogo\ symbol.
\emph{Beluga's Breach} environments\footnote{Named after the act of leaping out of the water, also known as cresting, which allows one to observe cetaceans such as belugas or orcas--or, in our context, formalisation details--without having to go underwater.} %
 highlight the main formalisation choices and differences between the encoding and its mathematical development; a more detailed discussion, together with a paper-to-artifact table, is provided in the \hrefurl{\repourl/blob/\latestcommit/overview.md}{overview file}.
We additionally refer the reader to the \hrefurl{\repourl/blob/\latestcommit/README.md}{README file} for installation and execution instructions, as well as for a brief description of the content of each file.
Lines of code (denoted by LOC) are measured using \href{https://github.com/AlDanial/cloc}{the cloc program} and a \hrefurl{\repourl/blob/\latestcommit/beluga_cloc_config.txt}{custom configuration} file that excludes comments and blank lines; our development amounts to approximately \bkloc{9}, and the examples span close to an additional \bkloc{1}.
The formalisation is entirely the authors' work, with no use of generative AI tools in either its development or validation.

\paragraph{Structure of the paper}
We first recall the axiomatic theory of reversibility~\cite{LPU24}, and then present a novel label-based definition of events, enabled by our observation on the uniqueness of independence under mild conditions (\autoref{sec:axiomatic}).
\autoref{sec:CCC} presents relations (core independence, immediate predecessor, initial and global conflicts) and details how they characterise true-concurrency relations in pre-reversible LTSIs sometimes subject to additional requirements.
\autoref{sec:proved-lts} recalls how \ccsk is extended with proof keyed labels to give \pccsk, and how the two systems are in bijection; this is also the starting point of our machine-checked development, and insights on the formalisation choices are discussed.
We use proof keyed labels to define independence and dependence separately, and then prove their complementarity on connected transitions (\autoref{sec:ind-complem}).
After instantiating the axiomatic approach to \pccsk and \ccsk, \autoref{sec:true-conc-rel} \enquote{lifts} our relations from transitions to events and relates them to the axiomatic-based definitions of concurrency, causality and conflict.
\autoref{sec:key-based} finally shows how key-based properties capture the axiomatic ordering on events.

\paragraph{Changelog}

This paper is a revised and expanded version of a conference paper~\cite{APU25} and its technical report~\cite{APU24}.
It also integrates and strengthens the formal development that was presented as a workshop paper~\cite{Cec25}.
Furthermore: 

\begin{itemize}
\item %
	The label-based definition of events (\autoref{def-event-label-equivalence}) and proof of its equivalence (\autoref{prop:eveqt evleqt}) are novel.
	\item \autoref{sec:CCC} is entirely new and contains only unpublished results.
	\item The formal development of \pccsk~\cite{cecilia_2025_15660907} was significantly extended, from about \bkloc{2} to \bkloc{6},
	 and improved: in particular, an issue with name binding was resolved, as discussed in \autoref{dis:open-closed}.
	\item The bijection between \pccsk and \ccsk has been mentioned repeatedly~\cite{aubert2023c,APU25,APU24}, %
	but it is now entirely formalised and discussed precisely in \autoref{ssec-pccsk-ccsk-bijection}.
	\item All the examples in Sections~\ref{sec:proved-lts} and \ref{sec:ind-complem} have been formalised.
	\item The results of Sections~\ref{ssec:true-conc-rel-ccskp} and \ref{ssec:proof-axioms-ccsk} have been formalised and are now presented more clearly.
	\item Sections~\ref{ssec:true-concu-relation-pccsk} and \ref{sec:key-based} contain additional elements, results, and their complete proofs.
\end{itemize}

%% file: sections/axioms.tex
\section{The Axiomatic Approach and Uniqueness Results}%
\label{sec:axiomatic}
We recall the basics of the axiomatic approach~\cite{LPU24} in \autoref{ssec:reminders-axiom}, before presenting our uniqueness results and alternative definition of events in \autoref{ssec:uniqueness}.

\subsection{The Axiomatic Approach}
\label{ssec:reminders-axiom}

Once a model of computation (presented as a labelled transition system with independence~\cite{Sassone1996}) is reversed, for example a process calculus, it is challenging to prove that it satisfies desired properties such as
\emph{causal consistency}~\cite{DK04} or \emph{causal safety} and \emph{causal liveness}~\cite{LPU24}.
The axiomatic approach allows us to obtain such properties---among others---from simpler axioms.
We recall the basic axioms, and introduce
\emph{polychotomies} and the conditions under which these properties hold (\autoref{prop-poly}).
Note that all these axioms (and therefore all derived properties) will hold for both \pccsk (\autoref{thm-axioms-hold-pccsk}) and \ccsk (\autoref{thm-axioms-hold-ccsk}).

\begin{definition}[LTSI~\protect{\cite[Definitions 2.1--2.3]{LPU24}}]%
	\label{def-ltsi}
	Let $\Proc$ be a set of \emph{processes}, ranged over by $P,Q$, and $\Lab$ a set of \emph{labels}, ranged over by $a,b$.
	A \emph{combined LTS} is a forward LTS $(\Proc,\Lab,\r{f})$ together with a backward LTS $(\Proc,\Lab,\r{b})$ satisfying the Loop Lemma: $P \r{f}[a] Q$ iff $Q \r{b}[a] P$.
	To refer to transitions which may be either forward or backward, we introduce backward labels $\rev{\Lab} = \{\rev{a} : a \in \Lab\}$, and let $\alpha$, $\beta$ range over \emph{directed labels}, \ie members of the disjoint union $\Lab \union \rev{\Lab}$.
	Then $P \r{fb}[\alpha] Q$ denotes $P \r{f}[a] Q$ if $\alpha = a$ and $P \r{b}[a] Q$ if $\alpha = \rev{a}$.
	We let $\rev{\rev{a}} = a$ and define the \emph{inverse} of a transition $t: P \r{fb}[\alpha] Q$ to be 
	$\rev{t}: Q \r{fb}[\rev{\alpha}] P$.
	
	We say that $(\Proc,\Lab,\r{fb},\ind)$ is a \emph{labelled transition system with independence} (LTSI) if $(\Proc,\Lab,\r{fb})$ is a combined LTS and $\ind$
	is an irreflexive symmetric binary relation on transitions---the \emph{independence} relation.
\end{definition}

\begin{remark}
	\label{rem:figure-convention}
	In our figures we will indicate either the name of the transition (\(t\), \(u\), etc.) or the label (\(a\), \(b\), \(c\), etc.) on the middle of the arrow.
\end{remark}

\begin{definition}[Notions on transitions, paths]%
	\label{def-transitions-paths}
	A transition \(t: P \r{bf}[\alpha] Q\) has for \emph{source} \(\srcof{t} = P\), for \emph{target} \(\tgtof{t} = Q\), and for \emph{label} \(\lblof{t} = \alpha\).
	Transitions \(t\), \(u\) are \emph{coinitial} if \(\srcof{t} = \srcof{u}\), \emph{cofinal} if \(\tgtof{t} = \tgtof{u}\), \emph{composable} if \(\tgtof{t} = \srcof{u}\), and \emph{adjacent} if they are either coinitial, cofinal or composable (in either order).
	
	A \emph{path} is a sequence of transitions \(r = t_1 t_2 \cdots t_n\) such that \(t_i\) and \(t_{i+1}\), for \(1 \leqslant i < n\), are composable.
	Its source \(\srcof{r}\) is \(\srcof{t_1}\), its target \(\tgtof{r}\) is \(\tgtof{t_n}\), its length \(\len{r}\) is \(n\), and it is \emph{rooted} if $\srcof{r}$ cannot perform a backward transition.
	Given a path $r = t_1 \cdots t_n$, its inverse path $\rev{r}$ is $\rev{t_n} \cdots \rev{t_1}$, and we let $r$ and $s$ range over paths.
		
	A transition \(t\) \emph{is connected to} a transition \(u\) if there exists a (possibly empty) path $r$ \st \(\srcof{r} = \srcof{t}\) and \(\tgtof{r} = \tgtof{u}\) (\autoref{fig:connectedness}).
	Two transitions are \emph{connected} if one is connected to the other.
\end{definition}

\begin{figure}
	\input{figures/connected.tex}
	\caption{Connectedness of transitions $t$ and $u$ (\autoref{def-transitions-paths}) with %
		the wavy line representing a path $r$ of any length.}%
	\label{fig:connectedness}
	\Description[short description]{long description}%
\end{figure}
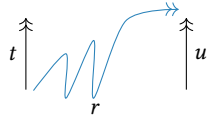

Intuitively, the coinitial transitions in a commuting square (\aka diamond) %
 are \emph{concurrent}
since they can happen in any order (or at the same time) before arriving at the same target process. %
Hence, existence of \emph{non-degenerate} commuting squares in an LTS, where all %
 states are distinct, 
indicates the presence of \emph{concurrency}.

\begin{definition}[Events, general definition~\protect{\cite[Definition 4.1]{LPU24}}]%
	\label{def-event-general}
	Consider a pre-reversible LTSI and let $\sqeqt$ be the smallest equivalence relation satisfying:
	if 
	\(P\), \(Q\), \(R\) and \(S\) form a commuting square with 
	$t:P \r{fb}[\alpha] Q$, $u:P \r{fb}[\beta] R$,
	$u':Q \r{fb}[\beta] S$, $t':R \r{fb}[\alpha] S$,
	and $t \ind u$, $\rev{u} \ind t'$, $\rev{t'} \ind \rev{u'}$, \(u' \ind \rev{t}\) and 
	\begin{itemize}
		\item \(Q \neq R\) if \(\alpha\) and \(\beta\) are both forward or both backward,
		\item \(P \neq S\) otherwise,
	\end{itemize}
	then $t \sqeqt t'$. %
	The equivalence classes of transitions, denoted by $[t]$, are the \emph{events}.
	We say that an event is \emph{forward} if it is the equivalence class of a forward transition; similarly for \emph{backward} events.
	Given an event \(e = [t]\) we let \(\rev{e} = [\rev{t}]\). We write $t\in e$ if $e=[t]$.
\end{definition}

Next, we present some of the axioms and properties of the axiomatic theory of reversible computation~\cite{LPU24}. 
Additional properties (\NRE, \BLD, \UT and \ED) are introduced when used, and all bold acronyms are hyperlinked to their definition.

\begin{figure}%
	\begin{tcolorbox}[title = {Basic Axioms}]
		\begin{tabularx}{\textwidth}{>{\bfseries}l X}
			SP & whenever $t:P \r{fb}[\alpha] Q$, $u:P \r{fb}[\beta] R$ with $t \ind u$, then there are cofinal transitions $u': Q \r{fb}[\beta] S$ and $t':R \r{fb}[\alpha] S$.\label{ax:sp}\\
			BTI & whenever $t:P \r{b}[a] Q$ and $u: P \r{b}[b] Q'$ and $t \neq u$, then $t \ind u$.\label{ax:bti}\\
			WF & there is no infinite reverse computation, \ie there are no $P_i$ (not necessarily distinct) such that $P_{i+1} \r{f}[a_i] P_i$ for all $i \in\mathbb{N}$.\label{ax:wf}\\
			PCI & whenever $t:P \r{fb}[\alpha] Q$, $u:P \r{fb}[\beta] R$, $u': Q \r{fb}[\beta] S$ and $t':R \r{fb}[\alpha] S$ with $t \ind u$, then $u' \ind \rev{t}$.\label{ax:pci}
		\end{tabularx}
	\end{tcolorbox}
	\begin{tcolorbox}[title = {Other Useful Properties}]
		\begin{tabularx}{\textwidth}{>{\bfseries}l X}
			ID & whenever $t:P \r{fb}[\alpha] Q$, $u:P \r{fb}[\beta] R$, $u': Q \r{fb}[\beta] S$ and $t':R \r{fb}[\alpha] S$, with $Q \neq R$ if $t$ and $u$ have the same direction; $P \neq S$ otherwise; then $t \ind u$.\label{ax:id}\\
			IRE & whenever $t \sqeqt t' \ind u$, then $t \ind u$.\label{ax:ire}\\
			CIRE & whenever $t \sqeqt t' \ind u' \sim u$ and $t',u'$ and $t,u$ are coinitial, then $t \ind u$.\label{ax:cire}\\
			RPI & whenever $t \ind t'$, then $\rev{t} \ind t'$.\label{ax:rpi}\\
            PL & for any path $r$ there are forward-only paths $s$, $s'$ such that $r$ and $\rev{s} s'$ (the \emph{parabolic path}) are coinitial and cofinal and $\len s + \len {s'} \leq \len r$.
            Moreover, for all \(t\) in \(\rev{s} s'\), there exists \(t'\) in \(r\) such that \(t \eveqt t'\). %
            \label{ax:pl}
		\end{tabularx}
	\end{tcolorbox}
	\caption{%
		Main properties studied by the axiomatic approach.
		In the tables above, an LTSI satisfies \textbf{Axiom} or \textbf{Property} if the condition on the right holds.%
	}\label{fig:axiomatic-properties}
	\Description[short description]{long description}%
\end{figure}
\begin{definition}[Axioms and pre-reversible LTSI~\protect{\cite{LPU24}}]
	\label{def-basic}
	\label{def-other-properties}
	\label{def-prerev}
	\autoref{fig:axiomatic-properties} presents
	\begin{description}
		\item[the basic axioms] \emph{Square property} (\SP), \emph{Backward transitions are independent} (\BTI), \emph{Well-founded} (\WF) and \emph{Propagation of coinitial independence} (\PCI);
		\item[other useful properties]
		\emph{Independence of diamonds} (\ID),
		\emph{Independence respects events} (\IRE), \emph{Coinitial independence respects events} (\CIRE),
		\emph{Reversing preserves independence} (\RPI) and \emph{Parabolic lemma} (\PL).
	\end{description}
	An LTSI is \emph{pre-reversible} if it satisfies the basic axioms.
\end{definition}
If two coinitial transitions \(t\), \(u\) are independent, then they must form a commuting square (\SP axiom).
If we have %
    coinitial backward %
    transitions, %
then they must be independent (\BTI axiom). %
\SP and \BTI are complementary: \SP expresses soundness of a definition of independence, while \BTI expresses its completeness.

Non-degenerate commuting squares have independent coinitial transitions (\ID property), and we must be able to propagate
independence from coinitial transitions in commuting squares to side transitions (\PCI axiom).
We shall also use an axiom that guarantees that there are no infinite backward computations (\WF axiom), and obtain that reversing 
a transition preserves its independence of other transitions (\RPI property).

The \IRE property gives that if \(t\), \(u\) are independent, then $u$ is also independent of all transitions $t'$ that belong to the same event as $t$ (\ie \(t \sqeqt t'\)). It is easy to see that \IRE implies \CIRE. 
However, the converse does not hold:

\begin{example}[\protect{\cite[Example~5.15]{LPU24}}]%
	\label{ex-CLG CSi}
    Consider the commuting square $t:P \r{f}[a]  Q$, $u:P \r{f}[ b] R$, $u':Q \r{f}[b] S$ and $t':R \r{f}[a] S$, with $t \ind u$, $\rev u \ind t'$, $\rev{t'} \ind \rev{u'}$, $u' \ind \rev t$, \ie we have independence at all corners of the square.
	This square has two forward events, labelled with $a$ and $b$ respectively. Axioms \SP, \BTI, \WF and \PCI hold.
	\CIRE holds because all coinitial transitions, namely $t$ and $u$, are independent.
	We have $t' \sqeqt t \ind u$ but $t' \ind u$ does not hold, so that \IRE fails.
\end{example}

Finally, \PL states that any path can be converted to a \emph{parabolic path}, \ie one composed of a backward-only path followed by a forward-only one.
Our modified version of the original statement~\cite[Proposition 3.4]{LPU24} additionally asserts that the new parabolic path %
 will not introduce new events compared to the old path. %
It is implied by two basic axioms:

\begin{lemma}%
    \label{lem-sp-bti-imply-pl}
    If an LTSI satisfies \SP and \BTI, then it also satisfies \PL.
\end{lemma}

\begin{proof}
    The proof %
    follows from axioms \SP and \BTI~\cite[Proposition 3.4]{LPU24}.
    We additionally observe that any new transitions in $s$ or $s'$ are obtained by replacing $t \rev u$ (for some forward transitions $t,u$) by $\rev {u'} t'$ where $t' \eveqt t$ and $u' \eveqt u$.
\end{proof}

\begin{definition}[Counting events in path~\protect{\cite[Definition 4.11]{LPU24}}]
	\label{def-count-events}
	Let $r$ be a path, $e$ be an event, and \(tr\) the path obtained by prepending the transition \(t\) to \(r\). %
	We define $\cte(r,e)$ as follows:
	\begin{align*}
		\cte(\es,e) & = 0 &&& 
		\cte(tr,e) & =
		\begin{cases}
			\cte(r,e)+1 & \text{if } t \in e \\ %
			\cte(r,e) -1 & \text{if } t \in \rev e \\%
			\cte(r,e) & \text{otherwise} %
		\end{cases}
	\end{align*}
\end{definition}

This work will consider \emph{concurrency}, \emph{causality} and \emph{conflict} on events or transitions in addition to equality on events and equivalence of transitions.
We represent general concurrency between events in two different ways, namely \emph{core independence} and \emph{concurrency}, using the independence relation $\ind$; conversely, we define the \emph{causality} and \emph{conflict} relations via standard properties on rooted paths.

\begin{definition}[Relations on events, transitions~\protect{\cite[Definitions 4.14, 4.23, 4.27]{LPU24}}]%
	\label{def-event relations}
	Two events $e_1$, $e_2$ are 
	\begin{itemize}
	\item \emph{core independent}\footnote{Originally called \enquote{coinitially independent}~\cite{LPU24}, but this naming is confusing 
	when $\coind$ is extended to (not necessarily coinitial) transitions.}, 
	$e_1 \coind e_2$, iff there are coinitial  \(t_1\), \(t_2\) \st $t_1\in e_1$, $t_2 \in e_2$ and $t_1 \ind t_2$;
	\item  \emph{concurrent}, $e_1\co e_2$, iff  $e_1 \coind e_2$ and  for all coinitial  \(t_1\), \(t_2\),  if $t_1\in e_1$, $t_2\in e_2$  then $t_1 \ind t_2$.
	\end{itemize}
	Two forward events $e_1$, $e_2$ are 
	\begin{itemize}
		\item \emph{causally related} (or \emph{causally ordered}), $e_1 \leq e_2$, iff for all rooted paths $r$, if $\cte(r,e_2) > 0$ then $\cte(r,e_1) > 0$;
		\item in \emph{conflict}, $e_1 \Cf e_2$, iff there is no rooted path $r$ \st $\cte(r,e_1) >0$ and $\cte(r,e_2) > 0$.
	\end{itemize}
	We write $e_1 < e_2$ ($e_1$ is a \emph{cause} of $e_2$), if $e_1 \leq e_2$ and $e_1 \neq e_2$.
	We also extend those relations to transitions by letting $t_1 \coind t_2$ iff $[t_1] \coind [t_2]$,
	\(t_1 \co t_2\) iff \([t_1] \co [t_2]\), and similarly for forward transitions for causal ordering and conflict.
\end{definition}

\begin{figure}
	\input{figures/relation-example.tex}
	\caption{The LTSI used in \autoref{ex-true-conc-rel} to illustrate the relations of \autoref{def-event relations}.}%
	\label{fig:relation-example}
	\Description[short description]{long description}%
\end{figure}

\begin{example}%
	\label{ex-true-conc-rel}%
	Consider the pre-reversible LTSI in \autoref{fig:relation-example}, %
	 where all pairs of transitions labelled \(a\) and \(b\) are independent.

	 Notice that there are three forward events labelled \(a\), \(b\) and \(c\), and that, for convenience, only one transition labelled \(b\) is shown as reversed, although any transitions could be reversed. 
	According to \autoref{def-event relations}, the events labelled \(a\) and \(b\) are core independent and concurrent, while those labelled \(a\) and \(c\) are causally related, and those labelled \(b\) and \(c\) are in conflict. Similarly for transitions: \eg the bottommost transition labelled \(a\) is a cause of the transition labelled \(c\). This follows the intuition that, once this transition \(a\) is executed, one can execute a path (in this case, performing \(b\) backwards) where \(a\) is not undone and that enables the execution of \(c\).
\end{example}

The relation \(\leq\) is a partial ordering of events, \ie a reflexive, antisymmetric and transitive relation. Consequently, \(<\) is a strict partial ordering. 

\begin{lemma}[\protect{\cite[Lemma 4.24]{LPU24}}]%
	\label{lem-partial-ordering}
If an LTSI is pre-reversible, then \(\leq \) is a partial ordering of events.
\end{lemma}

Core independence is for now the only relation relying on an existentially quantified property.
In \autoref{sec:CCC}, we shall impose mild conditions on pre-reversible LTSIs 
to obtain %
alternative definitions of causality and conflict (\autoref{prop-chain ip} and \autoref{thm-conf is conf dep}), so that all three relations are represented in terms of existentially quantified properties.

Core independence between two events $[t]$ and $[u]$ is defined via independence between coinitial transitions equivalent to $t, u$, potentially allowing $t$ and $u$ themselves not to be independent.
Concurrency strengthens this notion by additionally requiring that all coinitial transitions equivalent to $t$ and $u$ are independent.
It follows immediately (\eg by the definitions of $\co$ and $\coind$) that, in pre-reversible LTSIs, concurrency implies core independence:

\begin{lemma}%
	\label{lem-con implies coind}
Let $e_1$, $e_2$ be events in a pre-reversible LTSI. If $e_1\co e_2$ then $e_1\coind e_2$.
\end{lemma}

However, the converse does not hold in general: \autoref{ex-resolvable conflict} below provides a counterexample.
Core independence implies concurrency only under an additional requirement: for every independent coinitial pair of transitions $t_1$ and $t_2$, all coinitial transitions $t'_1 \in [t_1]$ and $t'_2 \in [t_2]$ must be independent (\ie \CIRE holds).
\autoref{prop-IRE co coind} establishes the result under the assumption of \CIRE.

\begin{notation}
	We write that \emph{an event \(e\) has label \(a\)} if \(e = [t]\) and \(\lblof{t} = a\)%
	~\cite[p.~11]{LPU24}.
	Henceforth, we shall also drop \enquote{labelled} for transitions and events, if the labelling is unambiguous, and write, \eg \enquote{event \(a\)} and \enquote{transition \(b\)} for, respectively, \enquote{the event (which is the equivalence class of transitions) labelled \(a\)} and \enquote{the transition labelled \(b\)}.
\end{notation}

\begin{example}%
	\label{ex-resolvable conflict}
	Consider the LTSI in \autoref{fig:resolvable conflict}~\cite[Example 5]{Glabbeek2009}\cite[Figure 9]{LPU24}.
	It contains three pairs of transitions labelled \(b\) and \(c\), highlighted in the figure.
	The first two pairs contain coinitial transitions, and the last one contains cofinal transitions.
	Independence is given by closing under \BTI and \PCI, which produces, among others, independence between the second pair of transitions \(b\) and \(c\). 
	Clearly \WF and \SP hold; hence the LTSI is pre-reversible.
	
	There are three events $a$, $b$ and $c$.
	The first pair of transitions $b$ and $c$ are \emph{not} independent (otherwise \SP would require to have cofinal transitions labelled \(b\) and \(c\)), and thus the events $b$ and $c$ are not concurrent.
	However, the events $b$ and $c$ are core independent thanks to the independence of the second pair of transitions \(b\) and \(c\).
	Consequently, since \(b \ci c\) and yet \(b \co c\) does not hold, we exhibited a counterexample to the converse of \autoref{lem-con implies coind}.
	Moreover, one can observe that \CIRE does not hold in this LTSI.
\end{example}

\begin{proposition}\label{prop-IRE co coind}
Let \(e_1\), \(e_2\) be events in a pre-reversible LTSI that satisfies \CIRE. Then $e_1\co e_2$ iff $e_1\coind e_2$.
\end{proposition}

\begin{proof}
($\Rightarrow$) By \autoref{lem-con implies coind}. 
	
($\Leftarrow$) Assume $e_1\coind e_2$, so that there are coinitial $t_1\in e_1, t_2\in e_2$ \st  $t_1\ind t_2$. \CIRE ensures that for all $t_1'\in e_1$, $t_2'\in e_2$ %
if \(t_1'\), \(t_2'\) are coinitial then $t_1'\ind t_2'$. Hence, $e_1\co e_2$.
\end{proof}

\CIRE is a reasonable assumption since \IRE, which implies \CIRE, holds in all case studies discussed in~\cite{LPU24}.

\begin{figure}%
	\input{figures/example-resolvable-conflict}
	\caption{%
		The LTSI used in \autoref{ex-resolvable conflict}, with the three pairs of transitions labelled \(b\), \(c\) highlighted. %
	}\label{fig:resolvable conflict}
   	\Description[short description]{long description}%
\end{figure}
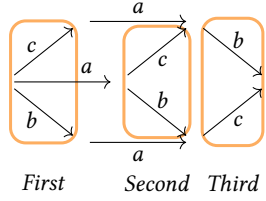

In pre-reversible LTSIs events can occur multiple times in a path, but the counting function (\autoref{def-count-events}) decrements the counter of a forward event \(e\) when its backward counterpart \(\rev{e}\) is met, so that we have:

\begin{definition}[No repeated events (\axiom{NRE})~\protect{\cite[Definition 4.18]{LPU24}}]%
	\label{def-NRE}
	For any rooted path $r$ and any forward event $e$ we have $\cte(r,e) \leq 1$.
\end{definition}

\begin{lemma}[\protect{\cite[Proposition 4.21]{LPU24}}]%
	\label{lem-NRE}
	If an LTSI is pre-reversible then it satisfies \NRE.
\end{lemma}

Core independence, and not concurrency, is the complement of the equality, causality and conflict %
on events: 

\begin{proposition}[Polychotomy for events~\protect{\cite[Definition 4.28, Proposition 4.29]{LPU24}}]%
	\label{prop-poly}%
	Pre-reversible LTSIs satisfy \emph{polychotomy for forward events}, \ie for all forward events \(e_1\), \(e_2\), exactly one of the following holds: %
		\begin{multicols}{6}
		\noindent
		\begin{itemize}[label={}]
			\item $e_1 = e_2$;
			\item $e_1 < e_2$;
			\item $e_2 < e_1$;
			\item $e_1 \Cf e_2$;
			\item or
			\item $e_1 \coind e_2$.
		\end{itemize}
	\end{multicols}
\end{proposition}

A consequence of \autoref{prop-IRE co coind} is that polychotomy holds in a pre-reversible LTSI that satisfies \CIRE, with $\co$ replacing $\coind$.  Since $<$, $\Cf$ and $\coind$ on transitions are closed under $\sim$, polychotomy for forward transitions holds: 

\begin{corollary}[Polychotomy for transitions]%
	\label{prop-poly-trans}
	Pre-reversible LTSIs satisfy \emph{polychotomy for forward transitions}, \ie for all forward transitions $t_1$, $t_2$, exactly one of the following holds: 
	\begin{multicols}{6}
        \noindent
		\begin{itemize}[label={}]%
			\item $t_1\sim t_2$;
			\item $t_1< t_2$;
			\item $t_2< t_1$;
			\item $t_1 \Cf t_2$;
			\item or
			\item $t_1\coind t_2$.
		\end{itemize} %
	\end{multicols}
\end{corollary}

%% file: figures/connected.tex
\begin{tikzpicture}
	\node (P) {%
	};
	\node [above = 3em of P] (Q) {%
	};
	\draw[fb]  (P) -- node[left, pos=.5]{\(t\)} (Q);
	\node [right = 6em of P](R) {%
	};
	\node at (Q -| R) (S) {%
	};
	\draw[fb] (R) --  node[right, pos=.5]{\(u\)} (S);
	\draw [fb, color=ind, text=black] plot [smooth] coordinates {($(P)+(.1, .1)$)  ($.5*(Q)!.5!(S)$) ($.5*(P)!.5!(R)$) ($.65*(Q)!.65!(S)$) ($.65*(P)!.65!(R)$) ($(S) - (.8, .15)$) ($(S) - (.1, 0)$)} node[above = -1.5cm, xshift=-1.1cm]{\(r\)};		
	
\end{tikzpicture}

%% file: figures/relation-example.tex
\centering

\begin{tikzpicture}[
	x={(1, 0)},  %
	y={(0, .8)}, %
	baseline,
	anchor=base
	]
	
	\node (a) at (0, -1) {};
	\node (b) at (1, -2) {};
	\node (d) at (1.5, -1) {};
	\node (f) at (2.5, -2) {};
	\node (g) at (2.5, 0) {};

	\draw[f] (a) -- node[below, pos=.3]{\(b\)} (b);
	\draw[f] (d) -- node[above, pos=.3]{\(c\)} (g);
	\draw[f] (a) -- node[above, pos=.5]{\(a\)} (d);
	
	\draw[f] ($(d) + (.1, 0)$) -- node[below, pos=.3]{\(b\)} ($(f) + (.1, 0)$);
	
	\draw[b] ($(f)+(.2, .1)$) -- node[above, pos=.4]{\(b\)} ($(d)+(.2, .1)$);
	
	\draw[f] (b) -- node[below, pos=.5]{\(a\)} (f);
\end{tikzpicture}

%% file: figures/example-resolvable-conflict.tex
\begin{tikzpicture}[
	x={(1, 0)},  %
	y={(0, .8)}, %
	baseline,
	anchor=base
	]

	\node (a) at (0, -1) {};
	\node (b) at (1, -2) {};
	\node (c) at (1, 0) {};
	\node (d) at (1.5, -1) {};
	\node (e) at (3.5, -1) {};
	\node (f) at (2.5, -2) {};
	\node (g) at (2.5, 0) {};

	\begin{scope}[on background layer]
		\node[%
        fill=none,
        draw=ord,
        line width=0.4mm,
        rectangle,
        inner xsep = -5pt,
        inner ysep = -3pt,
        rounded corners=0.2cm,
        fit=(a) (c) (b),
        label={[yshift=-.9em]below:\emph{First}}
		] {};

		\node [%
        fill=none,
        draw=ord,
        line width=0.4mm,
        rectangle,
        inner xsep = -5pt,
        inner ysep = -5pt,
        rounded corners=0.2cm,
        fit=(d) (g) (f),
        label={[yshift=-1.1em]below:\emph{Second}}
		] {};

		\node [%
        fill=none,
        draw=ord,
        line width=0.4mm,
        rectangle,
		inner xsep = -6pt,
        inner ysep = -2pt,
		rounded corners=0.2cm,
        fit=(g) (f) (e),
        label={[yshift=-.8em]below:\emph{Third}},
		] {};
	\end{scope}
	
	\draw[f] (a) -- node[below, pos=.3]{\(b\)} (b);
	\draw[f] (a) -- node[above, pos=.3]{\(c\)} (c);
	\draw[f] (d) -- node[below, pos=.6]{\(c\)} (g);
	\draw[f] (g) -- node[above, pos=.6]{\(b\)} (e);
	\draw[f] (c) -- node[above, pos=.5]{\(a\)} (g);
	\draw[f] (a) -- node[above, pos=.75]{\(a\)} (d);
	\draw[f] (d) -- node[above, pos=.6]{\(b\)} (f);
	\draw[f] (b) -- node[below, pos=.5]{\(a\)} (f);
	\draw[f] (f) -- node[below, pos=.6]{\(c\)} (e);	
	
\end{tikzpicture}

%% file: sections/uniqueness.tex
\subsection{Uniqueness Results and Label-Based Definition of Events}%
\label{ssec:uniqueness}

This subsection proves that, under the mild condition of being pre-reversible (or just \enquote{admitting} pre-reversibility (\autoref{def-admit-prever})), LTSIs essentially accept only one independence relation (\autoref{prop-prerev-coinitial-unique}), %
which furthermore uniquely determines the notions of events, causal ordering, conflict and independence (\autoref{thm-uniqueness}).
We also introduce a definition of events that relies on labels only, and prove it coincides with the independence-based one (\autoref{prop:eveqt evleqt}).
We first recall a lemma and a proposition to facilitate the proof of the uniqueness results.

\begin{lemma}[\protect{\cite[Lemma 4.7]{LPU24}}]%
	\label{lem-non-degenerate}
	Suppose that an LTSI is pre-reversible.
	Then every commuting square $t:P \r{fb}[\alpha] Q$, $u:P \r{fb}[\beta] R$, $u': Q \r{fb}[\beta] S$ and $t': R \r{fb}[\alpha] S$ with $t \ind u$ is \emph{non-degenerate}, \ie $P$, $Q$, $R$, $S$ are distinct processes.
\end{lemma}

\begin{proposition}[\protect{\cite[Proposition 4.10]{LPU24}}]%
	\label{prop-ID}
	LTSIs satisfying \BTI and \PCI satisfy \ID.
\end{proposition}

\begin{proposition}[Uniqueness of coinitial independence]%
	\label{prop-prerev-coinitial-unique}
	If two pre-reversible LTSIs \((\Proc,\Lab,\r{fb},\ind_1)\) and $(\Proc,\Lab,\r{fb},\ind_2)$  have the same underlying combined LTS $(\Proc,\Lab,\r{fb})$, 
	then $\ind_1$ and $\ind_2$ agree on coinitial transitions. 
\end{proposition}

\begin{proof}
	Suppose that $t$ and $u$ are coinitial %
	with $t \ind_1 u$.
	By \SP for $\ind_1$ we get a square that is non-degenerate  %
	 by \autoref{lem-non-degenerate}.
	We deduce that $t \ind_2 u$ using \ID which holds by \autoref{prop-ID}.
	By symmetry we deduce the result.
\end{proof}

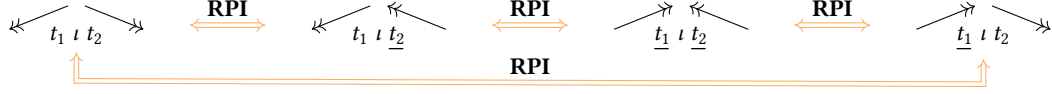
\begin{figure}
    \input{figures/rpi.tex}
    \caption{Mutually defined independences on adjacent transitions in pre-reversible LTSIs satisfying \RPI (\autoref{remark-on-adjacent-indep}).}\label{fig:rpi}
    \Description[short description]{long description}%
\end{figure}

\begin{remark}[\protect{\cite[Preamble]{aubert2023c}}]
	\label{remark-on-adjacent-indep}
    In the case of pre-reversible LTSIs satisfying \RPI, such as \pccsk and \ccsk (\autoref{sec:true-conc-rel}), \autoref{prop-prerev-coinitial-unique} extends to \emph{adjacent} transitions (\autoref{def-transitions-paths}), since independence of cofinal and composable transitions (in any direction) can be reduced to independence of coinitial transitions by reversing one or two transitions (\autoref{fig:rpi}).
\end{remark}

\begin{definition}%
    \label{def-admit-prever}
	A combined LTS $(\Proc,\Lab,\r{fb})$ \emph{admits pre-reversibility} if there exists an independence relation $\ind$ such that $(\Proc,\Lab,\r{fb},\ind)$ is a pre-reversible LTSI.
\end{definition}

\begin{theorem}[Uniquenesses]%
    \label{thm-uniqueness}
	If a combined LTS $(\Proc,\Lab,\r{fb})$ admits pre-reversibility, then the notions of events, core independence, causal ordering and conflict are uniquely determined. 
\end{theorem}

In other words, this theorem implies that the notion of events (\autoref{def-event-general}), which then determines the relations defined in \autoref{def-event relations},  does not depend on the choice of the independence relation. %

\begin{proof}
	By \autoref{prop-prerev-coinitial-unique} the independence relation $\ind$ is determined for coinitial transitions.
    It is easy to see that event equivalence $\eveqt$ between transitions is determined by $\ind$ on coinitial transitions.  Furthermore, it is clear from \autoref{def-event relations} that $\leq$, $\Cf$ and $\coind$ are then uniquely determined.
\end{proof}

We are not aware of any previous such uniqueness result. Its benefit is that, 
rather than requiring one to provide a definition for, \eg independence and then \enquote{manually} prove that it satisfies various properties, our result stipulates that those properties actually admit only one definition for most of the core notions.

\autoref{thm-uniqueness} suggests that event equivalence $\eveqt$ can be defined without relying on independence.
Such an example of \enquote{independence-free} definition of events~\cite[p.~125]{vGV97} has already been used in the context of reversibility~\cite[Definition 2.2]{PU07a}.
But we now give a new, alternative definition which does not use independence and is simpler than previous definitions, in that it relies on the existence of a single path.

\begin{definition}[Event label equivalence]%
	\label{def-event-label-equivalence}
	In a combined LTS $(\Proc,\Lab,\r{fb})$, two forward transitions with the same label $t_1:P_1 \r{f}[a] P'_1$ and $t_2:P_2 \r{f}[a] P'_2$ are \emph{event label equivalent} ($t_1 \evleqt t_2$) if there is a path $r: P'_1 \r{bf}^* P'_2$ such that $a$ does not occur in any transition of $r$.
	We extend to backward transitions by letting $t_1 \evleqt t_2$ iff $\rev{t_1} \evleqt \rev{t_2}$.
\end{definition}

Note that \autoref{def-event-label-equivalence} does not make any use of an independence relation on transitions.
It is clear that $\evleqt$ is an equivalence relation.
However, in order for it to be meaningful, we need the combined LTS to admit pre-reversibility, in which case event label equivalence is the same as event equivalence (\autoref{prop:eveqt evleqt}).
We recall two consequences of pre-reversibility, needed to prove this result:

\begin{definition}[Unique Transition (\axiom{UT})~\protect{\cite[Definition 3.11]{LPU24}}]%
	\label{def-UT}
	If $t : P \r{f}[a] Q$ and $u : P\r{f}[b] Q$  then $a = b$.
\end{definition}

\begin{lemma}[\protect{\cite[Corollary 3.12]{LPU24}}]%
	\label{lem-UT}
	If an LTSI is pre-reversible then it satisfies \UT.
\end{lemma}

\begin{definition}[Backward Label Determinism (\axiom{BLD})~\protect{\cite[Definition 4.5]{LPU24}}]%
	\label{def-BLD}
	If $t : P \r{b}[a] Q$ and $u : P\r{b}[a] R$ are coinitial
backward transitions with the same label then $t = u$.
\end{definition}

\begin{lemma}[\protect{\cite[Proposition 4.6]{LPU24}}]%
	\label{lem-BLD}
	If an LTSI is pre-reversible then it satisfies \BLD.
\end{lemma}

\begin{proposition}[Event equivalences coincide]%
	\label{prop:eveqt evleqt}
	In a pre-reversible LTSI, for all \(t_1\), \(t_2\) we have 
	$t_1 \eveqt t_2$ iff $t_1 \evleqt t_2$.
\end{proposition}

\begin{proof}
	($\Rightarrow$)
	Suppose that $t_1 \eveqt t_2$.
	We can suppose that $t_1$ and $t_2$ are both forward; if they are both backward we consider $\rev{t_1}$ and $\rev{t_2}$.
	We can use the general definition of events (\autoref{def-event-general}) and consider a single commuting square, with \(t_1\) and \(t_2\) on opposite sides, and with their targets connected by a transition \(t\) such that \(\rev{t_1} \ind t\) and $\rev{t} \ind \rev{t_2}$.
	Without loss of generality, suppose that $t$ is forward.  Then $\srcof t \neq \tgtof{t_2}$ by \autoref{def-event-general}.
	So $t_2$ and $t$ have different labels by \BLD (\autoref{lem-BLD}) and we can instantiate \autoref{def-event-label-equivalence} to obtain that $t_1 \evleqt t_2$.
	
	($\Leftarrow$)
	Suppose given %
	forward transitions $t_1$ and $t_2$ with $t_1 \evleqt t_2$: if \(t_1\) and \(t_2\) are backward then we reason similarly using \(\rev{t_1}\) and \(\rev{t_2}\).
	Since event label equivalent transitions have the same direction, this is without loss of generality.
		
	Let $t_1:P_1 \r{f}[a] P'_1$ and $t_2:P_2 \r{f}[a] P'_2$, we proceed by induction on the length of the path from $P'_1$ to $P'_2$.
	If the path is of length zero then \(t_1\) and \(t_2\) are cofinal and we have $\rev{t_1} = \rev{t_2}$ by \BLD and hence $t_1 \eveqt t_2$.
	
	If the path has non-zero length, by \PL we can convert the path into a parabolic path without increasing the length of the path considered and without introducing new labels.
	Hence it is still the case that $a$ does not occur in the parabolic path.
	By the parabolic property, either the first transition is backward or the last transition is forward.  We consider the case where the last transition is forward; the other case is similar.
	
	So suppose that the last transition in the path is $t:Q' \r{f}[b] P'_2$, where $b \neq a$.
	By \BTI we have $\rev t \ind \rev {t_2}$.  We then use \SP to complete a commuting square with transitions $t'_2: Q \r{f}[a] Q'$ and $t':Q \r{f}[b] P_2$ for some $Q$.
	We see that $Q' \neq P_2$ by \UT (\autoref{lem-UT}). 
	By \PCI $\rev{t'_2} \ind t$, $t'_2 \ind t'$ and $\rev {t'} \ind t_2$.  So $t'_2 \eveqt t_2$.
	By inductive hypothesis $t_1 \eveqt t'_2$; hence $t_1 \eveqt t_2$.
\end{proof}

We prefer not to regard \autoref{def-event-label-equivalence} as the primary definition of event equivalence, since it is not as clearly motivated from first principles as \autoref{def-event-general}.
However, we consider it more intuitive to manipulate, and will study in \autoref{sec:key-based} how label- and key-based definitions of events relate to causality and core independence.

%% file: figures/rpi.tex
        \begin{tikzpicture}[
    x={(1, 0)},%
    y={(0, 1)},%
    ]
    \node (a) at (0, 0){};
    \node (b) at (1, .4){};
    \node (c) at (2, 0){};
    
    \draw[fb] (b) -- (a);
    \draw[fb] (c) -- (b);
    \node (cap1) [below  = .1cm of b] {\(t_1 \ind \rev{t_2}\)};
    
    \node (a0) at (-2, 0){};
    \node (b0) at (-3, .4){};
    \node (c0) at (-4, 0){};
    
    \draw[fb] (b0) -- (a0);
    \draw[fb] (b0) -- (c0);
    \node (cap0) [below  = .1cm of b0] {\(t_1 \ind t_2\)};
    
    \draw [double equal sign distance, Implies-Implies, ord] ($(b)+(-1.5, -.3)$) to node[above]{\RPI} ($(b0)+(1.5, -.3)$);
    
    \node (a2) at (4, 0){};
    \node (b2) at (5, .4){};
    \node (c2) at (6, 0){};
    
    \draw[fb] (a2) --  (b2);
    \draw[fb] (c2) --  (b2);
    
    \node (cap2) [below  = .1cm of b2] {\(\rev{t_1} \ind \rev{t_2}\)};
    
    \draw [double equal sign distance, Implies-Implies, ord] ($(b)+(1.5, -.3)$) to node[above]{\RPI} ($(b2)+(-1.5, -.3)$);
    
    \node (a3) at (8, 0){};
    \node (b3) at (9, .4){};
    \node (c3) at (10, 0){};
    
    \draw[fb] (a3) -- (b3);
    \draw[fb] (b3) -- (c3);
    
    \node (cap3) [below  = .1cm of b3] {\(\rev{t_1} \ind t_2\)};
    
    \draw [double equal sign distance, Implies-Implies, ord] ($(b2)+(1.5, -.3)$) to node[above]{\RPI} ($(b3)+(-1.5, -.3)$);
    \draw [double equal sign distance, Implies-Implies, ord] (cap0) to ($(cap0)+(0, -.6)$) to node[above]{\text{\RPI}} ($(cap3)+(0, -.6)$) to (cap3);
\end{tikzpicture}

%% file: sections/ccc.tex
\section{Capturing True-Concurrency Relations for Pre-Reversible LTSIs}%
\label{sec:CCC}

We have seen in the previous section that, in pre-reversible LTSIs, the set of pairs of forward events admits the standard polychotomy %
into core independence, causality, conflict and equality (\autoref{prop-poly}). In this section we take a different perspective and aim to study how these relations interact with independence and its complement.
In particular, by means of suitable relations based on non-independence, we obtain new existentially quantified characterisations of causality and conflict. %
Moreover, in calculi like \pccsk where independence and its complement can be defined in terms of transition labels (\autoref{def-ind-dep-pccsk}), these results yield label-based characterisations of true-concurrency relations.

\autoref{subsec:adj-event} focuses on adjacent events and determines, in this setting, under which conditions the true-concurrency relations are captured by independence and its complement (\(\nind\)).
\autoref{subsec:causality} introduces an \emph{immediate predecessor} relation (\(\ip\)) and develops representations of non-independence and causal order in terms of immediate predecessor pairs.
Last, \autoref{subsec:conflict} characterises conflict in terms of \emph{initial conflict} (\(\cfi\)) between coinitial forward transitions, which is then propagated globally via chains of immediate predecessor pairs.
\autoref{fig:summary-sec-ccc} summarises the main results of this section.

\begin{figure}%
	\begin{tcolorbox}[title = {Summary of results}]
		\begin{tabular}{c l c l}
			\({\coind} = {\ind}\) & on adjacent events & if \IRE and \RPI hold & \autoref{lem-coind ind adj} \\
			\({<} = {\ip} \) & on composable forward events & & \autoref{lem-cau iff ip} %
			\\
			\({\ip} = {\nind}\) & on composable forward events & if \IRE and \RPI hold\footnote{If the reverse of the first event is non-independent with the second event, then only \IRE is required, by \autoref{lem-ip-CIRE}.%
			} & \autoref{lem-IRE RPI ip comp} %
			\\
			\({\Cf} = {\nind} = {\Cfi}\) & on coinitial forward events & if \CIRE and \ED hold  & \autoref{lem-Cfi iff conf}\\ 
			${<} = {\ip^+}$ & on connected forward events & %
			& \autoref{prop-chain ip} \\
			${\Cf} = {\Cfg}$ & on connected forward events & if \CIRE and \ED hold & \autoref{thm-conf is conf dep}
		\end{tabular}
	\end{tcolorbox}
	\Description[short description]{long description}%
	\caption{Summary of main results of \autoref{sec:CCC} for pre-reversible LTSIs. We recall that if \CIRE holds then
	${\coind} = {\co}$ for arbitrary events (\autoref{prop-IRE co coind}). Since \IRE implies \CIRE, we also have ${\coind} = {\ind} = {\co}$ on adjacent events when 
	\IRE and \RPI hold, by \autoref{lem-coind ind adj}.
}
	\label{fig:summary-sec-ccc}
\end{figure}

\subsection{Adjacent Events}%
\label{subsec:adj-event}

This section has two main goals: lifting the relations in \autoref{def-transitions-paths} from transitions to events, and then focusing on characterisations applicable to adjacent events (\ie events that are coinitial, cofinal or composable):

\begin{definition}[Relations on events]%
	\label{def-ind sdep event}
	Let \(e_1\), \(e_2\) be events in an LTSI.
	\begin{itemize}
		\item 
		\(e_1\), \(e_2\) are \emph{connected} (\resp \emph{coinitial}, \emph{cofinal}, \emph{composable}, \emph{adjacent}) if there exist transitions $t_1 \in e_1$ and $t_2 \in e_2$ such that \(t_1\), \(t_2\) are connected (\resp coinitial, cofinal, composable, adjacent);
		\item $e_1 \ind e_2$ if there exist transitions $t_1 \in e_1$ and $t_2 \in e_2$ such that $t_1 \ind t_2$%
            \footnote{If the LTSI satisfies \IRE, this is equivalent to requiring \emph{all} transitions in \(e_1\) and \(e_2\) to be independent.
                      We prove this for \pccsk with \autoref{prop-ind-dep-on-events}.}.
	\end{itemize}
\end{definition}

\begin{notation}[Complement of relations]%
	\label{notation-complement}
	We will in general denote the complement of a relation by putting a bar over it.
	We write \({\nind}\), \({\ncind}\) and \({\ncf}\) and  for the complement of \({\ind}\), \({\cind}\) and \({\cf}\): \eg $e_1 \nind e_2$ iff $e_1 \ind e_2$ does not hold.
	The complements of equivalence \({\equ}\), causal order \({<}\) and equality \({=}\) are the exceptions, denoted by \({\nequ}\), \({\not <}\) and \({\neq}\) as usual.
\end{notation}

Next, we identify the conditions under which independence represents core independence among adjacent events
and, orthogonally, the conditions under which the complement of independence represents either causality, conflict or equality on adjacent events. 

\begin{lemma}%
	\label{lem-coind ind adj}
	Let \(e_1\), \(e_2\) be events in a pre-reversible  LTSI.
    \begin{enumerate}
		\item If $e_1 \coind e_2$ then $e_1$, $e_2$ are adjacent and $e_1 \ind e_2$. \label{item:adj ind coind-1} \label{item-adj-ind-coind-1}
		\item If the LTSI satisfies \IRE and \RPI, $e_1,e_2$ are adjacent and $e_1 \ind e_2$  then $e_1 \coind e_2$. \label{item-adj-ind-coind-2}
    \end{enumerate}
\end{lemma}

\begin{proof}
    \begin{enumerate}
    \item 
    Immediate from Definitions~\ref{def-transitions-paths}, \ref{def-event relations} and \ref{def-ind sdep event}.
    \item
    Suppose that $t_1 \in e_1$ and $t_2 \in e_2$ are adjacent, and that $t'_1 \ind t'_2$ where $t'_1 \in e_1$, $t'_2 \in e_2$.
    By \IRE we have $t_1 \ind t_2$.
    Since \(t_1\), \(t_2\) are adjacent, we can reverse one or both transitions and corresponding events to obtain coinitial $t''_1 \in e''_1$, $t''_2 \in e''_2$. 
    Using \RPI we obtain $t''_1 \ind t''_2$\footnote{Note that \RPI is not needed if \(t_1\), \(t_2\) are already coinitial.}.
    Hence $e_1 \coind e_2$ as required.
    \qedhere
    \end{enumerate}
\end{proof}

A consequence of \autoref{lem-coind ind adj}~\itemref{item-adj-ind-coind-1} is that non-adjacent events cannot be core independent.
An interesting case of \autoref{lem-coind ind adj}~\itemref{item-adj-ind-coind-2} is for composable forward events, and we can see that both
\IRE and \RPI are necessary.
A consequence of \autoref{lem-coind ind adj}~\itemref{item-adj-ind-coind-2} is that
if adjacent events $e_1,e_2$ are \emph{not} core independent, 
implying $e_1=e_2$, $e_1<e_2$ or $e_1\Cf e_2$ by polychotomy (\autoref{prop-poly}),
then they are not independent, namely $e_1 \nind e_2$.  

Core independence, and thus concurrency (recall their equivalence under \CIRE, by \autoref{prop-IRE co coind}), coincides with independence on adjacent events in pre-reversible LTSIs with \IRE and \RPI.
We now investigate causal order, conflict and equality
on adjacent non-independent events, focusing on forward events.

\begin{lemma}%
    \label{lem-adj non-ind vs ccc}
    Let \(e_1\), \(e_2\) be forward non-independent events (\ie $e_1 \nind e_2$) in a pre-reversible LTSI. %
    \begin{enumerate}
        \item If \(e_1\), \(e_2\) are coinitial then $e_1=e_2$ or $e_1 \Cf e_2$; \label{lem-adj non-ind vs ccc one}
        \item if \(e_1\), \(e_2\) are composable then $e_1<e_2$; \label{lem-adj non-ind vs ccc two}
        \item if \(e_1\), \(e_2\) are cofinal then $e_1=e_2$. \label{lem-adj non-ind vs ccc three}
    \end{enumerate}
\end{lemma}

\begin{proof}
    \begin{enumerate}
    \item Assume $t_1 \in e_1, t_2\in e_2$ are coinitial with source $P$. Since \(e_1\), \(e_2\) are non-independent ($e_1 \nind e_2$)
    we have $t_1 \nind t_2$ for all $t_1 \in e_1$, $t_2\in e_2$.
    Hence, $e_1 \ncind e_2$.
    By \NRE (\autoref{lem-NRE}) no multiple transitions in $e_1$ can occur in any forward-only rooted path. %
    Assume $e_1 < e_2$ for a contradiction. 
    This means there is $t_1'\in e_1$ that appears in a rooted forward-only path to $P$, and hence before $t_1$, 
    contradicting \NRE. 
    We obtain $e_1 \not < e_2$, and $e_2 \not < e_1$ can be obtained similarly, so we conclude using polychotomy.
    
    \item Since \(e_1\), \(e_2\) are composable forward events in a pre-reversible LTSI then either $e_1\coind e_2$  or $e_1<e_2$. Assume
    $e_1\coind e_2$ for a contradiction, implying there are coinitial forward $t_1\in e_1$, $t_2\in e_2$ \st $t_1\ind t_2$. However,
    $e_1\nind e_2$ means no such \(t_1\), \(t_2\) exist. Hence,  $e_1<e_2$.
    
    \item %
        Assume $e_1\neq e_2$ for a contradiction.
	We have forward $t_1\in e_1$, $t_2\in e_2$ and \(t_1\), \(t_2\) are cofinal. By \BTI $\rev{t_1}\ind \rev{t_2}$.
	Then we use \SP to produce a commuting square with the sides $t_2' t_1$ and $t_1' t_2$, where $t_1'\in e_1$ and $t_2'\in e_2$ and $t_1', t_2'$ are coinitial.
	Applying \PCI twice we obtain $t_1'\ind t_2'$, contradicting $e_1 \nind e_2$. 
	Hence, $e_1=e_2$.
    \qedhere
    \end{enumerate}
\end{proof}

We have seen how, in pre-reversible LTSIs satisfying \IRE and \RPI, true-concurrency relations on adjacent forward events 
partition independent and non-independent events. 
There is also a tight correspondence between types of adjacency and true-concurrency relations.

\subsection{Causality}%
\label{subsec:causality}

In this section we begin by determining the conditions under which causality between composable forward transitions can be represented via non-independence. We then turn to forward events and introduce the notion of immediate predecessor; first focusing on composable forward events, we show that it coincides with causality and, with additional assumptions, with non-independence. Next, we generalise to arbitrary forward events and show how the immediate predecessor relation, and hence non-independence, can be used to characterise causality. Finally, we leverage this result to obtain a method for checking causality of transitions based on non-independence.

Given a pre-reversible LTSI, we aim to determine when causality between composable forward transitions can be characterised by non-independence. We first prove that, assuming \IRE, non-independence between composable forward transitions implies causality:

\begin{lemma}%
	\label{lem-not ind implies cau}
	Let \(t_1\), \(t_2\) be composable forward transitions in a pre-reversible LTSI that satisfies \IRE. %
	If $t_1\nind t_2$ then $[t_1]<[t_2]$.
\end{lemma}
\begin{proof}
Suppose $t_1\nind t_2$. Since \IRE holds, $[t_1] \nind [t_2]$: if $t_1' \ind t_2'$ for some $t_1' \in [t_1]$, $t_2' \in [t_2]$, by \IRE we would obtain $t_1\ind t_2$, and thus a contradiction. By applying \autoref{lem-adj non-ind vs ccc}~\itemref{lem-adj non-ind vs ccc two} we conclude that $[t_1]<[t_2]$.
\end{proof}

The converse of \autoref{lem-not ind implies cau} does not hold, as shown by the following example:

\begin{example}%
	\label{ex-RPI needed}
Consider $t: P \r{f}[a] Q$ and $u: Q \r{f}[b] S$ with $t\ind u$. This LTSI is pre-reversible and \IRE holds.
Clearly $t<u$  and $[t]<[u]$, so causality and independence can overlap in pre-reversible LTSIs with \IRE.
\end{example}

Assuming \RPI in addition to \IRE, the converse of \autoref{lem-not ind implies cau} holds and causality of composable forward transitions can be characterised by their non-independence.

\begin{lemma}\label{lem-not ind iff ip}
	Let \(t_1\), \(t_2\) be composable forward transitions in a pre-reversible LTSI that satisfies \IRE and \RPI. Then  $t_1\nind t_2$ iff $[t_1]<[t_2]$.
\end{lemma}

\begin{proof}
	$(\Rightarrow)$ Follows by \autoref{lem-not ind implies cau}. 
	
	$(\Leftarrow)$ Suppose $[t_1]<[t_2]$. Assume for a contradiction that $t_1\ind t_2$. Since \RPI holds we get $\rev{t_1}\ind t_2$, where $\rev{t_1}$ and $t_2$ are coinitial.
	By \SP there is a square with some $t_2'$ coinitial with $t_1$ such that \(t_2'\sim t_2\), and a rooted forward-only path to $t_2'$ that does not contain a transition in $[t_1]$ (since \NRE holds), contradicting $[t_1]<[t_2]$.
	Hence $t_1 \nind t_2$.
\end{proof}

We note that \RPI and \IRE are reasonable assumptions since all major case studies discussed in~\cite{LPU24} satisfy both.

The following two lemmas will be useful for proofs in this section.
\begin{lemma}[\protect{\cite[Lemma 4.15]{LPU24}}]%
	\label{lem-rpi-ev}
	Assume an LTSI is pre-reversible. If \(e_1 \ci e_2\) then we have also \(\rev{e_1} \ci e_2\).
\end{lemma}

\begin{lemma}[Sideways square~\protect{\cite[Lemma 2.11]{PU07a}}\protect{\cite[Lemma 7]{aubert2023c}}]%
	\label{lem-sideways}
	Let \(t_1\), \(t_2\) be composable forward transitions in a pre-reversible LTSI that satisfies \CIRE. %
	If $t_1 \coind t_2$ then there are composable forward transitions $t_1'$, $t_2'$ such that $t_1$, $t_2'$ are coinitial, $t_1\ind t_2'$, 
	and $t_2$, $t_1'$ are cofinal.
\end{lemma}

\begin{proof}
	Suppose $t_1 \coind t_2$ and let $e_1 = [t_1]$ and $e_2 = [t_2]$, then $e_1 \coind e_2$.
	By~\autoref{lem-rpi-ev}, $\rev{e_1} \coind e_2$.  Using \CIRE, we obtain $\rev{t_1} \ind t_2$.  Then we can use \SP and \PCI 
	to get the desired $t_1'$, $t_2'$.
\end{proof}

\autoref{lem-sideways} would be invalid without \CIRE: recall the LTSI in \autoref{fig:resolvable conflict} from 
\autoref{ex-resolvable conflict}, where \CIRE does not hold. The topmost transitions $a$ and $b$ are composable and related by $\ci$.
However, there are no transitions $b$, $a$ that complete a commuting square for them.
 
Causality between composable transitions, and thus composable events, can be equivalently expressed via the notion of immediate predecessor:

\begin{definition}[Immediate predecessor]%
	\label{def-immed pred}
	Let \(e_1\), \(e_2\) be forward events.  We say that $e_1$ is an \emph{immediate predecessor} of $e_2$, denoted by $e_1 \ip e_2$, if $e_1 < e_2$ and there is no event $e$ such that $e_1 < e < e_2$.
\end{definition}

\begin{lemma}%
	\label{lem-cau iff ip}
	Let $e_1$, $e_2$ be composable forward events in a pre-reversible LTSI. 
	Then $e_1 < e_2$ iff $\, e_1 \ip e_2$.
\end{lemma}

\begin{proof}
$(\Rightarrow)$  Assume for a contradiction that $e_1<e<e_2$ for some $e$, and suppose \(t_1\), \(t_2\) are composable for some
$t_1\in e_1, t_2\in e_2$.
There is a forward-only rooted path $rt_1t_2$ for some $r$. There are no transitions 
of $e_1$ in $r$ since \NRE holds in a pre-reversible LTSI (\autoref{lem-NRE}). 
Also, since $e_1<e$ 
there are no transitions of $e$ in $r$. This implies that
there are no transitions of $e$ in $rt_1t_2$, contradicting $e<e_2$. Hence no such $e$ exists, so $e_1\ip e_2$.

$(\Leftarrow)$  By the definition of ${\ip}$. 
\end{proof}

\begin{lemma}%
	\label{lem-CIRE ip comp}
	Let \(e_1\), \(e_2\) be forward events in a pre-reversible  LTSI that satisfies \CIRE.
	Then $e_1 \ip e_2$ iff $e_1$ is composable with $e_2$ and $e_1 \ncind e_2$.
\end{lemma}
\begin{proof}
	($\Rightarrow$)
	As $e_1 \ip e_2$ implies $e_1 < e_2$, polychotomy gives $e_1 \ncind e_2$.
	To prove that $e_1$ is composable with $e_2$, consider a rooted forward-only path $r$ finishing with $t_2 \in e_2$. 
	Since $e_1 < e_2$, we have that $r$ contains $t_1 \in e_1$, so $r = r't_1st_2$ for some $r', s$. 
	We proceed by induction on the length of $s$:
    if it is zero, then $e_1$ is composable with $e_2$.
	Otherwise, let $E=\{e \mid \cte(s,e)=1 \ \text{and} \ e_1 < e \not< e_2\}$.
    
    First suppose that $E$ is empty.
	We can write $s = ts'$ for some $t$.
        Note that by definition $\cte(s,[t])=\cte(s',[t])+1=1$, since $s'$ cannot contain transitions in $[t]$ by \NRE, and cannot contain transitions in $\rev{[t]}$ being a forward-only path. %
	Since $[t] \notin E$, either $e_1 \not< [t]$ or $[t] < e_2$ holds; in the latter case, however, $e_1 \not< [t]$ must also hold, otherwise $e_1 < [t] < e_2$ would contradict $e_1 \ip e_2$.
	By polychotomy, it follows that $e_1 \coind [t]$;
	hence, by~\autoref{lem-sideways} we can swap $t_1$ and $t$, obtaining a rooted path $r't't'_1s't_2$
	with $t'_1 \in e_1$.
    The inductive hypothesis allows to conclude that $e_1$ is composable with $e_2$.
	
	Now suppose that $E$ is non-empty. Note that $E$ is a finite partially ordered set: in fact, the number of events occurring in a path $s$ is finite, %
	and \(<\) is a strict partial ordering by \autoref{lem-partial-ordering}. Therefore, there exists a maximal element $e \in E$. We write $r = r't_1s_1ts_2t_2$ for some $t \in e$.
	We claim that $e \coind [u]$ for all $u \in s_2$. By polychotomy, either $e \coind [u]$ or $e < [u]$ holds. If we assume the latter, we obtain that $e_1 < [u]$ by transitivity of $<$. If $[u] < e_2$, we contradict $e_1 \ip e_2$; conversely, if $[u] \not< e_2$, we contradict the maximality of $e$.
	Therefore, we can repeatedly apply ~\autoref{lem-sideways} to move $t$ to the end of $s_2$; moreover, by polychotomy, $e \not< e_2$ yields $e \coind e_2$, so $t$ and $t_2$ can be swapped as well. In conclusion, we obtain a new path $r = r't_1s_1s_2't_2't'$, where $t_2' \in e_2$ and $t' \in e$, such that $s_1s_2'$ is shorter than $s$. We then apply the inductive hypothesis.
	
	($\Leftarrow$) %
	Let $t_1 \in e_1,t_2 \in e_2$ be composable forward transitions, and let $rt_1t_2$ be a rooted forward-only path.
	Clearly $e_1 = e_2$, $e_2 < e_1$ and $e_1 \Cf e_2$ are impossible, and 
	$e_1 \ncind e_2$.
	By polychotomy, $e_1 < e_2$.
	Suppose an event $e$ is such that $e_1 < e < e_2$.
	In the path $rt_1t_2$ there would have to be $t \in e$ preceding $t_2$ and succeeding $t_1$, which is impossible.
	Hence $e_1 \ip e_2$. %
\end{proof}
Lemmas~\ref{lem-CIRE ip comp} and \ref{lem-coind ind adj} combined %
give the following relationship between immediate predecessor and non-independence:
\begin{lemma}%
	\label{lem-IRE RPI ip comp}
	Let \(e_1\), \(e_2\) be forward events in a pre-reversible  LTSI that satisfies \IRE and \RPI.
	Then $e_1 \ip e_2$ iff $e_1$ is composable with $e_2$ and $e_1 \nind e_2$.
\end{lemma}

\begin{proof}
	($\Rightarrow$)
	Assume $e_1 \ip e_2$.
	Using \CIRE (a consequence of \IRE), by \autoref{lem-CIRE ip comp} $e_1$ is composable with $e_2$ and $e_1 \ncind e_2$.
	So $e_1$ is adjacent with $e_2$.
	By \autoref{lem-coind ind adj} \itemref{item-adj-ind-coind-2} we have $e_1 \nind e_2$.
	
	($\Leftarrow$) %
	Assume $e_1$ is composable with $e_2$ and $e_1 \nind e_2$.
	Then $e_1$ is adjacent with $e_2$ and $e_1 \ncind e_2$ by \autoref{lem-coind ind adj} \itemref{item:adj ind coind-1}.
	We deduce $e_1 \ip e_2$ using \autoref{lem-CIRE ip comp}.
\end{proof}

	We note that in fact we can check $e_1 \ip e_2$ under a weaker condition in pre-reversible LTSIs:
\begin{lemma}\label{lem-ip-CIRE}
    Let \(e_1\), \(e_2\) be forward events in a pre-reversible  LTSI that satisfies \IRE. Then $e_1 \ip e_2$ iff $e_1$ is composable with $e_2$ and $\rev{e_1} \nind e_2$.
\end{lemma}

\begin{proof}
    ($\Rightarrow$)
    Assume $e_1 \ip e_2$.
    Using \CIRE (a consequence of \IRE), by \autoref{lem-CIRE ip comp} $e_1$ is composable with $e_2$ and $e_1 \ncind e_2$.
    Hence $\rev{e_1} \ncind e_2$ by~\autoref{lem-rpi-ev}. This implies $\rev{t_1} \nind t_2$ for all coinitial
    $\rev{t_1}\in \rev{e_1}, t_2\in e_2$. By \IRE we deduce $\rev{t_1'} \nind t_2'$ for all $\rev{t_1'}\in \rev{e_1}, t_2\in e_2$. Hence  $\rev{e_1} \nind e_2$ as required.
    
    ($\Leftarrow$) %
    Assume $e_1$ is composable with $e_2$ and $\rev{e_1} \nind e_2$.
    By \autoref{lem-coind ind adj}~\itemref{item:adj ind coind-1} we have $\rev{e_1} \ncind e_2$.
    Hence $e_1 \ncind e_2$ by~\autoref{lem-rpi-ev}.
    We deduce $e_1 \ip e_2$ using \autoref{lem-CIRE ip comp}.    
\end{proof}

Having studied the relationship between causality and non-independence for composable forward events,
we now turn our attention to arbitrary (not necessarily composable) forward events. 
Unlike in the composable case, \IRE and \RPI are not sufficient to ensure that, in pre-reversible LTSIs, causality and non-independence coincide.
Indeed, the following example shows that forward events can be causally related even though they are independent:

\begin{example}%
    \label{ex-ind-dep-ccsk-and-causal-1}
    Consider $t_1: P \r{f}[a] Q$, $t_2: Q \r{f}[b] R$ and $t_3: R \r{f}[c] S$, 
    with $t_1 \nind t_2$, $t_2\nind t_3$ and $t_1\ind t_3$; we let $\rev{t_1}\ind t_3$, $t_1\ind \rev{t_3}$ and $\rev{t_1}\ind \rev{t_3}$.
    This LTSI is pre-reversible and \IRE and \RPI hold. We have $[t_1]\ind [t_3]$, and also $[t_1]<[t_3]$ by 
    the definitions of events and $<$. 
    A variation of this example is re-used later in \autoref{def-ind-dep-ccsk-causal} and illustrated in \autoref{fig:example-1}.
\end{example}

Although causality between arbitrary forward events can overlap with independence,
we can nevertheless provide a method for determining causality that makes use of non-independence.
Its main ingredient is the following %
characterisation: %
in a pre-reversible LTSI, causality coincides with the transitive closure of the immediate predecessor relation, denoted by \(\ip^+\).
Equivalently, two forward events $e_1, e_2$ are causally related precisely when there is a chain of forward events $e_1 \ip e_1'\ip \cdots \ip e_n' \ip e_2$ such that each consecutive pair is in the immediate predecessor relation.

\begin{proposition}%
	\label{prop-chain ip}
	Let $e_1$, $e_2$ be forward events in a pre-reversible LTSI.
	Then $e_1 < e_2$ iff $e_1 \ip^+ e_2$.
\end{proposition}

\begin{proof}
    $(\Rightarrow)$
     Given $e_1 < e_2$, we aim to exhibit a chain of forward events $e_1 \ip e_1'\ip \cdots \ip e_n' \ip e_2$.
    We note that $\{e : e < e_2\}$ is finite: if we take any rooted forward-only path $r$ ending with $t_2 \in e_2$ then any $e <  e_2$ must belong to $r$. This means that by the general theory of finite partial orderings the required chain exists. 
    We show it by induction on the size of $E(e_1,e_2)=\{e \mid e_1 < e < e_2\}$. If $E(e_1,e_2)$ is non-empty then we consider a minimal element $e_1'\in E(e_1,e_2)$. Clearly we have $e_1 \ip e_1'$, and then we consider a smaller $E(e_1',e_2)$. By the inductive hypothesis there is a chain of events between $e_1'$ and $e_2$, where the neighbouring events are related by $\ip$. Hence, the required chain between $e_1$ and $e_2$. 
    
    $(\Leftarrow)$  By the definition of ${\ip}$, we have that $e_1 <^+ e_2$. Then we conclude by transitivity of $<$ (\autoref{lem-partial-ordering}).  
\end{proof}

\autoref{prop-chain ip} gives us a characterisation of causality in terms of 
chains of pairs of events, where each pair is an event and its immediate predecessor. 
To identify such pairs, we have two local methods based on non-independence:
if \IRE and \RPI hold, then events need to be composable and non-independent (\autoref{lem-IRE RPI ip comp}). 
If only \CIRE holds, then events need to be composable and the reverse of the first event must be non-independent with the second
(\autoref{lem-ip-CIRE}).

We shall see in \autoref{sec:true-conc-rel} that the LTSIs for \pccsk and \ccsk are pre-reversible and that also \IRE and \RPI hold.
Consequently, we can use \autoref{lem-IRE RPI ip comp} and \autoref{prop-chain ip} to check causality between arbitrary events.
In practice, however, one is often interested in checking the causality between transitions, rather than events.
We therefore conclude this section by %
proposing a method for checking causality for transitions, that we further illustrate with \autoref{fig:chain}:

\begin{proposition}\label{prop:leq chain}
	Let $t$, $u$ be forward transitions in an LTSI satisfying \IRE and \RPI.
	Then $t < u$ iff there are forward transitions $t = t_1,\ldots,t_k$ and $t'_1,\ldots,t'_k = u$ (for some $k > 1$) such that
	\begin{multicols}{2}
    \noindent
	\begin{itemize}[label={}]
		\item 
		$t_i \eveqt t'_i$ ($i = 1,\ldots,k$);
		\item
		$t'_i$, $t_{i+1}$ are composable and $t'_i \nind t_{i+1}$  ($i = 1,\ldots,k-1$).
	\end{itemize}
    \end{multicols}
\end{proposition}

\begin{proof}
By \autoref{prop-chain ip} and \autoref{lem-IRE RPI ip comp}.
\end{proof}

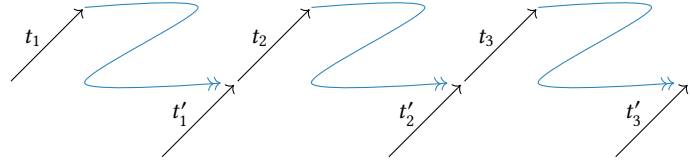
\begin{figure}
	\input{figures/chain.tex}
	\caption{Transitions in \autoref{prop:leq chain} when $k = 3$.
	The curved lines are (possibly empty) paths joining the targets of transitions in the same event using event label equivalence (\autoref{def-event-label-equivalence}, \autoref{prop:eveqt evleqt}).}\label{fig:chain}
		\Description[short description]{long description}%
\end{figure}

\subsection{Conflict}%
\label{subsec:conflict}

In this section we focus on conflict, aiming to characterise it in terms of non-independence. The first ingredient required for this characterisation is the notion of initial conflict, which relates coinitial forward transitions that are not independent. We begin by determining the conditions under which initial conflict does not overlap with core independence on coinitial transitions. We then consider coinitial events and establish when initial conflict coincides with conflict. Finally, we generalise to arbitrary connected forward events: we define a notion of global conflict in terms of initial conflict and causality, and we prove that it coincides with conflict. This will provide a method for checking conflict of connected transitions based on non-independence.

We define a relation of \emph{initial conflict} between coinitial transitions, following the approach taken in occurrence nets~\cite{NPW81} and (reversible) Petri nets~\cite{Rei85,MMU20,Davalos2026}: %
  
\begin{definition}[Initial conflict on transitions]%
	\label{def-initial conf}
	Let \(t_1\), \(t_2\) be distinct forward transitions in pre-reversible LTSI.  
	We say that \(t_1\), \(t_2\) are in \emph{initial conflict}, denoted by $t_1 \Cfi t_2$, if \(t_1\), \(t_2\) are coinitial and $t_1\nind t_2$.
\end{definition}

Note that we explicitly require \(t_1\) and \(t_2\) to be distinct: otherwise, since non-independence is reflexive, the relation would also include pairs of equal transitions and therefore could not characterise conflict.
The following example shows that, in pre-reversible LTSIs where \CIRE does not hold, initial conflict may overlap with independence:

\begin{example}%
    \label{ex-ci coincides with Cfi}
    Recall the LTSI from \autoref{ex-resolvable conflict} shown in \autoref{fig:resolvable conflict}, where \CIRE does not hold. 
    Events $b$ and $c$ are core independent due to the rightmost square, but they are not concurrent because the leftmost coinitial $b, c$ transitions are not independent. This means that the transitions are both in \emph{initial conflict} and core independent, which is not desirable. We note that after the leftmost transition $a$ takes place, initial conflict between $b$ and $c$ is resolved 
    by the commuting square for $b$ and $c$. 
    This pattern of behaviour is traditionally called \emph{resolvable conflict}~\cite{Glabbeek2009}.
\end{example}

If a pre-reversible LTSI satisfies \CIRE, then initial conflict does not overlap with core independence on transitions:
\begin{proposition}\label{lem-coind and Cfi}
Let \(t_1\), \(t_2\) be distinct coinitial forward transitions in a pre-reversible LTSI that satisfies \CIRE. 
Then $t_1\coind t_2$ iff $t_1 \ncfi t_2$.
\end{proposition}

\begin{proof}
($\Rightarrow$)
Since $t_1\coind t_2$ iff $t_1\co t_2$ (\autoref{prop-IRE co coind} which requires \CIRE), $t_1 \ncfi t_2$ follows 
by the definition of $\co$.

($\Leftarrow$)
Assume $t_1 \ncfi t_2$. Since \(t_1\), \(t_2\) are coinitial we have $t_1\ind t_2$ by \autoref{def-initial conf}.
Hence $[t_1]\ci [t_2]$ and thus $t_1\ci t_2$.
\end{proof}

In order to be able to characterise conflict via initial conflict we need another property to hold (in addition to \CIRE): 

\begin{definition}[Event determinism (\ED)~\protect{\cite{Sassone1996,PU07a}}]%
	\label{def-ED}
If  $t, u$ are coinitial forward transitions and $t\sim u$ then $t, u$ are cofinal.  
\end{definition}

If \ED fails then we can have transitions in initial conflict that are not in conflict:

\begin{figure}%
	\input{figures/example-non-ED}
	\caption{
		LTSI for \autoref{ex-non ED}: without \ED, transitions can be in initial conflict but not in conflict.
	}\label{fig:non ED}
	\Description[short description]{long description}%
\end{figure}
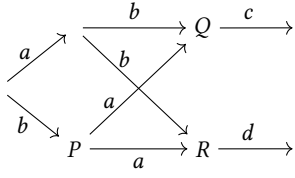

\begin{example}\label{ex-non ED}
    Consider the LTSI in \autoref{fig:non ED} (adapted from \cite[Figure~1]{PU07a}).
    Independence is given by closing under \BTI and \PCI, which produces, among others, independence between the leftmost transitions $a$, $b$. 
    Clearly \WF and \SP hold; hence the LTSI is pre-reversible. Also \CIRE holds.
    There are four events, labelled $a$, $b$, $c$ and $d$. Events $a$ and $b$ are independent, thus core independent. 
    The event $c$ is caused by both $a$ and $b$, and correspondingly for $d$. 
    
    There are two different coinitial transitions labelled $a$, namely $t: P \r{f}[a] Q$ and $u:P\r{f}[a] R$.
    They belong to the event $a$, so $t\sim u$, but they are not cofinal; hence \ED fails.
    We also have $t\nind u$; hence $t \Cfi u$. 
    However, $t \Cf u$ fails since $[t]=[u]$. 
\end{example}

We now lift the initial conflict relation from transitions to events.
We observed with \autoref{ex-ci coincides with Cfi} that \CIRE is necessary to guarantee that if some coinitial forward transitions are non-independent, then all coinitial forward transitions in the same events are non-independent; moreover, assuming \CIRE allows us
to define initial conflict on events via an existential property.  \ED is needed to ensure that initial conflict on events is irreflexive (\autoref{ex-non ED}).

\begin{definition}[Initial conflict on events]
	\label{def-conf dep}
	Let \(e_1\), \(e_2\) be %
	forward events in a pre-reversible LTSI that satisfies \CIRE and \ED.
	The events \(e_1\), \(e_2\) are in \emph{initial conflict}, denoted by $e_1\Cfi e_2$, 
	iff there are $t_1 \in e_1, t_2\in e_2$ \st $t_1 \Cfi t_2$.
\end{definition}

We now show that $\Cfi$ and $\Cf$ coincide on coinitial forward events in pre-reversible LTSIs that satisfy \CIRE and \ED:  

\begin{lemma}%
	\label{lem-Cfi iff conf}
	Let \(e_1\), \(e_2\) be coinitial  forward events in a pre-reversible LTSI with \CIRE and \ED. Then $e_1 \Cfi e_2$ iff
	$e_1 \Cf e_2$.
\end{lemma}

\begin{proof}
	
	($\Rightarrow$) Assume $e_1 \Cfi e_2$. So $t_1 \nind t_2$ for some different coinitial forward $t_1\in e_1, t_2\in e_2$. By \CIRE we have $e_1 \ncind e_2$. 
Clearly $e_1 \not< e_2$ and $e_2 \not< e_1$. So by polychotomy	 $e_1 = e_2$ or $e_1 \Cf e_2$.
Suppose $e_1 = e_2$.  Then $t_1 \eveqt t_2$, so that $t_1 = t_2$ by \ED, which is a contradiction.  Hence $e_1 \Cf e_2$.

	($\Leftarrow$) Polychotomy gives $e_1\neq e_2$ and $e_1 \ncind e_2$.  Take coinitial $t_1\in e_1$ and $t_2\in e_2$.
We have that $t_1, t_2$ are different (since $e_1\neq e_2$) and $t_1 \nind t_2$ by \CIRE. This gives $t_1 \Cfi t_2$; hence $e_1 \Cfi e_2$. 	
\end{proof}
A consequence of this lemma is that, assuming \CIRE and \ED, coinitial forward events are in conflict precisely when they are not independent.

We now turn our attention from coinitial to arbitrary connected forward events.
We first observe that, unlike in the coinitial case, \CIRE and \ED are not sufficient to ensure that conflict and independence do not overlap:

\begin{example}%
	\label{ex-ind-dep-ccsk-and-causal-2}
	Consider $t_1: P \r{f}[\tau] Q$, $t_2: P\r{f}[\tau] R$ and $t_3: R \r{f}[b] S$, as shown in \autoref{fig:example-2}. 
	Let $t_1 \nind t_2$, $t_2\nind t_3$ and $t_1\ind t_3$. 
	This LTSI is pre-reversible, and both \ED and \IRE (hence \CIRE) hold.
	Clearly $[t_1]\Cf [t_3]$ and $[t_1]\ind [t_3]$.
\end{example}

Nevertheless, we can provide a method to determine conflict between arbitrary events using non-independence.
To this end, we introduce the following new notion of conflict on events:

\begin{definition}[Global conflict on events]
	\label{def-gl conf dep}
	Let \(e_1\), \(e_2\) be %
	forward events in a pre-reversible LTSI that satisfies \CIRE and \ED.
	The events \(e_1\), \(e_2\) are in \emph{global conflict}, denoted by $e_1 \Cfg e_2$,  
	iff there are forward events $e_3$, $e_4$ such that $e_3 \Cfi e_4$,  $e_3 \leq e_1$ and $e_4 \leq e_2$.
\end{definition}

We aim to show that conflict and global conflict coincide on connected forward events (\autoref{thm-conf is conf dep}). For this purpose,
we state and prove the following two lemmas:
\begin{lemma}[Conflict inheritance]\label{lem-conf inh}
Let $e_1$, $e_2$ and $e_3$ be connected forward events in a pre-reversible LTSI. If $e_1 \Cf e_2$ and $e_2 \leq e_3$ then 
$e_1 \Cf e_3$.
\end{lemma}

\begin{proof}
Every rooted forward-only path that has $t_3\in e_3$ has also some $t_2\in e_2$. Such paths have no
occurrences of $t_1\in e_1$ since $e_1 \Cf e_2$. Consider any rooted forward only
path $r$ that contains some $t_1\in e_1$. Assume for contradiction that $t_3'\in e_3$ occurs in $r$.
Since $e_2 \leq e_3$, $r$ also contains some $t_2'\in e_2$, which contradicts $e_1 \Cf e_2$.
Hence $e_1 \Cf e_3$. 
\end{proof}

\begin{lemma}\label{lem-ED paths}
    Let $r_1$, $r_2$ be forward-only coinitial paths in a pre-reversible LTSI satisfying \CIRE and \ED.
    If $r_1$, $r_2$ have the same events, then $r_1, r_2$ are cofinal.
\end{lemma}

\begin{proof}
Since $r_1, r_2$ have the same events and are forward-only, they have the same length by \NRE. 
We use induction on the length of $r_1$. 
Let $r_1=t_1r_1'$ and $r_2=r_2't'_1r_2''$, where $t_1'\in [t_1]$. The event $[t_1]$ is core independent with all events in $r_2'$ by polychotomy. Using Sideways Square (\autoref{lem-sideways}), which requires \CIRE, we move $t'_1$ to the start of $r_2$, 
giving $t_1''\in [t_1]$. 
Since $t_1$, $t_1''$ are coinitial, \ED implies they are cofinal. The result follows by the inductive hypothesis on $r_1'$.
\end{proof}

\begin{theorem}%
	\label{thm-conf is conf dep}
	Let $e_1$, $e_2$ be connected forward events in a pre-reversible LTSI satisfying \CIRE and \ED.
	Then $e_1\Cf e_2$ iff $e_1\Cfg e_2$. 
\end{theorem}

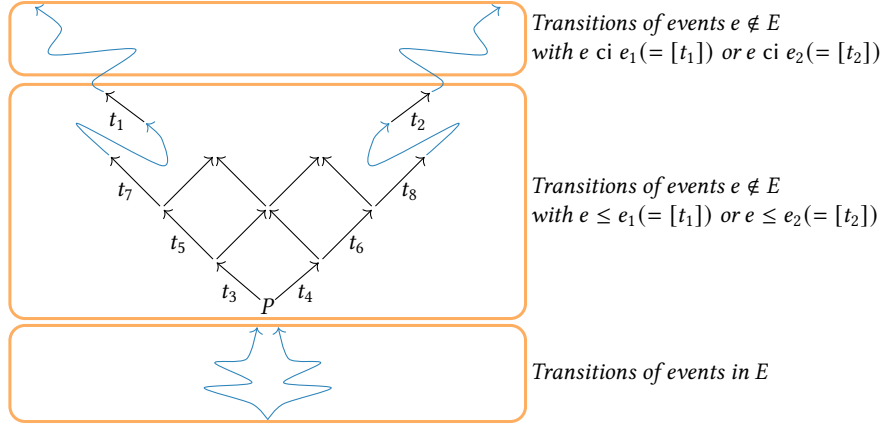
\begin{figure}%
	\input{figures/initial-conflict-updated.tex}
	\Description[short description]{long description}%
	\caption{%
		A visual aid to understand the construction in the proof of \autoref{thm-conf is conf dep}.
		The diagram shows forward-only rooted paths $r_1$ (on the left containing $t_3t_5t_7$ and $t_1$) and 
		$r_2$ (on the right containing $t_4 t_6 t_8$ and $t_2$). 
		The distance from $(t_3, t_4)$ to itself is 0, while
		the distance from $(t_5, t_4)$ and from $(t_3, t_6)$ to $(t_3, t_4)$ is 1.
		There are three pairs of transitions 
		at distance 2 from $(t_3, t_4)$, namely $(t_7, t_4)$, $(t_5, t_6)$ and $(t_3, t_8)$.
		We assume that the transitions in at least one of these pairs are not core independent;
		conversely, the transitions of all pairs at distance at most 1 are core independent, as indicated by the commuting squares.}
	\label{fig:initial_cofl}
\end{figure}

\begin{proof}
	$(\Rightarrow)$
	Suppose that $e_1\Cf e_2$.
	Since \(e_1\), \(e_2\) are connected, there exist $t_i''\in e_i$ ($i = 1,2$) and a path $s$ between \(\srcof{t_1''}\) and \(\tgtof{t_2''}\).
	By \WF, there exist forward-only rooted paths $r_i''$ whose target is \(\srcof{t_1''}\) and \(\srcof{t_2''}\) respectively.
	First, we show that $r_1'',r_2''$ have the same source and thus are coinitial.
	Since $r_1'' s \rev{t_2''} \rev{r_2''}$ is a path between $\srcof{r_1''}$ and $\srcof{r_2''}$, by \PL we obtain a parabolic 
	path between them. As  $\srcof{r_1''}$ and $\srcof{r_2''}$ are irreversible processes, this path is empty and the two processes coincide, so $r_1'',r_2''$ are coinitial.
	Next, starting from paths \(r_1''\) and \(r_2''\), we will construct (in two stages) paths \(r_1\) and \(r_2\), containing the same events as \(r_1''t_1''\) and \(r_2''t_2''\) respectively, and
	where we show $e_1\Cfg e_2$. The paths $r_1, r_2$ are illustrated in \autoref{fig:initial_cofl}. %

	In the first stage we move any transitions in $r_i''$, which do not belong to the events in $E$, towards the back of the path,
	beyond the transitions of the events in $E$.
	The transitions we move are those of the events which are core independent with 
	the transitions belonging to the events in $E$.
	This is done using polychotomy and Sideways Square (\autoref{lem-sideways}), which requires \CIRE.
	Note that the use of Sideways Square preserves events, simply swapping the order of two events.
	It gives us forward-only paths $r_i'$, containing some $t_i'\in e_i$. These paths reach a common state $P$ after performing transitions of 
	the events in $E$ by \autoref{lem-ED paths}.
	
	In the second stage, the events in  $r_i'$ which are beyond $P$ either cause $e_i$  
	or are core independent with $e_i$. Note that they cannot be equivalent to $e_i$ because \NRE holds in pre-reversible LTSIs. 
	We use \autoref{lem-sideways} to move the transitions of all the events $e$
	that are core independent with $e_i$ to the end of the paths $r_i'$ beyond $t_i'$, resulting in forward-only paths $r_i$
	with $t_i\in e_i$.  We note that the set of such events is upwards closed: if $e \coind e_i$ and $e < e'$ then $e' \coind e_i$.
	
	Assume that paths $r_i$ branch out at $P$ with the transitions $t_3$ on $r_1$ and $t_4$ on $r_2$. 
	We then inspect pairs of transitions after $P$, one from $r_1$ starting with $t_3$ and moving towards $t_1$ and the other 
	from $r_2$ starting with $t_4$ and moving towards $t_2$. We check whether or not they are core independent. 
	Consider a pair of such transitions $(t_5,t_6)$. We define the distance from $(t_3,t_4)$ to $(t_5,t_6)$ 
	as the sum of the distances from $t_3$ to $t_5$ and from $t_4$ to $t_6$. The distance between a forward transition 
	and itself is zero, and the distance between %
	transitions $u$ and $u'$ 
	in a forward-only path $u r u'$ is %
	the length of the (possibly empty) path $r$ %
	plus 1.
	Assume that $(t_5,t_6)$ is one of the pairs of non-core-independent transitions with the smallest distance from $(t_3,t_4)$. 
	Note that there could be multiple such pairs: see \autoref{fig:initial_cofl} which shows three pairs %
	at distance 2 from $(t_3,t_4)$.
	This means that all pairs at a smaller distance from $(t_3,t_4)$ are core independent. %
	Clearly, $[t_5]\leq [t_1]$ and $[t_6]\leq  [t_2]$ for each such $t_5$, $t_6$. %
	Hence, it remains to show $[t_5]\Cfi [t_6]$.
	There are two cases:
	
	\begin{enumerate}
	\item
	$t_5=t_3(\neq t_1)$ and $t_6=t_4(\neq t_2)$ and $t_5 \ncind t_6$. 
	Since $t_5, t_6$ are coinitial and \CIRE holds, 
	we have $t_5\nind t_6$. Since $t_3\neq t_4$ we have $t_5\Cfi t_6$, and thus $[t_5]\Cfi [t_6]$. 
		
	\item $t_5 \neq t_3$ or $t_6\neq t_4$. %
	Recall that all the transitions 
	from $t_3$ to $t_5$ are core independent with all the transitions from $t_4$ to $t_6$, except for $t_5$ and $t_6$, 
	meaning we have $t_5 \ncind t_6$.
	
	Using \SP and \CIRE we can construct a grid of commuting squares for transitions $t_3$ to $t_5$ on one side
	and transitions $t_4$ to $t_6$ on the other side, such that there are coinitial $t_5'\in [t_5]$ and $t_6'\in [t_6]$. 
	We have $t_5' \nind t_6'$: assuming the opposite would give $[t_5] \coind [t_6]$, which contradicts $t_5 \ncind  t_6$. 
	Overall, $t_5'\Cfi t_6'$ and  $[t_5]\Cfi [t_6]$.
	\end{enumerate}
		
	$(\Leftarrow)$ Using \autoref{lem-Cfi iff conf}, which requires \CIRE and \ED, and conflict inheritance (\autoref{lem-conf inh}).
\end{proof}

We conclude by giving a method to check conflict between connected transitions based on non-independence.

\begin{proposition}\label{prop:cf leq}
Let $t_1,t_2$ be connected forward transitions in an LTSI satisfying \CIRE and \ED.
Then $t_1 \Cf t_2$ iff there are distinct coinitial forward transitions $t'_1,t'_2$ such that the following three hold:
\begin{multicols}{4}
    \noindent
	\begin{enumerate}
		\item $t'_1 \nind t'_2$;
		\item $t'_1 \leq t_1$;
		\item[] and 
		\item $t'_2 \leq t_2$.
	\end{enumerate}
\end{multicols}
\end{proposition}

\begin{proof}
	By \autoref{thm-conf is conf dep} and the definitions of initial and global conflict.
\end{proof}

%% file: figures/chain.tex
    
    \begin{tikzpicture}[
        x={(1, 0)},  %
        y={(0, 1)}, %
        node distance=0,
        inner sep=0
        ]%
        \node (M0) at (-1, 1) {};
        \node (H0) at (0, 2) {};
        \node (L1) at (1, 0) {};
        \node (M1) at (2, 1) {};
        \node (H1) at (3, 2) {};
        \node (L2) at (4, 0) {};
        \node (M2) at (5, 1) {};
        \node (H2) at (6, 2) {};
        \node (L3) at (7, 0) {};
        \node (M3) at (8, 1) {};
        \node (H3) at (9, 2) {};
        
        \draw[f] (M0) -- node[left = .5em, pos=.6]{\(t_1\)} (H0);	
        \draw[f] (L1) -- node[left = .5em, pos=.6]{\(t'_1\,\)} (M1);	
        \draw[f] (M1) -- node[left = .5em, pos=.6]{\(t_2\)} (H1);	
        \draw[f] (L2) -- node[left = .5em, pos=.6]{\(t'_2\,\)} (M2);	
        \draw[f] (M2) -- node[left = .5em, pos=.6]{\(t_3\)} (H2);	
        \draw[f] (L3) -- node[left = .5em, pos=.6]{\(t'_3\,\)} (M3);	
        
        \draw [fb, color=ind] plot [smooth] coordinates {($(H0)$) ($(H0)!.5!(H1)$) ($(M0) + (1,0)$) ($(M1) - (.2, 0)$)};		
        \draw [fb, color=ind] plot [smooth] coordinates {($(H1)$) ($(H1)!.5!(H2)$) ($(M1) + (1,0)$) ($(M2) - (.2, 0)$)};		
        \draw [fb, color=ind] plot [smooth] coordinates {($(H2)$) ($(H2)!.5!(H3)$) ($(M2) + (1,0)$) ($(M3) - (.2, 0)$)};		
        
    \end{tikzpicture}

%% file: figures/example-non-ED.tex
\begin{tikzpicture}[
	x={(1, 0)},  %
	y={(0, .8)}, %
	baseline,
	anchor=base
	]
	
	\node (a) at (0, -1){};
	\node (b) at (1, -2.13){$P$};
	\node (c) at (1, 0){};
	\node (e1) at (4, 0){};
	\node (e2) at (4, -2){};
	\node (f) at (2.7, -2.13){$R$};
	\node (g) at (2.7, -0.09){$Q$};

	\draw[f] (a) -- node[below, pos=.3]{\(b\)} (b);
	\draw[f] (a) -- node[above, pos=.3]{\(a\)} (c);
	\draw[f] (b) -- node[above, pos=.2]{\(a\)} (g);
	\draw[f] (g) -- node[above, pos=.4]{\(c\)} (e1);
	\draw[f] (c) -- node[above, pos=.5]{\(b\)} (g);
	\draw[f] (c) -- node[above, pos=.4]{\(b\)} (f);
	\draw[f] (b) -- node[below, pos=.5]{\(a\)} (f);
	\draw[f] (f) -- node[above, pos=.4]{\(d\)} (e2);

\end{tikzpicture}

%% file: figures/initial-conflict-updated.tex
	\begin{tikzpicture}[
		x={(.7, 0)},  %
		y={(0, .7)}, %
		baseline,
		anchor=base,
		node distance=0,
		inner sep=0
		]
		
		\node (o) at (0, -2){}; %
		
		\node (p) at (0, 0) {$P$
		};	
		\node (a) at (-1, 1){%
		};
		\node (b) at (1, 1) {%
		};
		\node (c) at (-2, 2){%
		};
		\node (d) at (0, 2) {%
		};
		\node (e) at (2, 2) {%
		};
		\node (f) at (-3, 3){%
		};
		\node (g) at (-1, 3){%
		};
		\node (h) at (1, 3) {%
		};
		\node (i) at (-2.3, 3.6){%
		};
		\node (j) at (-3.1, 4.2){%
		};
		\node (m) at (3, 3){%
		};
		\node (k) at (2.3, 3.6){%
		};
		\node (l) at (3.1, 4.2){%
		};
		
		\draw[f] (p) -- node[below = .2, pos = .7]{\(t_3\)} (a);
		\draw[f] (p) -- node[below = .2, pos = .7]{\(t_4\)} (b);
		\draw[f] (a) -- %
			(d);
		\draw[f] (b) -- %
			(d);
		\draw[f] (a) -- node[below = .2, pos = .7]{\(t_5\)} (c);
		\draw[f] (d) -- (g);
		\draw[f] (c) -- (g);
		\draw[f] (d) -- (h);
		\draw[f] (e) -- (h);
		\draw[f] (b) -- node[below = .2, pos = .7]{\(t_6\)} (e);
		\draw[f] (c) -- node[below = .2, pos = .7]{\(t_7\)}(f);
		
		\draw [f, color=ind] plot [smooth] coordinates {($(f)$) ($(f) + (-.6, .6)$) ($(f) + (1, -.2)$) ($(i) + (.3, -.3)$) ($(i)$)}; 
		
		\draw[f] (i) --node[below = .2, pos = .7]{\(t_1\)} (j);
		\draw[f] (e) -- node[below = .2, pos = .7]{\(t_8\)} (m);
		\draw [f, color=ind] plot [smooth] coordinates {($(m)$) ($(m) + (.6, .6)$) ($(m) + (-1, -.2)$) ($(k) + (-.3, -.3)$) ($(k)$)}; 
		
		\draw[f] (k) --node[below = .2, pos = .7]{\(t_2\)} (l);

		\draw [f, color=ind] plot [smooth] coordinates {($(o)$) ($(o) + (-.3, .3)$) ($(o) + (-1.2, .6)$) ($(o) + (-.3, .7)$) ($(o)!.5!(p) + (-1, 0)$) ($(p) + (-.4, -1)$) ($(p) + (-.2, -.4)$)}; %
		\draw [f, color=ind] plot [smooth] coordinates {($(o)$) ($(o) + (.3, .3)$) ($(o) + (1.2, .6)$) ($(o) + (.3, .7)$) ($(o)!.5!(p) + (1, 0)$) ($(p) + (.4, -1)$) ($(p) + (.2, -.4)$)}; %

		\node (extr) at ($(l.east)+(1.2, 0)$) {}; %
		\node (extl) at ($(j.west)+(-1.2, 0)$) {}; %
		\node (above) at ($(extr)+(0, 2)$) {}; %

		\node (bot) at (p -| extr) {};
		\node (low) at (o -| extl) {};
		
		\begin{scope}[on background layer]
			\node [fill=none, draw=ord, line width=0.4mm,  fit= (bot) (extl), rectangle, label={[xshift=1, align=left]right:\emph{Transitions of events \(e\notin E\)}\\\emph{with \(e \leq e_1(=[t_1])\) or \(e\leq e_2(=[t_2])\)}},
			inner xsep = 10, inner ysep = 3pt, yshift=-1pt,
			rounded corners=0.2cm
			] {};
			
			\node [fill=none, draw=ord, line width=0.4mm,  fit= (bot) (low), rectangle, label={[xshift=1, align=left]right:\emph{Transitions of events in \(E\)}},
			inner xsep = 10,
            inner ysep = -4pt,
            yshift=-4pt,
			rounded corners=0.2cm
			] {};
			
			\node [fill=none, draw=ord, line width=0.4mm,  fit= (extl) (above), rectangle, label={[xshift=1, align=left]right:\emph{Transitions of events \(e\notin E\)}\\\emph{with \(e \ci  e_1 (=[t_1])\) or \(e \ci e_2 (=[t_2])\)}},
			inner xsep = 10,
            inner ysep = -7pt,
			rounded corners=0.2cm
			] {};
			
		\end{scope}

		\draw [f, color=ind] plot [smooth] coordinates {($(l)$) ($(l) + (.2, .2)$) ($(l) + (-.6, .8)$) ($(l) + (1, .7)$) ($(l) + (.8, 1.2)$) ($(l) + (1.3, 1.6)$)}; 
		\draw [f, color=ind] plot [smooth] coordinates {($(j)$) ($(j) + (-.2, .2)$) ($(j) + (.6, .8)$) ($(j) + (-1, .7)$) ($(j) + (-.8, 1.2)$) ($(j) + (-1.3, 1.6)$)}; 
	\end{tikzpicture}

%% file: sections/proved-lts.tex
\section{\texorpdfstring{\pccsk}{CCSKP} and \texorpdfstring{\ccsk}{CCSK}} %
\label{sec:proved-lts}

We recall the extension of \ccsk with proof labels~\cite{aubert2023c}, denoted by \pccsk, and then prove that its labelled transition system (LTS) is in bijection with the LTS of \ccsk, one of its main sources of inspiration\footnote{The other being \pccs, the \enquote{proved} extension of \ccs~\cite{BC88,BC94,DeganoGP03}, briefly discussed in the introduction.}.
All the definitions, results and examples from now on until \autoref{ssec:true-concu-relation-pccsk}
have been formalised and %
can be found in \tbel{code/} and \tbel{examples/}. %
Precise links are given throughout the development, and a paper-to-artifact table is provided in the \hrefurl{\repourl/blob/\latestcommit/overview.md}{overview file}.

\subsection{Definition of \texorpdfstring{\pccsk}{CCSKP}}%
\label{ssec:ccsk-def}

\begin{definition}[(Co-)names, labels and keys]%
	\label{def-co-names}
	Let \(\names\) be a set of {\bel{code/shared/definitions.bel}[\emph{names}][3][5]}, ranged over by $a$, $b$ and $c$. %
	A bijection \(\out{\cdot}:\names \to \out{\names}\), whose inverse is also denoted by~\(\out{\cdot}\), 
	gives the \emph{complement} of a name. %
	The set \(\labelset\) of {\bel{code/shared/definitions.bel}[\emph{labels}][13][18]} is \(\names \cup \out{\names} \cup\{\tau\}\), and we use \(\alpha\), \(\beta\) (\resp \(\lambda\)) to range over \(\labelset\) (\resp \(\labelset \bs \{\tau\}\)).
	
	Let \(\keyset\) be a denumerable set of {\bel{code/shared/definitions.bel}[\emph{keys}][7][11]}, ranged over by \(l\), \(k\), \(m\) and \(n\).
	{\bel{code/ccsk/definitions.bel}[\emph{Keyed labels}][5][8]}, denoted by \(a[l]\), \(a[k]\), \(b[m]\), \etc are elements of \(\labelset \times \keyset = \klabelset\).
\end{definition}

\begin{discussion}%
In the formalisation, the type encoding names does not present any constructor and is dynamically inhabited through the use of \bel{code/shared/definitions.bel}[contexts][30][31]; this design follows other Beluga formalisations of concurrent calculi in which restriction is treated as a binding operator~\cite{SanoKP23,momigliano24,EPTCS448.1}. %
Conversely, keys are encoded by an explicit type isomorphic to the natural numbers, equipped with a decidable \bel{code/shared/definitions.bel}[inequality][58][62] predicate at the LF level. This choice simplifies the encodings of keyed labels, proof labels and keyed prefixes by avoiding meta-level functions indexed by keys, while preserving infinite branching over keys in forward transition rules.
\end{discussion}

\begin{definition}[{\bel{code/shared/definitions.bel}[Processes][20][28]}]%
	\label{def-operators}
	The set \(\kprocset\) of \emph{\ccsk processes} is defined as follows:
	
	\noindent\begin{minipage}{.46\linewidth}
		\begin{align}
			X, Y \coloneqq ~ 
			& \nil \tag{Inactive process} \\
			\|~ &  \alpha. X \tag{Prefix} \\
			\|~ &   X \bs a \tag{Restriction}
		\end{align}
	\end{minipage}
	\begin{minipage}{.5\linewidth}
		\begin{align}
			\|~ & X + Y  \tag{Sum} \\ 
			\|~ & X \mid Y \tag{Parallel composition}\\
			\|~ & \alpha[k].X \tag{Keyed prefix}
		\end{align}
	\end{minipage}

	We write \(\keys{X}\) for the \emph{set of keys in} \(X\),
    and say that \(X\) is \bel{code/shared/definitions.bel}[\emph{standard}][64][71] %
    iff \(\keys{X} = \emptyset\).
	The set \(\procset\) of \ccs \emph{processes} is \(\{X \setst \keys{X} = \emptyset\}\), %
	we let \(P\), \(Q\) range over it.
    Finally, we omit the inactive process \(\nil\) when preceded by a (keyed) prefix.  %
\end{definition}

\begin{discussion}%
	\label{dis:key-implementation}
	While the syntax of \ccs and \ccsk processes is typically presented with restriction over both names and co-names~\cite{milner80lncs,PU07}, in our setting restriction binds only names. This choice does not reduce expressiveness, since every process with restrictions on co-names is behaviourally equivalent to one in which only names are restricted (\eg \(X \bs a\) and \(X \bs \out{a}\) are behaviourally equivalent); moreover, it simplifies the treatment of binding in the formalisation.
	
	In the encoding, processes are considered modulo $\alpha$-renaming, \ie processes that differ only in the choice of names bound by restrictions are identified. Names bound by restrictions are represented as arguments of meta-level functions, following the \emph{higher-order abstract syntax} (HOAS) approach~\cite{pfenning88pldi}; this allows $\alpha$-renaming and capture-avoiding substitutions to be delegated to the meta-level.
	
	An explicit encoding of the set \(\keys{X}\) (\eg  as a list of keys) would require additional machinery for emptiness and (non-)membership.
    Since, in practice, this set is only needed to define standard processes and to express the side conditions of the LTS rules for parallel composition (given later in \autoref{fig:provedltsrulesccskfw} as \(\lmidl\) and \(\rlmidl\)), the formalisation avoids representing the set \(\keys{X}\) and its associated predicates directly.
    Instead, it is more convenient to introduce two inductive predicates, one encoding \bel{code/shared/definitions.bel}[standard processes][64][71] and one expressing \bel{code/shared/definitions.bel}[non-occurrence of a key in a process][73][81].
\end{discussion}

Next we define proof keyed labels for \pccsk.

\begin{definition}[{\bel{code/ccskp/definitions.bel}[Proof keyed labels][17][25]}]%
	\label{def-proof-keyed-labels}
	We let \(\D\) range over the \emph{directions} \(\L\)(eft) and \(\R\)(ight), \(\upsilon\), \(\upsilon_{1}\) and \(\upsilon_{2}\) range over strings in \(\{\lmidd, \lplusd\}^*\), and \(\theta\) range over \emph{proof keyed labels}, denoted by \(\kplabelset\): %
	\begin{align*}
		\theta & \coloneqq \upsilon \alpha[k]  ~\|~ \upsilon \cpair{\upsilon_{1} \lambda[k]}{\upsilon_{2} \out{\lambda}[k]}
	\end{align*}
\end{definition}

\begin{notation}%
    \label{notation:key}
    We let \(\OpDir{\D} = \R\) if \(\D = \L\), else \(\OpDir{\D} = \L\).
    We generally omit \enquote{keyed} and simply write \enquote{proof label}.
\end{notation}

\begin{discussion}%
	\label{dis:pr-lab}
To simplify the encoding of proof labels and obtain a more direct induction principle, their formalisation is by design over-expressive, allowing any two proof labels to be paired regardless of their key or label: the formalisation does not forbid ill-formed proof labels such as \(\cpair{a[k]}{\out{b}[k]}\) and \(\cpair{a[k]}{\out{a}[l]}\).
Well-formedness is however recovered through a separate \bel{code/ccskp/definitions.bel}[validity predicate][80][91], which rules out such spurious proof labels.
\end{discussion}

Given a proof label, we introduce functions for extracting its label and key.
\begin{definition}[{\bel{code/ccskp/definitions.bel}[Label and key functions][38][56]}]%
	We define the \emph{label function} \(\ell : \kplabelset \to \labelset\) as
	\begin{align*}
		\ell(\upsilon \alpha[k]) &= \alpha &&& \ell(\upsilon \cpair{\upsilon_{1}\lambda[k]}{\upsilon_{2} \out{\lambda}[k]}) & = \tau
	\end{align*}
	 and the \emph{key function} \(\kayop: \kplabelset \to \keyset\) as 
	 \begin{align*}
	 	\kay{\upsilon \alpha[k]} &= \kay{\upsilon \cpair{\upsilon_{1}\lambda[k]}{\upsilon_{2} \out{\lambda}[k]}} = k\text{.}
	 \end{align*}
\end{definition}

\begin{definition}[LTS for \pccsk~\cite{aubert2023c}]%
	\label{def:lts-pccsk}
	The \emph{LTS for \pccsk} is 
	$(\kprocset, \kplabelset, \pr{fb}[\theta])$, where \bel{code/ccskp/definitions.bel}[$\pr{fb}[\theta]$ is the union of transition relations][214][218] %
	generated by 
	the \bel{code/ccskp/definitions.bel}[forward][111][148] and \bel{code/ccskp/definitions.bel}[backward][156][194] rules of 
	\autoref{fig:provedltsrulesccskfw}.	
\end{definition}

\begin{figure}%
	\input{figures/ccskp-lts.tex}
	Rules \(\lmidr\),\(\rlmidr\), \(\lplusr\) and \(\rlplusr\) are the symmetric of \(\lmidl\),\(\rlmidl\), \(\lplusl\) and \(\rlplusl\) and are omitted.
	\caption{Forward and backward transition rules for \pccsk. }
	\Description[short description]{long description}%
	\label{fig:provedltsrulesccskfw}
\end{figure}

\begin{discussion}%
	\label{dis:open-closed}
	In the LTS rules for restrictions res and \tRev{res} in \autoref{fig:provedltsrulesccskfw}, the side condition \enquote{\(\labl{\theta} \notin \{a, \out{a}\}\)} does not prevent the restricted name \(a\) from occurring free in the transition label \(\theta\). In such cases, since the encoding identifies processes up to \(\alpha\)-renaming, there is a mismatch: restricted processes abstract away from any renaming of \(a\), whereas the proof label \(\theta\) does not. This issue was overlooked in the original formalisation~\cite{cecilia_2025_15660907}, whose encoding of the res and \tRev{res} rules does not allow \(a\) to occur free in \(\theta\), making the encoding incomplete. %
	
	As a solution, our new encoding introduces \bel{code/ccskp/definitions.bel}[open proof labels][27][30], \ie meta-level functions from names to proof labels, that correspond to proof labels considered modulo \(\alpha\)-renaming of a name occurring in them. In contrast, proof labels without name identifications are referred to as \emph{closed}. This distinction propagates throughout the entire development, leading to a somewhat cumbersome duplication of LTS rules into \emph{closed} and \emph{open} transition rules; however, the bijection with \ccsk (\autoref{lem-bijection}) and the instantiation of the axiomatic theory to \pccsk (\autoref{thm-axioms-hold-pccsk}) provide strong evidence of the soundness of this approach.	
\end{discussion}

\begin{example}
	\label{ex-open}
	Consider the transition \( (a \mid \out{a})  \bs a \ \pr{f}[\cpair{ a \protect{[k]}}{ \out{a} \protect{[k]}}] \ (a[k] \mid \out{a}[k]) \bs a \), obtained by applying the rules act, syn and res, that expresses the synchronisation of two threads over a restricted name \(a\), and let \( \theta = \cpair{ a \protect{[k]}}{ \out{a} \protect{[k]}}\). Since \(\labl{\theta} = \tau\), the side condition \enquote{\(\labl{\theta} \notin \{a, \out{a}\}\)} is trivially satisfied, even if the name $a$ occurs in \(\theta\). \bel{examples/examples-paper.bel}[In the encoding][6][19], %
	 we represent \( \theta \) as the open proof label {\emacsfont open \textbackslash a.(pr\_sync (pr\_base (inp a) k) (pr\_base (out a) k))}, where {\emacsfont open} is the unique constructor for the type {\emacsfont pr\_lab\_op} of open proof labels and takes as argument a meta-level function {\emacsfont \textbackslash a.(...)} that maps names into closed proof labels.
\end{example}

\begin{remark}
	\label{rem:transitions-in-beluga}
	We import from \autoref{def-transitions-paths} and encode in Beluga the notions of coinitial transitions, (rooted) \bel{code/ccskp/definitions.bel}[paths][220][225], \bel{code/ccskp/definitions.bel}[connected transitions][232][235], \etc
    We formalise paths via the reflexive transitive closure of the transition relation \(\pr{fb}\), denoted by \(\pr{fb}^*\).
	Notations differ, since $X \pr{fb}[\theta] X'$ means either $X \pr{f}[\theta] X'$ or $X \pr{b}[\theta] X'$, but for LTSIs $P \r{fb}[\alpha] Q$ means $P \r{f}[a] Q$ or $P \r{b}[a] Q$ depending on $\alpha$: this small gap is consistent with previous work, and will be resolved by context.
\end{remark}

\begin{notation}[Transition key and label]%
	\label{def-transition-key}
	Given a transition \(t: X \pr{bf}[\theta] Y\), we write \(\key{t} = \kay{\theta}\) for its key and  \(\lblof{t} = \theta\) for its label. 
\end{notation}

The following definition is useful to rule out processes with an inconsistent use of keys, \eg \(a[k].\out{b}[m] \mid b[m].\out{a}[k] \) or \(a[k].b[k]\), that cannot be legitimately reached in any computation trace starting from a standard process:

\begin{definition}[{\bel{code/ccskp/definitions.bel}[Reachable processes][227][230]}]%
	\label{definition:reachable}
	We say that \(X\) is \emph{reachable} if there exists an \emph{origin process} \(\orig{X}\) \st \(\keys{X} = \emptyset\) and a rooted path \(r_{X}: \orig{X} \pr{fb}^* X\).
\end{definition}

Note that \(\orig{X}\), if it exists, is unique and is easily obtained by erasing the keys in \(X\)~\cite{LaneseP21}.
We only consider reachable processes from now on.
As \(\pr{f}\) and \(\pr{b}\) are symmetric, we have:

\begin{lemma}[{\bel{code/ccskp/basic-properties.bel}[Loop lemma][33][109]}]%
	\label{lem-loop_proved}
	For all \(t: X \pr{f}[\theta] Y\), there exists \(\Rev{t}: Y \pr{b}[\theta] X\), and conversely.
	Furthermore, \(\Rev{\Rev{t}} = t\).
\end{lemma}

\begin{example}%
	\label{ex-transitions}
	Some of the \bel{examples/examples-paper.bel}[processes reachable from \(a[k] \Par (\out{a} + b)\)][25][99] are %
	 in~\autoref{fig:example-transitions}. 
	Since we suppose infinitely many keys, \pccsk is infinitely branching, as suggested by the transitions labelled \(\lmidr \lplusr b[n]\) and \(\lmidr \lplusr b[m]\): each forward transition could use different keys, provided they do not occur in the source process. %
	Following \autoref{lem-loop_proved}, \bel{examples/examples-paper.bel}[all \(\pr{f}\) transitions could become \(\pr{b}\)][101][110], and vice versa. 
	\bel{examples/examples-paper.bel}[The origin process is \(a \Par (\out a + b)\)][112][119], and \bel{examples/examples-paper.bel}[all transitions are connected][121][135].
\end{example}

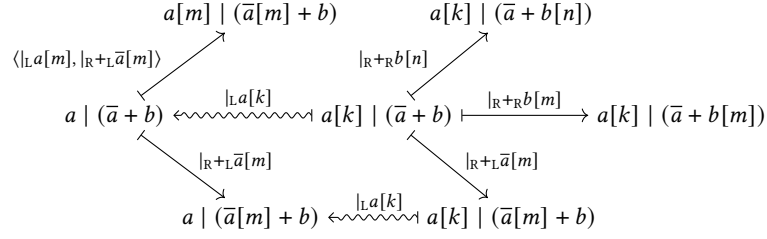
\begin{figure}
	\input{figures/example-transition.tex}
	\caption{A sample of processes reachable from \(a[k] \Par (\out{a} + b)\), used in Examples~\ref{ex-transitions} and \ref{ex-part2}.}%
	\Description[short description]{long description}%
	\label{fig:example-transitions}
\end{figure}

A useful property of \pccsk's LTS is that constructing the derivation tree of a transition is completely deterministic:

\begin{lemma}[{\bel{code/ccskp/unique-step.bel}[Unique derivation trees of \pccsk transitions][562][588]}]%
	\label{lem-derivation-uniqueness-ccskp}
	For any transition \(t\) in \pccsk, there exists exactly one derivation tree whose conclusion is $t$.
\end{lemma}

The proof addresses closed and open transitions, as well as forward and backward transitions, separately, and proceeds by induction on the length of the derivation tree of \(t\). 
The restriction cases for open transitions rely on an \bel{code/ccskp/unique-step.bel}[auxiliary lemma][276][336] showing that, for any given source and target processes, there cannot simultaneously exist both a closed and an open transition between them. 
Although applications of the act, pre, and res rules (and their inverses) are not reflected in the proof label, they can be uniquely determined from the structure of the source and target processes.

\begin{remark}%
	\label{rem:note-on-uniqueness}
	Note that \autoref{lem-derivation-uniqueness-ccskp} would not hold as stated if our LTS was incorporating a structural congruence~\cite{aubert2023c} containing \eg \(X \Par Y \equiv Y \Par X\), as parallel processes could be swapped any even number of times in the derivation trees.
\end{remark}

\subsection{Bijection Between \texorpdfstring{\pccsk}{CCSKP} and \texorpdfstring{\ccsk}{CCSK}}%
\label{ssec-pccsk-ccsk-bijection}

This section defines \ccsk~\cite{PU07} and proves that \pccsk and \ccsk's transitions are in bijection (\autoref{lem-bijection}).
In a nutshell, \ccsk%
 is \enquote{\pccsk without the proof part of the label}: %
 its LTS can be obtained from \pccsk by replacing $\kplabelset$ with \(\klabelset\) and $\theta$ with $\ell(\theta)[\kay{\theta}]$ in \autoref{fig:provedltsrulesccskfw}.
The bijection will prove particularly useful for 
transferring the independence and dependence relations (\autoref{def-ind-dep-ccsk}), as well as the proof that the axiomatic theory holds (\autoref{thm-axioms-hold-ccsk}), to \ccsk.

\begin{definition}[LTS for \ccsk]%
	\label{def-ccsk}
	The set of \bel{code/shared/definitions.bel}[\emph{processes}][20][28] for \ccsk is $\kprocset$ (\autoref{def-operators}) and the set of \bel{code/ccsk/definitions.bel}[\emph{labels}][5][8] is $\klabelset$ (\autoref{def-co-names}). %
	The \bel{code/ccsk/definitions.bel}[\emph{forward transition relation}][10][29] for \ccsk, denoted by \(\r{f}[\alpha][k]\), is generated by the rules in \autoref{fig:ltsrulesccskfw}---with \(\keysop(X)\) as in \autoref{def-operators}. %
	The \bel{code/ccsk/definitions.bel}[\emph{backward transition relation}][31][50] for \ccsk, denoted by \(\r{b}[\alpha][k]\), is defined as the symmetric of \(\r{f}[\alpha][k]\)~\cite[Figure 3]{PU07}\cite{LaneseP21}.
	The \bel{code/ccsk/definitions.bel}[\emph{combined transition relation}][52][56] for \ccsk, denoted by \(\r{fb}[\alpha][k]\), is defined as the union of \(\r{f}[\alpha][k]\) and \(\r{b}[\alpha][k]\). 
	The \emph{LTS for \ccsk} is $(\kprocset, \klabelset, \r{fb}[\alpha[k]])$.
\end{definition}

Observe that the notations for the transition relations of \pccsk (\(\pr{f}\), \(\pr{b}\) and \(\pr{bf}\)) and \ccsk (\(\r{f}\), \(\r{b}\) and \(\r{bf}\)) differ on the presence or absence of the vertical bar%
, in addition to the proved part being removed from the label.
We otherwise use the same functions, notations and conventions as for \pccsk: \eg we denote by \(\lblof{t}\) the label of a transition \(t\) (\autoref{def-transition-key}) and consider only reachable processes (\autoref{definition:reachable}).

\begin{figure}
    \input{figures/ccsk-lts.tex}
	\caption{Forward transition rules for \ccsk. Backward rules are the symmetric versions of the forward rules, thus are omitted.}
	\label{fig:ltsrulesccskfw}%
	\Description[short description]{long description}%
\end{figure}

\begin{lemma}[{\bel{code/ccsk/unique-step.bel}[Unique derivation trees of \ccsk transitions][205][217]}]%
	\label{lem-derivation-uniqueness-ccsk}
	For any transition \(t\) in \ccsk, there exists exactly one derivation tree whose conclusion is $t$.
\end{lemma}

The proof proceeds as in \autoref{lem-derivation-uniqueness-ccskp}, with the structure of the source and target processes guiding the derivation.

\begin{definition}[Mappings between \pccsk and \ccsk]%
	\label{def-ccsk-pccsk-bijection}
	We define the \bel{code/bijection/definitions.bel}[\emph{proof forgetful} (\(\base{\cdot}\))][53][59] and \bel{code/bijection/definitions.bel}[\emph{proof enrichment} (\(\prov{\cdot}\))][61][67] mappings between transitions of \pccsk and \ccsk
	\begin{align*}
		\base{\cdot}&: (X_1 \pr{bf}[\theta] X_2) \mapsto (X_1 \r{bf}[\ekay{\theta}] X_2) \tag{Proof forgetful (\pccsk to \ccsk)} \\
		\prov{\cdot}&: (X_1 \r{bf}[\alpha][k] X_2) \mapsto (X_1 \pr{bf}[\theta] X_2) & \text{\st $\labl{\theta} = \alpha$, $\kay{\theta} = k$} \tag{Proof enrichment (\ccsk to \pccsk)}
	\end{align*}
	as follows:
	
	\begin{description}
		\item[$\base{\cdot}$] leverages that derivation trees in \pccsk are unique (\autoref{lem-derivation-uniqueness-ccskp}) to construct the \ccsk derivation using the corresponding \ccsk rule (in the forward LTS if the transition was forward, in the backward LTS otherwise). %
		\item[$\prov{\cdot}$] follows the same reasoning but swapping \pccsk and \ccsk, and using \autoref{lem-derivation-uniqueness-ccsk}. 
	\end{description}
\end{definition}

	\begin{example}
	\label{ex-mapping}
	We provide in \autoref{fig:mapping-example} an \bel{examples/examples-paper.bel}[example][138][199] of how the forgetful mapping $\base{\cdot}$ proceeds.
    Consider the \pccsk transition 
    in the upper box, whose
    derivation tree is unique by \autoref{lem-derivation-uniqueness-ccskp}.
    One can then construct the corresponding derivation tree in \ccsk, whose conclusion, presented in the lower box, is the corresponding \ccsk transition.
	\end{example}

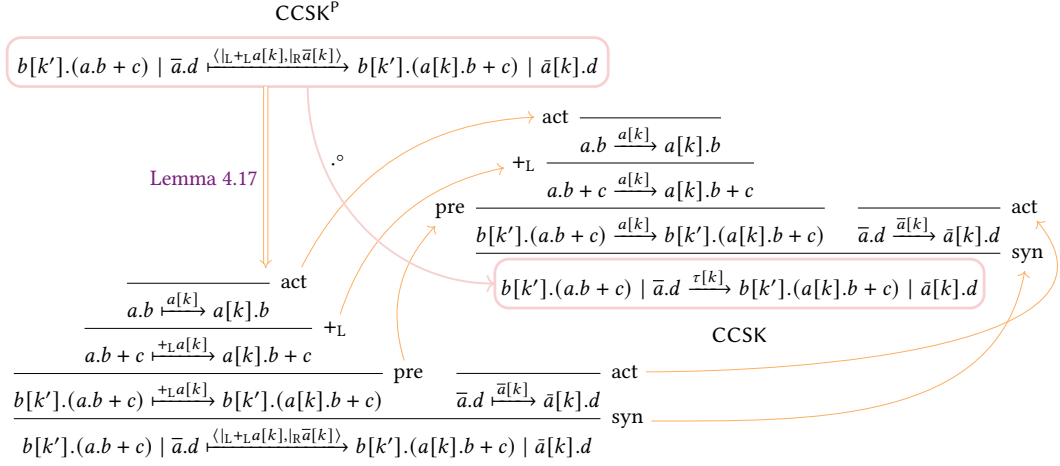
\begin{figure}
    \input{figures/mapping-forget.tex}
	\caption{An example of the forgetful mapping \(\base{\cdot}\) from \pccsk to \ccsk transitions.}
    \Description[short description]{long description}%
    \label{fig:mapping-example}
\end{figure}

\begin{discussion}
	Defining \(\base{\cdot}\) and \(\prov{\cdot}\) as relations requires to distinguish open and closed, forward and backward, transitions but is otherwise straightforward.
	Showing that the defined relations are indeed functions requires to prove that they are \bel{code/bijection/functionality.bel}[functional] and \bel{code/bijection/totality.bel}[total], relying on \bel{code/ccskp/lemmas-bijection.bel}[auxiliary lemmas].
\end{discussion}

\begin{lemma}[{\bel{code/bijection/bijection.bel}[Bijection between \pccsk and \ccsk]}]%
	\label{lem-bijection}
	The mappings \(\base{\cdot}\) and \(\prov{\cdot}\) are mutual inverses.
\end{lemma}

The formalised proof separately establishes that \bel*{code/bijection/bijection.bel}[enrich \(\circ\) forget = id][81][98] and that \bel*{code/bijection/bijection.bel}[forget \(\circ\) enrich = id][143][156]. %

This bijection proves crucial for \enquote{porting} the independence relation from \pccsk to \ccsk in \autoref{def-ind-dep-ccsk}, and then leveraging the properties of \pccsk to demonstrate the same properties for \ccsk in \autoref{ssec:proof-axioms-ccsk}.

%% file: figures/ccskp-lts.tex
	\begin{tcolorbox}[title = {Action, Prefix and Restriction}, sidebyside]
	\begin{tcolorbox}[adjusted title=Forward]
		\begin{prooftree}
			\hypo{}
			\infer[left label={\(\keys{X} = \emptyset%
				\)}]1[act]{\alpha. X \pr{f}[\alpha][k]  \alpha[k].X}
		\end{prooftree}
		\\[1.3em]
		\begin{prooftree}
			\hypo{X \pr{f}[\theta] X'}
			\infer[left label={\(\kay{\theta} \neq k\)}]1[pre]{\alpha[k]. X \pr{f}[\theta] \alpha[k].X'}
		\end{prooftree}
		\\[1.3em]
		\begin{prooftree}
			\hypo{ X \pr{f}[\theta] X '}
			\infer[left label={\(\labl{\theta} \notin \{a, \out{a}\}\)}]1[res]{X  \bs a  \pr{f}[\theta] X ' \bs a}
		\end{prooftree}
	\end{tcolorbox}
	\tcblower
	\begin{tcolorbox}[adjusted title=Backward]
		\begin{prooftree}
			\hypo{}
			\infer[left label={\(\keys{X} = \emptyset
				\)}]1[\tRev{act}]{ \alpha[k].X \pr{b}[\alpha][k] \alpha. X}
		\end{prooftree}
		\\[1.3em]
		\begin{prooftree}
			\hypo{X' \pr{b}[\theta] X}
			\infer[left label={\(\kay{\theta} \neq k\)}]1[\tRev{pre}]{\alpha[k].X'\pr{b}[\theta] \alpha[k]. X }
		\end{prooftree}
		\\[1.3em]
		\begin{prooftree}
			\hypo{ X' \pr{b}[\theta] X}
			\infer[left label={\(\labl{\theta} \notin \{a, \out{a}\}\)}]1[\tRev{res}]{X ' \bs a\pr{b}[\theta] X  \bs a  }
		\end{prooftree}
	\end{tcolorbox}
\end{tcolorbox}
\begin{tcolorbox}[title = Parallel, sidebyside]
	\begin{tcolorbox}[adjusted title=Forward]
		\begin{prooftree}
			\hypo{X \pr{f}[\theta] X'}
			\infer[left label={\(\kay{\theta} \notin \keys{Y}\)}]
			1[\(\lmidl\)]{X \Par Y \pr{f}[\lmidl\theta] X' \Par  Y}
		\end{prooftree}
		\\[1.3em]
		\begin{prooftree}
			\hypo{X \pr{f}[\upsilon_{\L} \lambda][k]  X'}
			\hypo{Y \pr{f}[\upsilon_{\R} \out{\lambda}][k]  Y'}
			\infer2[syn]{X \Par  Y \pr{f}[\cpair{\upsilon_{\L} \lambda \protect{[k]}}{\upsilon_{\R} \out{\lambda} \protect{[k]}}] X' \Par  Y'}
		\end{prooftree}
	\end{tcolorbox}
	\tcblower
	\begin{tcolorbox}[adjusted title=Backward]
		\begin{prooftree}
			\hypo{X' \pr{b}[\theta] X}
			\infer[left label={\(\kay{\theta} \notin \keys{Y}\)}]
			1[\rlmidl]{X' \Par Y \pr{b}[\lmidl\theta] X \Par  Y}
		\end{prooftree}
		\\[1.3em]
		\begin{prooftree}
			\hypo{X' \pr{b}[\upsilon_{\L} \lambda][k]  X}
			\hypo{Y' \pr{b}[\upsilon_{\R} \out{\lambda}][k]  Y}
			\infer%
			2[\tRev{syn}]{X' \Par  Y' \pr{b}[\cpair{\upsilon_{\L} \lambda \protect{[k]}}{\upsilon_{\R} \out{\lambda} \protect{[k]}}] X \Par  Y}
		\end{prooftree}
	\end{tcolorbox}
\end{tcolorbox}

\begin{tcolorbox}[adjusted title=Sum, sidebyside]
	\begin{tcolorbox}[adjusted title=Forward]
		\begin{prooftree}
			\hypo{X \pr{f}[\theta] X'}
			\infer[left label={\(\keys{Y} = \emptyset\)}]1[\( \lplusl \)]{X + Y \pr{f}[\lplusl \theta] X' + Y}
		\end{prooftree}
	\end{tcolorbox}
	\tcblower
	\begin{tcolorbox}[adjusted title=Backward]
		\begin{prooftree}
			\hypo{X' \pr{b}[\theta] X}
			\infer[left label={\(\keys{Y} = \emptyset\)}]1[\(\rlplusl\)]{X' + Y \pr{b}[\lplusl \theta] X + Y}
		\end{prooftree}
	\end{tcolorbox}
\end{tcolorbox}

%% file: figures/example-transition.tex
\tikzset{
	lbl/.style = {
		scale=.8,
	}
}

\begin{tikzpicture}[
	x={(1.5, 0)},  %
	y={(0, 1.5)}, %
	baseline,
	anchor=base
	]
	\node (aab) at (0, 0){\(a \Par (\out{a}+b)\)};
	\node (amamb) at (1.2, .9){\(a[m] \Par (\out{a}[m] + b)\)};
	\node (aamb) at (1.2, -.9){\(a \Par (\out{a}[m] + b)\)};
	\node (akab) at (2.4, 0){\(a[k] \Par (\out{a} + b)\)};
	\node (akamb) at (3.5, -.9){\(a[k] \Par (\out{a}[m] + b)\)};
	\node (akabm) at (5, 0){\(a[k] \Par (\out{a} + b[m])\)};
	\node (akabn) at (3.5, .9){\(a[k] \Par (\out{a} + b[n])\)};
		
	\draw[{Bar[]}->] (aab) -- node[lbl, xshift=-45]{\(\cpair{a[m]}{\lplusl \out{a}[m]}\)} (amamb);
	\draw[{Bar}-{Straight Barb[scale=0.8]}, reverse] (akab)  node[lbl, xshift=-65, yshift=6]{\(\lmidl a[k]\)} -- (aab);
	\draw[{Bar[]}->] (akab) -- node[lbl, xshift=0, yshift=5]{\(\lmidr \lplusr b[m]\)} (akabm);
	\draw[{Bar[]}->] (akab) -- node[lbl, xshift=-26, yshift=0]{\(\lmidr \lplusr b[n]\)} (akabn);
	\draw[{Bar[]}->] (akab) -- node[lbl, xshift=25, yshift=0]{\(\lmidr \lplusl \out{a}[m]\)} (akamb);
	\draw[{Bar[]}->] (aab) -- node[lbl, xshift=25, yshift=0]{\(\lmidr \lplusl \out{a}[m]\)} (aamb);
	\draw[{Bar}-{Straight Barb[scale=0.8]}, reverse] (akamb)  node[lbl, xshift=-62, yshift=5]{\(\lmidl a[k]\)} -- (aamb);
\end{tikzpicture}

%% file: figures/ccsk-lts.tex
	\begin{tcolorbox}[title = {Action, Prefix and Restriction}]
    \begin{prooftree}
        \hypo{}
        \infer[left label={\(\keys{X} = \emptyset%
            \)}]1[act]{\alpha. X \r{f}[\alpha][k]  \alpha[k].X}
    \end{prooftree}
    \hfill 
    \begin{prooftree}
        \hypo{X \r{f}[\beta][k] X'}
        \infer[left label={\(k \neq k'\)}]1[pre]{\alpha[k']. X \r{f}[\beta][k] \alpha[k'].X'}
    \end{prooftree}
    \hfill 
    \begin{prooftree}
        \hypo{ X   \r{f}[\alpha][k]X '}
        \infer[left label={\(\alpha \notin \{a, \out{a}\}\)}]1[res]{X  \bs a  \r{f}[\alpha][k]X ' \bs a}
    \end{prooftree}
\end{tcolorbox}

\begin{tcolorbox}[adjusted title=Parallel]
    \begin{prooftree}
        \hypo{X \r{f}[\alpha][k]X'}
        \infer[left label={\(k \notin \keys{Y}\)}]
        1[\(\lmidl\)]{X \Par   Y \r{f}[\alpha][k] X' \Par  Y}
    \end{prooftree}
    \hfill 
    \begin{prooftree}
        \hypo{Y \r{f}[\alpha][k]Y'}
        \infer[left label={\(k \notin \keys{X}\)}]
        1[\(\lmidr\)]{X \Par   Y  \r{f}[\alpha][k]X \Par  Y'}
    \end{prooftree}
    \hfill 
    \begin{prooftree}
        \hypo{X \r{f}[\lambda][k]  X'}
        \hypo{Y \r{f}[\out{\lambda}][k]  Y'}
        \infer%
        2[syn]{X \Par  Y \r{f}[\tau][k]  X' \Par  Y'}
    \end{prooftree}
\end{tcolorbox}
\begin{tcolorbox}[adjusted title=Sum]
    \makebox[.45\textwidth][c]{
        \begin{prooftree}
            \hypo{X \r{f}[\alpha][k] X'}
            \infer[left label={\(\keys{Y} = \emptyset%
                \)}]1[\( \lplusl \)]{X + Y \r{f}[\alpha][k] X' + Y}
        \end{prooftree}
    }
    \makebox[.45\textwidth][c]{
        \begin{prooftree}
            \hypo{Y \r{f}[\alpha][k] Y'}
            \infer[left label={\(\keys{X} = \emptyset%
                \)}]1[\( \lplusr \)]{X + Y \r{f}[\alpha][k] X + Y'}
        \end{prooftree}
    }
\end{tcolorbox}

%% file: figures/mapping-forget.tex
\tikzset{tikzmark prefix=ex3-}
\resizebox{\textwidth}{!}{%
\begin{tikzpicture}[remember picture]
    \node [%
    line width=0.4mm,
    inner sep=0pt,
    outer sep=0pt,
    rounded corners=0.2cm,
    draw = tsdep,
    minimum size=.7cm,
    text width=8.6cm,
    text centered,
    label={[yshift=2.3em]center:\(\pccsk\)}] (pccsk-transition)
    {\(b[k']. (a.b + c) \Par \out{a}.d \pr{f}[\cpair{\lplusl a\protect{[k]}}{\out{a}\protect{[k]}}] b[k']. (a[k].b + c) \Par \bar{a}[k].d\)};
    \node (pccsk-derivation) [below = 2.6cm of pccsk-transition, xshift=.3cm] {
        \begin{prooftree}
            \hypo{}
            \infer[]1[\subnode{pccsk-act1}{act}]{a.b \pr{f}[a][k] a[k].b}
            \infer[]1[\subnode{pccsk-plus}{\( \lplusl \)}]{a.b + c \pr{f}[\lplusl a[k]] a[k].b + c}
            \infer[]1[\subnode{pccsk-pre}{pre}]{b[k']. (a.b + c) \pr{f}[\lplusl a[k]] b[k']. (a[k].b + c)}
            \hypo{}
            \infer[]1[\subnode{pccsk-act2}{act}]{\out{a}.d \pr{f}[\out{a}][k] \bar{a}[k].d}
            \infer2[\subnode{pccsk-syn}{syn}]{b[k']. (a.b + c) \Par \out{a}.d \pr{f}[\cpair{\lplusl a\protect{[k]}}{\out{a}\protect{[k]}}] b[k']. (a[k].b + c) \Par \bar{a}[k].d}
    \end{prooftree}};

    \node (ccsk-derivation) [above right = 0cm and -3.2cm of pccsk-derivation] {
        \begin{prooftree}
            \hypo{}
            \infer[left label=\subnode{ccsk-act1}{act}]1[]{a.b \r{f}[a][k] a[k].b}
            \infer[left label=\subnode{ccsk-plus}{\( \lplusl \)}]1[]{a.b + c \r{f}[a[k]] a[k].b + c}
            \infer[left label=\subnode{ccsk-pre}{pre}]1[]{b[k']. (a.b + c) \r{f}[a[k]] b[k']. (a[k].b + c)}
            \hypo{}
            \infer[]1[\subnode{ccsk-act2}{act}]{\out{a}.d \r{f}[\out{a}][k] \bar{a}[k].d}
            \infer[]2[\subnode{ccsk-syn}{syn}]{~}
    \end{prooftree}};
    
    \node [%
    line width=0.4mm,
    inner sep=0pt,
    outer sep=0pt,
    rounded corners=0.2cm,
    draw = tsdep,
    minimum size=.7cm,
    text width=7cm,
    text centered,
    label={[yshift=-2.3em]center:\(\ccsk\)}, below = .05cm of ccsk-syn -| ccsk-derivation] (ccsk-transition)
    {\(b[k']. (a.b + c) \Par \out{a}.d \r{f}[\tau[k]] b[k']. (a[k].b + c) \Par \bar{a}[k].d\)};
    \draw[f, ord] (pccsk-act1) to[bend left] (ccsk-act1);
    \draw[f, ord] (pccsk-plus) to[bend left] (ccsk-plus);
    \draw[f, ord] (pccsk-pre) to[bend left] (ccsk-pre);
    \draw[f, ord] (pccsk-act2) edge[out=0,in=-50] (ccsk-act2);
    \draw[f, ord] (pccsk-syn) edge[out=0,in=-100] (ccsk-syn);
    \draw (pccsk-transition) edge[out=-90,in=180,->, thick, draw = tsdep] node[right, pos=.25]{\(\base{\cdot}\)} (ccsk-transition);
    \draw [ord] ($(pccsk-transition.south)+(-.6, 0)$) edge[%
        double equal sign distance, -Implies]
        node[left]{\autoref{lem-derivation-uniqueness-ccskp}} +(0, -2.6); %
\end{tikzpicture}
}

%% file: sections/complementarity.tex
\section{Complementarity of Independence and Dependence for \texorpdfstring{\pccsk}{CCSKP} and \texorpdfstring{\ccsk}{CCSK}}%
\label{sec:ind-complem}

This section defines three relations using only proof labels: independence, dependence, and connectedness, this later named after the relation on transitions it captures (\autoref{prop-connectednessadequacy}).
The important point is that any two proof labels in the connectedness relation are either dependent or independent (\autoref{thm-complementarity}).
While the independence and dependence relations are inspired by existing works on \ccs~\cite{BC88,BC94}, \pccs~\cite{DeganoGP03} and \pccsk~\cite{aubert2023c,DeganoGP03}, we are not aware of any direct formulation of independence and dependence that does not postulate their complementarity.
Moreover, it can be shown that our independence relation is a conservative extension over the concurrency relation for \pccs~\cite[Appendix B.1]{APU24}.
All the definitions, results and examples in this section have been formalised.

\begin{figure}%
    \input{figures/connected-rules.tex}
    \caption{Connectedness on proof labels.}%
    \Description[short description]{long description}%
    \label{fig:pccsk-connect} 
\end{figure}

\begin{figure}%
    \input{figures/dep-indep.tex}
    \caption{Independence and dependence on proof labels.}%
    \Description[short description]{long description}%
    \label{fig:pccsk-dep-ind}
\end{figure}

\begin{definition}[{\bel{code/ccskp/definitions.bel}[Relations on proof labels][239][321]}]%
	\label{def-relation-proof-labels}
	Two proof labels \(\theta_1\), \(\theta_2\) are \emph{connected} (\resp \emph{independent}, \emph{dependent}) if \(\theta_1 \conn \theta_2\) (\resp \(\theta_1 \lind \theta_2\), \(\theta_1 \sdep \theta_2\)) can be derived using the rules in \autoref{fig:pccsk-connect} (\resp \autoref{fig:pccsk-dep-ind}).
\end{definition}

\begin{remark}
	The names of the rules in \autoref{fig:pccsk-dep-ind} follow loosely the proved transition rule names of a system that inspired them~\cite[pp.~257--258]{BC94}.
	Reusing the same rule names across the different relations is intentional and will facilitate proving that connectedness can be split between independence and dependence.
	The only exception is the pair of P\(^2_k\) rules for independence and dependence, which split \autoref{fig:pccsk-connect}'s P\(^2\) by introducing an additional, and trivially complementary, condition on keys; accordingly, they are subscripted with a \enquote{\(k\)} to stress this extra requirement.
\end{remark}

\begin{discussion}%
	\label{rem:rules-overlap}
	Recall that the encoding distinguishes between closed and open proof labels (\autoref{dis:open-closed}), with the latter representing proof labels modulo \(\alpha\)-renaming.
    In Beluga, we first define connectedness, independence and dependence of closed proof labels according to the rules in Figures \ref{fig:pccsk-connect} and \ref{fig:pccsk-dep-ind}; we then lift the definition to open proof labels by requiring the corresponding judgment to hold for every possible instantiation of the bound name.
	
	In the original presentation of the connectedness and dependence relations~\cite{APU24}, the A\(^2\) rule for actions included the premise \enquote{\(\theta\) is not a prefix}, to ensure uniqueness of derivation trees for these relations. In this work we remove this premise, allowing the A\(^1\) and A\(^2\) rules to overlap. This modification has the benefit of simplifying the Beluga encoding, while all relevant properties of the relations on proof labels remain provable without relying on derivation uniqueness.
\end{discussion}

\begin{remark}%
	\label{rem:ind-irreflexivity}
	It is easy to prove that \(\ind\) is \bel{code/ccskp/complementarity.bel}[irreflexive][3][20] and 
	\bel{code/ccskp/complementarity.bel}[symmetric][22][48], and that \bel{code/ccskp/complementarity.bel}[if $\theta_1 \lind \theta_2$ then $\kay{\theta_1} \neq \kay{\theta_2}$][50][150]. %
	Similarly, $\conn$ is \bel{code/ccskp/complementarity.bel}[reflexive][152][171] and \bel{code/ccskp/complementarity.bel}[symmetric][173][203], and $\sdep$ is \bel{code/ccskp/complementarity.bel}[reflexive][205][224] and \bel{code/ccskp/complementarity.bel}[symmetric][226][258] as well.
\end{remark}

\begin{example}%
	\label{ex-part2}
	Re-using the labels of the transitions in \autoref{fig:example-transitions} we have, \eg 
	\begin{align}
		\lmidr \lplusl \out{a}[m] & \conn  \cpair{a[m]}{\lplusl \out{a}[m]} \tag{\bel{examples/examples-paper.bel}[By S\(^1\), C\(^1\) and  A\(^1\) for \(\conn\)][206][209]}\\%
		\lmidr \lplusl \out{a}[m]  & \sdep  \lmidr \lplusr b[n] \tag{\bel{examples/examples-paper.bel}[By P\(^1\) and  C\(^2\) for \(\sdep\)][211][214]} \\
		\lmidl a[k] & \lind \lmidr \lplusr b[n] \tag{\bel{examples/examples-paper.bel}[By P\(^2_k\) for \(\ind\)][216][219]}
	\end{align}
    This matches the intuition: \autoref{fig:example-transitions} displays transitions labelled \(\lmidr \lplusl \out{a}[m]\) and \(\cpair{a[m]}{\lplusl \out{a}[m]}\) that are connected.
    No process can perform both transitions labelled \(\lmidr \lplusr b[n]\) and \(\lmidr \lplusl \out{a}[m]\), illustrating how dependence captures the exclusive choices proper of summation.
    Finally, the process \(a[k] \Par (\out{a} + b)\) can perform transitions labelled \(\lmidl a[k]\) and \(\lmidr \lplusr b[n]\), showing that independence captures transitions that do not prevent each other from happening.
\end{example}

We first prove that connectedness of proof labels characterises indeed connectedness of transitions (\autoref{def-transitions-paths}):

\begin{proposition} %
	\label{prop-connectednessadequacy}
	\begin{enumerate}
		\item \bel{code/ccskp/connectedness-relationship-one.bel}[If \(t_1\) and \(t_2\) are connected  then \(\lblof{t_1} \conn \lblof{t_2}\)].  \label{prop-connectednessadequacy-1}
		\item \bel{code/ccskp/connectedness-relationship-two.bel}[If \(\theta_1 \conn \theta_2\), then there exist connected forward transitions \(t_1\) and \(t_2\) with \(\lblof{t_i} = \theta_i\) ($i = 1,2$)]. \label{prop-connectednessadequacy-2}
	\end{enumerate}
\end{proposition}

The original pen-and-paper proof spans about four pages~\cite[Section B.2]{APU24}.
The mechanised proof of \itemref{prop-connectednessadequacy-1} essentially follows the original (unsurprising, but tedious) argument.
However, the mechanised proof of \itemref{prop-connectednessadequacy-2} diverges significantly~\cite[Section 3.3.2]{Cec25}: in brief terms, the original argument was leveraging a short proposition~\cite[Proposition B.12]{APU24} incorrectly, and the proof strategy had to be completely re-designed.

\begin{discussion}
\label{disc:connectednessadequacy-2}
\autoref{prop-connectednessadequacy}~\itemref{prop-connectednessadequacy-2} holds on paper, and its proof has been machine-checked for closed proof labels. Since the encoding introduces open proof labels, it is natural to ask whether the result remains valid for open proof labels: it turns out that this is not the case. %
For instance, %
consider the (closed) proof labels
\begin{align*}
\theta_1 = \cpair{a[m]}{\lplusl \out{a}[m]} && \text{ and } && \theta_2 = \cpair{\lplusr b[k]}{\out{b}[k]}\text{.}
\end{align*}

It is easy to see that \(\theta_1 \conn \theta_2\) and that, \eg 
\begin{align*}
a . (\nil + b) \Par \out{b} . (\out{a} + \nil) & \pr{f}[\lmidr \out{b}][n] a . (\nil + b) \Par \out{b}[n] . (\out{a} + \nil)  \pr{f}[\theta_1] a[m] . (\nil + b) \Par \out{b}[n] . (\out{a}[m] + \nil) \text{,} \\ 
a . (\nil + b) \Par \out{b} . (\out{a} + \nil) & \pr{f}[\lmidl a][n] a[n] . (\nil + b) \Par \out{b} . (\out{a} + \nil)  \pr{f}[\theta_2] a[n] . (\nil + b[k]) \Par \out{b}[k] . (\out{a} + \nil) \text{.}
\end{align*}

However, proving \autoref{prop-connectednessadequacy}~\itemref{prop-connectednessadequacy-2} for open proof labels requires considering, \eg 
\[ \text{\emacsfont open \textbackslash b.(pr\_sync (pr\_sumr (pr\_base (inp b) k)) (pr\_base (out b) k))} \]
which corresponds to \(\theta_2\) viewed modulo \(\alpha\)-renaming of \(b\). To derive a transition using this open proof label, we must abstract away from the particular choice of \(b\) in \(\theta_2\), by adding a top-level restriction on \(b\) in \(a . (\nil + b) \Par \out{b} . (\out{a} + \nil)\).
However, one can see that the newly obtained process
\[(a . (\nil + b) \Par \out{b} . (\out{a} + \nil))\bs{b} \]
cannot perform the transition labelled \(\lmidr \out{b}[n]\) that previously enabled the execution of \(\theta_1\), precisely because of the restriction.
A more general reasoning, using process schemas, allows us to prove that no process can perform transitions labelled by \(\theta_1\) and by the open proof label which represents \(\theta_2\) modulo \(\alpha\)-renaming of \(b\), hence proving that \autoref{prop-connectednessadequacy}~\itemref{prop-connectednessadequacy-2} does not hold for open proof labels.

\end{discussion}

\begin{remark}%
	Note that the converse of \autoref{prop-connectednessadequacy}~\itemref{prop-connectednessadequacy-1} does not hold: \bel{examples/examples-paper.bel}[\(a[k] \conn \lmidr b[m]\)][224][226], but \(a \pr{f}[a][k] a[k]\) and \(\nil \Par b \pr{f}[\lmidr b][m] \nil \Par b[m]\) \bel{examples/examples-paper.bel}[are not connected][306][319].
    However, there exist two transitions \( a.(\nil \Par b) \pr{f}[a][k] a[k]. (\nil \Par b)\) and \(a[k]. (\nil \Par b) \pr{f}[\lmidr b][m] a[k].(\nil \Par b[m])\) that \bel{examples/examples-paper.bel}[are connected][322][333], illustrating \autoref{prop-connectednessadequacy}~\itemref{prop-connectednessadequacy-2}.
\end{remark}

\begin{example}
Recall 
$\lmidr \lplusl \out{a}[m]  \conn \cpair{a[m]}{\lplusl \out{a}[m]} $ from \autoref{ex-part2}. 
To exemplify \autoref{prop-connectednessadequacy}~\itemref{prop-connectednessadequacy-2}, we can find 
forward transitions with these proof labels in \autoref{fig:example-transitions}, namely
\bel{examples/examples-paper.bel}[$a \Par (\out{a} + b) \pr{f}[\cpair{a[m]}{\lplusl \out{a}[m]}] a[m] \Par (\out{a}[m] + b)$][27][33] and
\bel{examples/examples-paper.bel}[$a[k] \Par (\out{a} + b) \pr{f}[\lmidr \lplusl \out{a}[m]] a[k] \Par (\out{a}[m] + b)$][67][73]. \bel{examples/examples-paper.bel}[These transitions are connected via the following path][342][355]:
\[a \Par (\out{a} + b) \pr{f}[\lmidl a[k]] a[k] \Par (\out{a} + b) \pr{f}[\lmidr \lplusl \out{a}[m]] a[k] \Par (\out{a}[m] + b)\text{.}\]
\end{example}

Connectedness reduces the search space on proof labels: for example, \(\lplusl a[k]\) and \(\lmidr b[m]\) are neither dependent nor independent, but they also cannot belong to connected transitions, as no process can have both \(+\) and \(\Par\) at the top level.

\begin{theorem}[Complementarity on proof labels]
\label{thm-complementarity}
For all \(\theta_1\), \(\theta_2\),
\begin{enumerate}
	\item \bel{code/ccskp/complementarity.bel}[If $\theta_1 \lind \theta_2$ then $\theta_1 \conn \theta_2$][265][290]. \label{thm-complementarity-1}
	\item \bel{code/ccskp/complementarity.bel}[If $\theta_1 \sdep \theta_2$ then $\theta_1 \conn \theta_2$][292][321]. \label{thm-complementarity-2}
	\item \bel{code/ccskp/complementarity.bel}[If $\theta_1 \conn \theta_2$ then either $\theta_1 \lind \theta_2$ or $\theta_1 \sdep \theta_2$][335][405], \bel{code/ccskp/complementarity.bel}[but not both][407][440].  \label{thm-complementarity-3}
\end{enumerate}
\end{theorem}

The proofs of \itemref{thm-complementarity-1} and \itemref{thm-complementarity-2} both proceed by 
straightforward inductions on the derivation tree of the given relation.
To establish \itemref{thm-complementarity-3}, the formalisation first proves that two connected proof labels are either dependent or independent, and then more generally that two proof labels cannot be both dependent and independent. %

The independence and dependence relations %
are extended to the (connected) transitions of \pccsk in a straightforward manner, and then to the transitions of \ccsk using the proof enrichment mapping \(\prov{\cdot}\) (\autoref{def-ccsk-pccsk-bijection}).

\begin{definition}[{\bel{code/ccskp/definitions.bel}[Relations on \pccsk transitions][326][334]}]%
	\label{def-ind-dep-pccsk}
	For transitions \(t_1\), \(t_2\) in \pccsk, we write
	\begin{align*}
		t_1 \ind t_2 & \text{ iff $t_1$ and $t_2$ are connected and \(\lblof{t_1} \lind \lblof{t_2}\),} %
		&&&
		t_1 \dep t_2 & \text{ iff $t_1$ and $t_2$ are connected and \(\lblof{t_1} \ldep \lblof{t_2}\).} %
	\end{align*}	
\end{definition}

\begin{definition}[{\bel{code/ccsk/complementarity.bel}[Relations on \ccsk transitions][3][19]}]%
	\label{def-ind-dep-ccsk}
	For transitions \(t_1\), \(t_2\) in \ccsk, we write
	\begin{align*}
		t_1 \ind t_2 & \text{ iff $t_1$ and $t_2$ are connected and \(\lblof{\prov{t_1}} \lind \lblof{\prov{t_2}}\),} %
		&&& 
		t_1 \dep t_2 & \text{ iff $t_1$ and $t_2$ are connected and \(\lblof{\prov{t_1}} \ldep \lblof{\prov{t_2}}\).} %
	\end{align*}
\end{definition}

We easily deduce from \autoref{prop-connectednessadequacy} (\ref{prop-connectednessadequacy-1}) and \autoref{thm-complementarity}:

\begin{proposition}[Complementarity on transitions]%
	\label{prop-complementarity transitions}
	In \pccsk, \bel{code/ccskp/complementarity.bel}[if $t_1$ and $t_2$ are connected then exactly one of $t_1 \ind t_2$ and $t_1 \dep t_2$ holds][452][471]. \bel{code/ccsk/complementarity.bel}[Similarly for \ccsk][41][66].
\end{proposition}

Both formalisations follow the proof strategy of \autoref{thm-complementarity} (\ref{thm-complementarity-3}) by first proving that connected transitions are either dependent or independent, and then that transitions cannot be both dependent and independent.

Starting now, our contribution will focus on the \emph{LTSIs} of \pccsk $(\kprocset, \kplabelset, \pr{fb}[\theta], \ind)$ and \ccsk $(\kprocset, \klabelset, \r{fb}[\alpha[k]], \ind)$, \ie we will consider our definitions of independence to be central to our development.

%% file: figures/connected-rules.tex
	\begin{tcolorbox}[title = {Connectedness Relation}, fontupper=\linespread{.9}\small, sidebyside, sidebyside align=top]
    \begin{tcolorbox}[adjusted title=Action]
        \raisebox{-1.25em}{ %
            \makebox[.45\textwidth][c]{
                \begin{prooftree}
                    \hypo{}
                    \infer[]1[A\(^1\)]{\alpha[k] \conn \theta}
                \end{prooftree}
            }
        }
        \raisebox{-1.25em}{ %
            \makebox[.45\textwidth][c]{
                \begin{prooftree}
                    \hypo{}
                    \infer[]1[A\(^2\)]{\theta \conn \alpha[k]}
                \end{prooftree}
            }
        }
    \end{tcolorbox}
    \begin{tcolorbox}[adjusted title=Parallel]
        \makebox[.45\textwidth][c]{
            \begin{prooftree}
                \hypo{\theta \conn \theta'}
                \infer[]1[P\(^1\)]{\lmidd \theta \conn \lmidd \theta'}
            \end{prooftree}
        }
        \makebox[.45\textwidth][c]{
            \begin{prooftree}
                \hypo{}
                \infer[]1[P\(^2\)]{\lmidd \theta \conn \lmidod \theta'}
            \end{prooftree}
        }
    \end{tcolorbox}
    \begin{tcolorbox}[adjusted title=Choice]
        \makebox[.45\textwidth][c]{
            \begin{prooftree}
                \hypo{\theta \conn \theta'}
                \infer[]1[C\(^1\)]{\lplusd \theta \conn \lplusd \theta'}
            \end{prooftree}
        }
        \makebox[.45\textwidth][c]{
            \begin{prooftree}
                \hypo{}
                \infer[]1[C\(^2\)]{\lplusd \theta \conn \lplusod \theta'}
            \end{prooftree}
        }
    \end{tcolorbox}
    \tcblower		
    \begin{tcolorbox}[adjusted title=Synchronisation]
        \begin{prooftree}
            \hypo{\theta \conn \theta_{\D}}
            \infer[]1[S\(^1\)]{\lmidd \theta  \conn \cpair{\theta_{\L}}{\theta_{\R}}}
        \end{prooftree}
        \\[1.8em] %
        \begin{prooftree}
            \hypo{\theta_{\D} \conn \theta}
            \infer[]1[S\(^2\)]{\cpair{\theta_{\L}}{\theta_{\R}} \conn \lmidd \theta}
        \end{prooftree}
        \\[1.8em] %
        \begin{prooftree}
            \hypo{\theta_1 \conn \theta'_1}
            \hypo{\theta_2 \conn \theta'_2}
            \infer[]2[S\(^3\)]{\cpair {\theta_1} {\theta_2} \conn \cpair {\theta'_1} {\theta'_2}}
        \end{prooftree}
    \end{tcolorbox}
\end{tcolorbox}

%% file: figures/dep-indep.tex
	\begin{multicols}{2}
	\begin{tcolorbox}[title = {Independence Relation}, fontupper=\linespread{.9}\small]%
		\begin{tcolorbox}[adjusted title = {Action}]
			\vbox to 14
			pt {\vfil
				\hfill \emph{(empty)} \hfill \phantom{a} ~
				\vfil
			}
		\end{tcolorbox}			
		\begin{tcolorbox}[adjusted title=Parallel]
			\makebox[.45\textwidth][c]{
				\begin{prooftree}
					\hypo{\theta \lind  \theta'}
					\infer[]1[P\(^1\)]{\lmidd \theta \lind \lmidd \theta'}
				\end{prooftree}
			}
			\makebox[.45\textwidth][c]{
				\begin{prooftree}
					\hypo{\kay{\theta} \neq \kay{\theta'}}
					\infer[]1[P\(^2_k\)]{\lmidd \theta \lind \lmidod \theta'}
				\end{prooftree}
			}
		\end{tcolorbox}
		\begin{tcolorbox}[adjusted title=Choice]
			\makebox[.45\textwidth][c]{
				\begin{prooftree}
					\hypo{\theta \lind  \theta'}
					\infer[]1[C\(^1\)]{\lplusd \theta \lind \lplusd \theta'}
				\end{prooftree}
			}
			\makebox[.45\textwidth][c]{ %
			}
		\end{tcolorbox}
		\begin{tcolorbox}[adjusted title=Synchronisation]
			\makebox[.45\textwidth][c]{
			\begin{prooftree}
				\hypo{\theta \lind  \theta_{\D}}
				\infer[]1[S\(^1\)]{\lmidd \theta  \lind \cpair{\theta_{\L}}{\theta_{\R}}}
			\end{prooftree}
			}
			\makebox[.45\textwidth][c]{
			\begin{prooftree}
				\hypo{\theta_{\D} \lind \theta}
				\infer[]1[S\(^2\)]{\cpair{\theta_{\L}}{\theta_{\R}} \lind \lmidd \theta}
			\end{prooftree}
			}
			\\[1.3em]
			\begin{prooftree}
				\hypo{\theta_1 \lind  \theta'_1}
				\hypo{\theta_2 \lind  \theta'_2}
				\infer[]2[S\(^3\)]{\cpair {\theta_1} {\theta_2} \lind \cpair {\theta'_1} {\theta'_2}}
			\end{prooftree}
		\end{tcolorbox}
	\end{tcolorbox}
	\begin{tcolorbox}[title = {Dependence Relation}, fontupper=\linespread{.9}\small]%
		\begin{tcolorbox}[adjusted title=Action]
			\raisebox{-1.25em}{ %
				\makebox[.45\textwidth][c]{
					\begin{prooftree}
						\hypo{}
						\infer[]1[A\(^1\)]{\alpha[k] \sdep \theta}
					\end{prooftree}
				}
			}
			\raisebox{-1.25em}{ %
				\makebox[.45\textwidth][c]{
					\begin{prooftree}
						\hypo{}
						\infer[]1[A\(^2\)]{\theta \sdep \alpha[k]}
					\end{prooftree}
				}
			}
		\end{tcolorbox}			
		\begin{tcolorbox}[adjusted title=Parallel]
			\makebox[.45\textwidth][c]{
				\begin{prooftree}
					\hypo{\theta \sdep \theta'}
					\infer[]1[P\(^1\)]{\lmidd \theta \sdep\ \lmidd \theta'}
				\end{prooftree}
			}
			\makebox[.45\textwidth][c]{
				\begin{prooftree}
					\hypo{\kay{\theta} = \kay{\theta'}}
					\infer[]1[P\(^2_k\)]{\lmidd \theta \sdep\ \lmidod \theta'}
				\end{prooftree}
			}
		\end{tcolorbox}
		\begin{tcolorbox}[adjusted title=Choice]
			\makebox[.45\textwidth][c]{
				\begin{prooftree}
					\hypo{\theta \sdep \theta'}
					\infer[]1[C\(^1\)]{\lplusd \theta \sdep \lplusd \theta'}
				\end{prooftree}
			}
			\makebox[.45\textwidth][c]{
				\begin{prooftree}
					\hypo{}
					\infer[]1[C\(^2\)]{\lplusd \theta \sdep \lplusod \theta'}
				\end{prooftree}
			}
		\end{tcolorbox}
		\begin{tcolorbox}[adjusted title=Synchronisation]
			\makebox[.45\textwidth][c]{
			\begin{prooftree}
				\hypo{\theta \sdep \theta_{\D}}
				\infer[]1[S\(^1\)]{\lmidd \theta  \sdep \cpair {\theta_{\L}} {\theta_{\R}}}
			\end{prooftree}
			}
			\makebox[.45\textwidth][c]{
			\begin{prooftree}
				\hypo{\theta_{\D} \sdep \theta}
				\infer[]1[S\(^2\)]{\cpair {\theta_{\L}} {\theta_{\R}} \sdep\ \lmidd \theta}
			\end{prooftree}
			}
			\\[1.3em]
			\begin{prooftree}
				\hypo{\theta_i \sdep \theta'_i}
				\hypo{\theta_j \conn \theta'_j}
				\hypo{i \neq j \in \{1, 2\}%
				}
				\infer[]3[S\(^3\)]{\cpair {\theta_1} {\theta_2} \sdep \cpair {\theta'_1} {\theta'_2}}
			\end{prooftree}
		\end{tcolorbox}
	\end{tcolorbox}
\end{multicols}

%% file: sections/relations-pccsk.tex
\section{Applying the Axiomatic Approach to \texorpdfstring{\pccsk}{CCSKP} and \texorpdfstring{\ccsk}{CCSK}} 
\label{sec:true-conc-rel}

This section first proves that instantiating the axiomatic approach to \pccsk and \ccsk, where the independence relation $\ind$ is as in 
Definitions~\ref{def-ind-dep-pccsk} and \ref{def-ind-dep-ccsk}, produces pre-reversible LTSIs that also satisfy \IRE and \RPI (Theorems~\ref{thm-axioms-hold-pccsk} and \ref{thm-axioms-hold-ccsk}).
This implies, thanks to the uniqueness result (\autoref{thm-uniqueness}), that the independence relations of \pccsk and \ccsk are \enquote{their only one}\footnote{More precisely: any other independence relation must agree with them on adjacent transitions, per \autoref{prop-prerev-coinitial-unique} and \autoref{remark-on-adjacent-indep}.}.
In conjunction with \autoref{def-ind-dep-pccsk}, this further shows that all our rules for independence of labels in \autoref{fig:pccsk-dep-ind} are needed, and that none can be added:
otherwise, they would overlap with dependence and violate complementarity (\autoref{thm-complementarity}).
In summary, %
our notion of independence is as sharp and as concise as possible.

We also discuss some of its limitations and how to overcome them in \autoref{ssec:true-concu-relation-pccsk}: while independence, core independence and concurrency  all coincide on adjacent forward events and transitions (as illustrated in \autoref{fig:summary-sec-ccc}), the relations do not agree on non-coinitial transitions (\autoref{def-ind-dep-ccsk-causal}).
However, it can be proved that any two transitions that are either causally related or in conflict belong to the transitive closure of dependence (\autoref{leqcftsdep}).

\subsection{Mechanised Proofs of Axioms and Properties for \texorpdfstring{\pccsk}{CCSKP}} %
\label{ssec:true-conc-rel-ccskp}

We first show that the axioms 
for pre-reversibility, namely \SP, \BTI, \WF and \PCI, hold for the LTSI for \pccsk $(\kprocset, \kplabelset, \pr{fb}[\theta], \ind)$. We then define the notion of events 
for \pccsk, and finally prove the remaining axioms \IRE and \RPI, as well as the additional properties \ID, \CIRE, \ED and \FLD (\autoref{def-FLD}).

The pen-and-paper proof of those results spans six pages~\cite[Section D]{APU24} and leverages the concept of \emph{locally label-generated LTSI}~\cite[Definition D.2]{APU24} to lighten some of the proofs.
It refines and complements~\cite[Remark D.9]{APU24} the existing proofs of \SP, \BTI, \WF and \PCI~\cite{aubert2023c}.
Our mechanised proofs are to a large extent more direct and self-contained.

\begin{proposition}[{\bel{code/ccskp/axioms/}[The LTSI of \pccsk is pre-reversible]}]%
	\label{prop-pre-reversible-hold-pccsk}
	\SP, \BTI, \WF, and \PCI hold for the LTSI of \pccsk.%
\end{proposition}

The mechanised proof is split between the following files:

\begin{description}
	\item[{\bel{code/ccskp/axioms/defs-and-properties.bel}[defs-and-properties.bel]}] This file contains definitions required to state the axioms. For instance, \BTI requires proving that backward transitions (of type \bel*{code/ccskp/definitions.bel}[bstep][196][200]) are independent, while the independence relation is defined on combined transitions (of type \bel*{code/ccskp/definitions.bel}[step][214][218]): therefore, we introduce a type family \bel*{code/ccskp/axioms/defs-and-properties.bel}[back\_to\_comb][35][39] which acts as a coercion from {\emacsfont bstep} to {\emacsfont step}. These auxiliary definitions are accompanied by proofs of their properties, \eg  the \bel{code/ccskp/axioms/defs-and-properties.bel}[totality of the {\emacsfont back\_to\_comb} coercion][41][51].
	
	In addition, this file contains properties of the LTS of \pccsk that are useful to prove the axioms: for example, the proof that \bel{code/ccskp/axioms/defs-and-properties.bel}[backward transitions cannot use keys not present in their source][269][292].
	\item[{\bel{code/ccskp/axioms/sp.bel}[sp.bel]}] Due to the lack of syntactic sugar for existential quantification in Beluga, the proved statement of \bel{code/ccskp/axioms/sp.bel}[\textbf{SP}][1030][1092] leverages a type \bel{code/ccskp/axioms/sp.bel}[{\emacsfont square}][1023][1027] that witnesses the existence of a closing square diagram.
	The proof is by induction on the length of the derivation tree of the first transition; the length of the file (close to \bkloc{1}) comes from the need to consider all combinations of directions and open/closed transitions.
	\item[{\bel{code/ccskp/axioms/bti.bel}[bti.bel]}] In the formalisation, we first prove that \bel{code/ccskp/axioms/bti.bel}[backward transitions are either equal or independent][314][332], by considering all combinations of open and closed backward transitions and proceeding with a routine induction on the length of the derivation tree of the first transition. The proved statement of \bel{code/ccskp/axioms/bti.bel}[\textbf{BTI}][341][350] immediately follows.%
	\item[{\bel{code/ccskp/axioms/wf.bel}[wf.bel]}] \WF is proven by showing that \bel{code/ccskp/axioms/wf.bel}[every process is accessible][92][106], as usual in mechanised proofs~\cite[Section 3.2]{AAH19}.
	\item[{\bel{code/ccskp/axioms/pci.bel}[pci.bel]}] \PCI is straightforward once proven that \bel{code/ccskp/complementarity.bel}[\(\ind\) is symmetric][22][48].
\end{description}

Since \BTI and \PCI hold for the LTSI of \pccsk, \ID is also known to hold (\autoref{prop-ID}): in Beluga, we prove \bel{code/ccskp/axioms/id.bel}[\textbf{ID}] for \pccsk directly by taking inspiration from the pen-and-paper proof. %

Since the LTSI of \pccsk is pre-reversible, we can instantiate the definition of equivalence of transitions and of events (\autoref{def-event-general}) to it.
As \PCI holds, the commuting square from \autoref{def-event-general} is non-degenerate, \ie all pairs of vertices are distinct (\autoref{lem-non-degenerate}).
This allows us to use a simplified definition of equivalence of transitions and events:

\begin{definition}[{\bel{code/ccskp/axioms/events.bel}[Events, simplified definition][3][21]~\cite[Definition 4.8]{LPU24}}]%
	\label{def-event-simplified}
	Consider a pre-reversible LTSI. Let $\sqeqt$ be the smallest equivalence relation satisfying:
	if $t:P \r{fb}[\alpha] Q$, $u:P \r{fb}[\beta] R$,
	$u':Q \r{fb}[\beta] S$, $t':R \r{fb}[\alpha] S$,
	and $t \ind u$, then $t \sqeqt t'$.
\end{definition}

The equivalence between the \enquote{general} definition of events (\autoref{def-event-general}) and \autoref{def-event-simplified} \bel{code/ccskp/axioms/events.bel}[was formalised][24][108] as well. In particular, given transitions $t$, $u$, $u'$ and $t'$ as in \autoref{def-event-simplified}, one needs to show that the independence of $t$ and $u$ propagates to all pairs of adjacent transitions, and that the given commuting square is non-degenerate. The former follows by observing that the proof labels used by such pairs of transitions coincide with $\lblof{t}$ and $\lblof{u}$; the latter follows from the fact that, if the commuting square was degenerate, we could compose two transitions with the same proof label and direction (\eg  $t$ and $t'$), which is \bel{code/ccskp/axioms/defs-and-properties.bel}[impossible in \pccsk][828][846].

Having defined events, we prove the remaining properties of interest from the axiomatic theory (\autoref{sec:axiomatic}). 

\begin{proposition}[{\bel{code/ccskp/axioms/}[Other properties of the LTSI of \pccsk]}]%
	\label{prop-other-hold-pccsk}
	\IRE, \CIRE, \RPI, and \ED hold for the LTSI of \pccsk $(\kprocset, \kplabelset, \pr{fb}[\theta], \ind)$.
\end{proposition}
 
\begin{description}	
	\item[{\bel{code/ccskp/axioms/ire.bel}[ire.bel]}] Once \autoref{def-event-simplified} is set and it has been established that \bel{code/ccskp/axioms/events.bel}[equivalent transitions are connected][112][126], stating and proving \IRE is immediate.	
	\item[{\bel{code/ccskp/axioms/cire.bel}[cire.bel]}] Obtaining \CIRE is immediate since equivalent transitions have the same proof label, which uniquely determines if transitions are independent for \pccsk.
	\item[{\bel{code/ccskp/axioms/rpi.bel}[rpi.bel]}] \RPI is immediate considering our definition of independence, both on pen-and-paper and in the formalisation.
	\item[{\bel{code/ccskp/axioms/ed.bel}[ed.bel]}] The proof of \ED follows %
	from the property of \emph{Forward label determinism} (\FLD), stated and proved below.
\end{description}

\begin{definition}[Forward label determinism%
    ]%
	\label{def-FLD}
	If  $t, u$ are %
	coinitial forward transitions with $\lblof {t} = \lblof {u}$, then $t = u$.  
\end{definition}

\begin{proposition}
	\label{prop-FLD pccsk}
	\bel{code/ccskp/axioms/fld.bel}[Forward label determinism holds for the LTSI of \pccsk][159][168].
\end{proposition}

The proof requires to split between open and closed transitions and is carried out by induction on the length of the derivation tree of the given transition.
Handling restriction sub-cases for open transitions requires to prove a couple of \bel{code/ccskp/axioms/fld.bel}[technical lemmas about restrictions][50][123].

Propositions \ref{prop-pre-reversible-hold-pccsk} and \ref{prop-other-hold-pccsk} allow us to conclude:

\begin{theorem}%
	\label{thm-axioms-hold-pccsk} The axiomatic theory holds for the LTSI of \pccsk.
\end{theorem}

\subsection{Mechanised Proofs of Axioms and Properties for \texorpdfstring{\ccsk}{CCSK}}
\label{ssec:proof-axioms-ccsk}

To show that the axiomatic theory holds for the LTSI of \ccsk $(\kprocset, \klabelset, \r{fb}[\alpha[k]], \ind)$ as well, the pen-and-paper 
proofs~\cite[Propositions D.6 and D.8A]{APU24} essentially transfer properties from \pccsk to \ccsk by using the \bel{code/bijection}[bijection] between their transitions \(\base{\cdot} = {(\prov{\cdot}})^{-1}\) (\autoref{def-ccsk-pccsk-bijection}).
Considering that the only notion of independence we have for \ccsk comes from \pccsk (\autoref{def-ind-dep-ccsk}), we have to follow the same strategy.

Analogously to \pccsk, the mechanisation starts with \bel{code/ccsk/axioms/defs-and-properties.bel}[auxiliary definitions and properties] needed to state and prove the axioms and the simplified definition of \bel{code/ccsk/axioms/events.bel}[events][3][17] (\autoref{def-event-simplified}).
Moreover, transferring properties between LTSIs
requires \bel{code/bijection/lemmas-lifting.bel}[dedicated proofs] that are mostly routine. 
One difficulty was proving that \(t \eveqt u \) implies \(\lblof{\prov{t}} = \lblof{\prov{u}}\), \ie that  
\(\prov{\cdot}\) maps \ccsk transitions in the same event into \pccsk transitions with the same proof label.
Note that is not true for transitions belonging to different events, even with identical label: \eg the \ccsk transitions \(a.b + a \r{f}[a][k] a[k].b + a\) and \(a.b + a \r{f}[a][k] a.b + a[k]\) have different proof labels when mapped into \pccsk (\(\lplusl a[k]\) and \(\lplusr a[k]\)) .

This result, implicitly required to transfer the proofs of \PCI and \IRE from \pccsk to \ccsk, is implied by the following proposition, which more generally states that 
event equivalences %
are preserved by the bijection:

\begin{proposition}[Events are stable under bijection]%
    \label{prop-eveqt-stable}
~
\begin{multicols}{2}
    \noindent
		\begin{enumerate}
			\item
			For all \(t\), \(u\) in \pccsk, \bel{code/bijection/lemmas-lifting.bel}[if \(t \eveqt u\) then \( \base{t} \eveqt \base{u}\)][113][137]; \label{prop-eveqt-stable-1}
			\item
			For all \(t\), \(u\) in \ccsk, \bel{code/bijection/lemmas-lifting.bel}[if \(t \eveqt u \) then \( \prov{t} \eveqt \prov{u}\)][146][217] . \label{prop-eveqt-stable-2}
		\end{enumerate}
	\end{multicols}

\end{proposition}

Considering that this new result is not trivial, we discuss its Beluga proof with more detail than in the rest of this paper.
For \itemref{prop-eveqt-stable-2}, alternative approaches are certainly possible, but they would most likely be less computationally efficient.

\begin{figure}%
    \centering
    
    \input{figures/event-stable-bij.tex}
    \Description[short description]{long description}%
    \caption{Transitions in the proof of \autoref{prop-eveqt-stable}~\itemref{prop-eveqt-stable-2}, in \ccsk (on the left) and \pccsk (on the right).%
    } 
    \label{fig:proof-prop-eveqt-stable}
\end{figure}

\begin{proof}%
	\itemref{prop-eveqt-stable-1}
	By Definitions~\ref{def-ind-dep-pccsk} and \ref{def-ind-dep-ccsk}, \(\ind\) is preserved by \(\base{\cdot}\), and since \bel{code/bijection/lemmas-lifting.bel}[\(\base{\cdot}\) maps \pccsk transitions with the same proof label into \ccsk transitions with the same keyed label][79][102], the conclusion follows.
	
	\itemref{prop-eveqt-stable-2}
	Consider $t:P \r{fb}[\alpha][k] Q$, $u:P \r{fb}[\beta][m] R$, $u': Q \r{fb}[\beta][m] S$ and $t':R \r{fb}[\alpha][k] S$ with $t \ind u$, as in \autoref{def-event-simplified} and pictured in \autoref{fig:proof-prop-eveqt-stable}.
    Applying \(\prov{\cdot}\) yields transitions $\prov{t}:P \pr{fb}[\theta_1] Q$, $\prov{u}:P \pr{fb}[\theta_2] R$, $\prov{(u')}: Q \pr{fb}[\theta_2'] S$ and $\prov{(t')}:R \pr{fb}[\theta_1'] S$ with $\theta_1 \lind \theta_2$, but in general $\theta_i \neq \theta_i'$, for $i=1, 2$.
	Applying \SP gives transitions $Q \pr{fb}[\theta_2] S'$ and $R \pr{fb}[\theta_1] S'$, that can be mapped via \(\base{\cdot}\) to \ccsk transitions $u'': Q \r{fb}[\beta][m] S'$ and $t'': R \r{fb}[\alpha][k] S'$.
	Note, however, that there is no guarantee that \(S' = S\) at this point.	
	An \bel{code/ccsk/axioms/defs-and-properties.bel}[auxiliary lemma][250][1990]\footnote{As we detail in the \bel{overview.md\#proposition-67-events-are-stable-under-bijection}[overview file], the proof of this lemma spans \bkloc{1.7}, due to the need to perform case analysis on four combined \ccsk transitions.}, proved by induction on the length of the derivation trees of $t', u', t''$ and $u''$, establishes that the four transitions form a degenerate commutative diagram and that the processes $S$ and $S'$ coincide.
    By uniqueness of derivation trees for \pccsk transitions (\autoref{lem-derivation-uniqueness-ccskp}), it follows that $\prov{(t')} = \prov{(t'')}$ and $\prov{(u')} = \prov{(u'')}$, and therefore $\theta_i = \theta_i'$, for $i=1, 2$.
	Since their proof labels match pairwise, $\prov{t}$, $\prov{u}$, $\prov{(u')}$ and $\prov{(t')}$ form a commuting square as in \autoref{def-event-simplified}, and hence $\prov{t} \eveqt \prov{u}$.
\end{proof}

\begin{proposition}[{\bel{code/ccsk/axioms/}[The LTSI of \ccsk is pre-reversible]}]%
	\label{prop-pre-rev-ccsk}
	\SP, \BTI, \WF, and \PCI hold for the LTSI of \ccsk.%
\end{proposition}

The proofs of \bel{code/ccsk/axioms/sp.bel}[\textbf{SP}][9][30], \bel{code/ccsk/axioms/bti.bel}[\textbf{BTI}][3][33] and \bel{code/ccsk/axioms/wf.bel}[\textbf{WF}][3][24] are routine: as the axioms hold for \pccsk, it is easy to show that they hold for \ccsk using the bijection and the definitions of independence %
(Definitions~\ref{def-ind-dep-pccsk} and \ref{def-ind-dep-ccsk}).
As an example, we discuss the proof of \SP in more detail: we start by letting $t$ and $u$ be coinitial independent transitions in \ccsk.
By \autoref{def-ind-dep-ccsk}, $\lblof{\prov{t}} \lind \lblof{\prov{u}}$. By applying \SP in \pccsk, we obtain cofinal transitions $t'$ and $u'$, with the same directions and proof labels as $\prov{t}$ and $\prov{u}$ respectively, closing the commutative diagram; $\base{(t')}$ and $\base{(u')}$ close the corresponding diagram in \ccsk.
Since \bel{code/bijection/lemmas-lifting.bel}[\(\base{\cdot}\) maps \pccsk transitions with the same proof label into \ccsk transitions with the same keyed label][79][102], and since \bel{code/bijection/lemmas-lifting.bel}[ it maps pairs of transitions with the same direction into pairs of transitions with the same direction][4][25], $\base{(t')}$ and $\base{(u')}$ are the desired transitions that prove \SP for \ccsk.

The proof of \bel{code/ccsk/axioms/pci.bel}[\textbf{PCI}][3][34] is made easy by \autoref{prop-eveqt-stable}: given $t:P \r{fb}[\alpha][k] Q$, $u:P \r{fb}[\beta][m] R$, $u': Q \r{fb}[\beta][m] S$ and $t':R \r{fb}[\alpha][k] S$ with $t \ind u$, by \autoref{def-event-simplified} we have that $t \eveqt t'$ and $u \eveqt u'$.
By \autoref{prop-eveqt-stable}, since transitions in the same event have the same label, we infer that $\lblof{\prov{t}} = \lblof{\prov{t'}}$ and $\lblof{\prov{u}} = \lblof{\prov{u'}}$, and thus $u' \ind \rev{t}$ by \autoref{def-ind-dep-ccsk}.

Following the same approach, \ie leveraging the bijection between \pccsk and \ccsk and \autoref{prop-eveqt-stable} to transfer the properties of the LTSI of \pccsk (\autoref{prop-other-hold-pccsk}) to \ccsk, we can prove \bel{code/ccsk/axioms/id.bel}[\textbf{ID}][20][67], \bel{code/ccsk/axioms/ire.bel}[\textbf{IRE}][3][14], \bel{code/ccsk/axioms/cire.bel}[\textbf{CIRE}][3][16],  \bel{code/ccsk/axioms/rpi.bel}[\textbf{RPI}][3][13] and  \bel{code/ccsk/axioms/ed.bel}[\textbf{ED}][3][21] for \ccsk.
Note in particular that this approach allows us to prove \ED, although \FLD does not hold for \ccsk: \eg two different forward transitions labelled \(a[k]\) can have \(a + a\) for source. 
We can then easily conclude:

\begin{theorem}%
	\label{thm-axioms-hold-ccsk}
    The axiomatic theory holds for the LTSI of \ccsk.
\end{theorem}

\subsection{True-Concurrency Relations for \texorpdfstring{\pccsk}{CCSKP} via Independence and Dependence}%
\label{ssec:true-concu-relation-pccsk}

We have seen in the previous sections that the LTSIs for \pccsk and \ccsk satisfy all the axioms needed for the new alternative characterisations of true-concurrency relations from \autoref{sec:CCC}.
This section discusses more precisely how those results translate to the specific case of \pccsk.
We also discuss some of the limitations of the independence and dependence relations when it comes to capturing true-concurrency relations.
While all the results below apply equally to \ccsk, we focus on \pccsk to ease the presentation.

First, note that the relations and complementarity result on \pccsk transitions (\autoref{prop-complementarity transitions}) can be lifted to connected \emph{events}, since the LTSI of \pccsk satisfies \IRE (\autoref{prop-other-hold-pccsk}).
We first add dependence to our set of relations on events (\autoref{def-ind sdep event}):

\begin{definition}[Dependence of connected events]%
	\label{def-dep-on-events}
	Let \(e_1\), \(e_2\) be connected events in the LTSI of \pccsk, we let $e_1 \dep e_2$ if there are transitions $t_1 \in e_1$ and $t_2 \in e_2$ such that $t_1 \dep t_2$.
\end{definition}	

\begin{remark}%
	\label{rem:axiom-complementarity}
	Note that we are providing a direct definition of dependence for \pccsk, instead of defining it as the complement of independence on connected transitions (\autoref{notation-complement}) as done in the axiomatic approach.
	However, complementarity for events (\autoref{prop-complementarity events}) below proves that both definitions match for \pccsk. 
\end{remark}

Quantifiers in the definitions of independence and dependence on events can go from existential to universal:

\begin{proposition}%
	\label{prop-ind-dep-on-events}
	Let \(e_1\), \(e_2\) be connected events in the LTSI of \pccsk.
	\begin{enumerate}
		\item $e_1 \ind e_2$ iff $t_1 \ind t_2$ for all $t_1 \in e_1$, $t_2 \in e_2$;
		\item $e_1 \dep e_2$ iff $t_1 \dep t_2$ for all $t_1 \in e_1$, $t_2 \in e_2$.
	\end{enumerate}
\end{proposition}

\begin{proof}
	Both (\(\Leftarrow\)) directions are immediate by Definitions~\ref{def-ind sdep event} and \ref{def-dep-on-events}, so we prove only the other direction.
	\begin{enumerate}
		\item (\(\Rightarrow\)) By \autoref{def-ind sdep event}, there exist $t_1 \in e_1$, $t_2 \in e_2$ such that $t_1 \ind t_2$. 
		For all \(t'_1 \in e_1\), \(t'_2 \in e_2\), \(t_1' \sqeqt  t_1\) and \(t_2' \sqeqt t_2\), by \IRE we obtain \(t_1' \ind t_2'\) as desired.
		
		\item (\(\Rightarrow\)) By \autoref{def-dep-on-events}, there exist $t_1 \in e_1$, $t_2 \in e_2$ such that $t_1 \dep t_2$. 
		If there existed \(t'_1 \in e_1\) and \(t'_2 \in e_2\) with \(t_1' \ind t_2'\), then by \IRE we would obtain \(t_1 \ind t_2\), which would violate complementarity on transitions (\autoref{prop-complementarity transitions}). \qedhere
	\end{enumerate}
\end{proof}

We then immediately obtain as a corollary of Propositions~\ref{prop-complementarity transitions} and \ref{prop-ind-dep-on-events}:

\begin{corollary}[Complementarity for events]%
	\label{prop-complementarity events}
	Let \(e_1\), \(e_2\) be connected events in the LTSI of \pccsk.
	Then exactly one of $e_1 \ind e_2$ and $e_1 \dep e_2$ holds.
\end{corollary}

Since \CIRE holds, core independence and concurrency (\autoref{def-event relations})
coincide (\autoref{prop-IRE co coind}). Among adjacent transitions and thus adjacent events, concurrency is independence, or simply existence of non-degenerate commuting squares. If forward transitions are composable and dependent, then they are causally dependent
by \autoref{lem-adj non-ind vs ccc}~\itemref{lem-adj non-ind vs ccc two} and by complementarity of transitions and events for \pccsk. In fact, events of such transitions 
are in the immediate predecessor relationship, as we show below with \autoref{lem-immed pred composable}.
Correspondingly, by \autoref{lem-adj non-ind vs ccc}~\itemref{lem-adj non-ind vs ccc one} and by complementarity of transitions and events, distinct coinitial forward dependent transitions are in initial conflict.

\begin{figure}%
\centering

	\subfloat[%
    \((a.b\Par \out{b}.c)\bs b \): %
    \({[t_1] \ind [t_3]}\),%
    	\({[t_1] < [t_3]}\)%
		 ]{\label{fig:example-1}%
			\input{figures/example-ltsi-1}%
	}\qquad
	\subfloat[\(a. (b \Par c.d)\): %
    \({[t_1] \tdep [t_2]}\), \({[t_1] \ci [t_2]}\)
    ]{\label{fig:example-3}%
		\input{figures/example-ltsi-3}%
	}
    \qquad
   	\subfloat[\(((a \Par a.b)\Par \out{a}) \bs  a\): \({[t_1] \cf [t_3]}\), \({[t_1] \ind [t_3]}\) ]{\label{fig:example-2}%
        \input{figures/example-ltsi-2}\qquad%
    }
   	\Description[short description]{long description}%
	\caption{LTSIs for Examples~\ref{def-ind-dep-ccsk-causal}, \ref{conc-tran-clo}, \ref{causal-ind} and \ref{ex-conflict-ind}.%
	} 
	\label{fig:LTSIs}

\end{figure}

Since the LTSI of \pccsk is pre-reversible and satisfies \IRE and \RPI, \autoref{lem-coind ind adj} and 
\autoref{prop-IRE co coind} imply:

\begin{proposition}[Core Independence and Concurrency Match Independence on Adjacent Events]%
	\label{prop-coind-adj-ind}
	Let \(e_1\), \(e_2\) be events in the LTSI of \pccsk.
	Then $e_1 \coind e_2$ iff $e_1 \co e_2$ iff \(e_1\), \(e_2\) are adjacent and $e_1 \ind e_2$.
\end{proposition}

It is important to observe that adjacency is required for the \enquote{if} part of the implication.
Indeed, independence based on proof labels for \pccsk is more general than core independence for non-adjacent but connected events: 

\begin{example}[Independence does not imply core independence on non-adjacent events]%
	\label{def-ind-dep-ccsk-causal}
	Consider $(a.b\Par \out{b}.c) \bs  b$ and: %
	\[
	\rlap{$\underbrace{\phantom{(a.b\Par \out{b}.c)\bs b \pr{f}[\lmidl a][k] (a[k].b \Par \out{b}.c) \bs b}}_{t_1}$} %
	(a.b\Par \out{b}.c) \bs b \pr{f}[\lmidl a][k] %
	\overbrace{(a[k].b \Par \out{b}.c) \bs b \pr{f}[\cpair{b[l]}{\out{b}[l]}] \phantom{(a[k].b[l]\Par \out{b}[l].c) \bs  b}}^{t_2}%
	\nhphantom{$(a[k].b[l]\Par \out{b}[l].c) \bs  b$}
	\underbrace{(a[k].b[l]\Par \out{b}[l].c) \bs  b \pr{f}[\lmidr c][n] (a[k].b[l]\Par \out{b}[l].c[n]) \bs  b}_{t_3}\text{.}
	\]
    
	Those transitions have to be executed in this exact order, as concisely represented in \autoref{fig:example-1}, 
	yet \(t_1 \ind t_3\):
	\begin{align*}
	\lblof{t_1} = \lmidl a[k] \lind \lmidr c[n] = \lblof{t_3} \tag{By P\(^2_k\) (\autoref{fig:pccsk-dep-ind}), since producing \(t_3\) requires to have $k \neq n$}
	\end{align*}
	
	However, $[t_1]<[t_3]$ and polychotomy (\autoref{prop-poly}) yields that 
	$[t_1] \coind [t_3]$ does not hold\footnote{This can also be observed considering that \(t_1\) and \(t_3\) are neither coinitial nor event equivalent to coinitial transitions.}.
\end{example}

Similarly to \autoref{prop-coind-adj-ind}, immediate predecessor (\autoref{def-immed pred}) is non-independence on composable events:

\begin{proposition}[Immediate Predecessor Matches Dependence on Composable Events]%
	\label{lem-immed pred composable}
	Let \(e_1\), \(e_2\) be forward events  in the LTSI of \pccsk.
	Then $e_1 \ip e_2$ iff $e_1$ is composable with $e_2$ and $e_1 \dep e_2$.
\end{proposition}

\begin{proof}
	Since non-independence coincides with dependence on composable events (\autoref{prop-complementarity events}), we obtain both implications by \autoref{lem-IRE RPI ip comp}, which we can use since the LTSI of \pccsk is pre-reversible and satisfies \IRE, \RPI.
\end{proof}

Since the LTSI for \pccsk is pre-reversible and satisfies \IRE and \RPI, we can apply Propositions~\ref{prop-chain ip} and \ref{lem-immed pred composable}
to obtain an alternative label-based characterisation of the causality relation $<$: this will be useful in proving that the semantic ordering of events corresponds to a static ordering on keys for \pccsk and \ccsk processes (\autoref{thm-ordering-event-key}) in \autoref{sec:key-based}.

Since \pccsk further satisfies \ED, we can apply \autoref{thm-conf is conf dep} to obtain a label-based characterisation 
of conflict $\Cf$
via global conflict $\Cfg$ (\autoref{def-gl conf dep}) and thus initial conflict $\Cfi$ (\autoref{def-conf dep}) and causality.
For this we require the result below, which follows by \autoref{prop-complementarity events}, that relates initial conflict and dependence.

\begin{proposition}[Initial Conflict Matches Dependence on Coinitial Events]%
	\label{lem-iinit conflict}
Let \(e_1\), \(e_2\) be forward events  in the LTSI of \pccsk.
Then $e_1 \Cfi e_2$ iff $e_1, e_2$ are coinitial and $e_1 \dep e_2$.
\end{proposition}

We can transfer these results regarding causality and conflict from \pccsk to \ccsk using their bijection.

Propositions~\ref{prop-chain ip} and \ref{lem-immed pred composable} show that causality can be expressed as chains of (immediate predecessor, event) pairs or chains of pairs of composable events that are dependent, respectively.
Hence, it is natural to ask whether the transitive closure of $\dep$, denoted by $\tdep$, can play the role of the chain and capture causality.
We show below that $\tdep$ is necessary but not sufficient; hence not suitable for representing causality.

\begin{proposition}%
	\label{leqcftsdep}
	Let \(e_1\), \(e_2\) be connected forward events in the LTSI of \pccsk. %
	\begin{multicols}{2}
        \noindent
		\begin{enumerate}
			\item
			If $e_1 \leq e_2$ then $e_1 \tdep e_2$; \label{leqcftsdep-1}
			\item
			If $e_1 \Cf e_2$ then $e_1 \tdep e_2$. \label{leqcftsdep-2}
		\end{enumerate}
	\end{multicols}
\end{proposition}

\begin{proof}
    \begin{enumerate}
    \item By Propositions~\ref{prop-chain ip} and \ref{lem-immed pred composable}.%
    \item %
    By \autoref{thm-conf is conf dep}, $e_1\Cfg e_2$, so by \autoref{def-gl conf dep} there exist forward events $e_3$, $e_4$ such that $e_3 \Cfi e_4$, $e_3 \leq e_1$ and $e_4 \leq e_2$.
    By \itemref{leqcftsdep-1}, \(e_3 \tdep e_1\) and \(e_4 \tdep e_2\), and by \autoref{lem-iinit conflict}, $e_3 \Cfi e_4$ implies \(e_3 \dep e_4\), so that \(e_1 \tdep e_2\). \qedhere
    \end{enumerate}
\end{proof}
The converse version of \autoref{leqcftsdep} stipulating that two events in the transitive closure of dependence are either causally related or in conflict would not hold.
Indeed, $\tdep $ captures, in addition to causality and conflict, some core independence and independence, on adjacent and non-adjacent transitions:

\begin{example}[Transitive dependence and core independence are compatible]%
	\label{conc-tran-clo}    
    Consider the following transitions:  
   	\[
    \rlap{$\underbrace{\phantom{a[k]. (b \Par c.d) \pr{f}[\lmidl b][m] a[k].(b[m] \Par c.d)}}_{t_1}$} %
    a[k]. (b \Par c.d) \pr{f}[\lmidl b][m] %
    \overbrace{a[k].(b[m] \Par c.d) \pr{f}[\lmidr c][l] \phantom{a[k].(b[m] \Par c[l].d)}}^{t_2}%
    \nhphantom{$a[k].(b[m] \Par c[l].d)$}
    \underbrace{a[k].(b[m] \Par c[l].d) \pr{f}[\lmidr d][n] a[k].(b[m] \Par c[l].d[n])}_{t_3}\text{.}
    \]
    
	\autoref{fig:example-3} presents all forward transitions from \(a. (b \Par c.d)\), where additional transitions are such that:
	\begin{multicols}{4}
		\noindent
		\[\lblof{t_0} =  a[k]\]
		\[\lblof{t_2} = \lblof{t_2' }\]
		\[\lblof{t_1} = \lblof{t_1'} = \lblof{t_1''}\]
		\[\lblof{t_3} = \lblof{t_3'}\text{.}\]
	\end{multicols}

    We have at the same time
    \begin{align*}
        \lblof{t_1} \lind \lblof{t_2} \text{ by P\(^2_k\) (\autoref{fig:pccsk-dep-ind}), and \(t_1\), \(t_2\) are adjacent} & \implies [t_1] \ci [t_2] \text{ by \autoref{prop-coind-adj-ind}}\\
        \lblof{t_1} \sdep \lblof{t_0} \sdep \lblof{t_2} \text{ by A\(^1\) and A\(^2\) (\autoref{fig:pccsk-dep-ind})} %
        & \implies [t_1] \tdep [t_2]
    \end{align*}

    Similarly, for the non-adjacent transitions \(t_1\) and \(t_3\), we have $[t_1]\coind [t_3]$ (since $t_1\sim t_1' \ind t_3$) and yet $[t_1]\tdep [t_3]$.
\end{example}

One can also find situations where conflicting events (\autoref{causal-ind}), or causally dependent events (\autoref{ex-conflict-ind}), have independent labels and are in the \(\tdep\) relationship at the same time.
We finish this section with two such examples: %

\begin{example}%
	\label{causal-ind}
     Consider the following three forward transitions:
    \begin{align*}
        t_1 : ((a \Par a.b)\Par \out{a}) \bs a & \pr{f}[\cpair{\lmidl a[k]}{\out{a}[k]}] ((a[k] \Par a.b) \Par \out{a}[k]) \bs a \\ 
        \rlap{$\underbrace{\phantom{%
                    ((a \Par a.b)\Par \out{a}) \bs a \pr{f}[\cpair{\lmidr a[k]}{\out{a}[k]}] ((a \Par a[k].b) \Par \out{a}[k]) \bs a %
            }}_{t_2}$} %
        ((a \Par a.b) \Par \out{a}) \bs a  & \pr{f}[\cpair{\lmidr a[k]}{\out{a}[k]}]  %
        \overbrace{ ((a \Par a[k].b) \Par \out{a}[k]) \bs a \pr{f}[\lmidl \lmidr b][m] ((a \Par a[k].b[m]) \Par \out{a}[k])}^{t_3}%
        \text{.}
    \end{align*}
    
    They are such that \([t_1] \ind [t_3]\) and \([t_1] \tdep [t_3]\) both hold, since 
    \[
    \begin{alignedat}{7}
    	\lblof{t_1} & =  \cpair{\lmidl a[k]}{\out{a}[k]} && \lind \lmidl \lmidr b[n] = \lblof{t_3} \\
    	\lblof{t_1} & =  \cpair{\lmidl a[k]}{\out{a}[k]} && \sdep \cpair{\lmidr a[m]}{\out{a}[m]} \sdep \lmidl \lmidr b[n] = \lblof{t_3} 
    \end{alignedat}
    \]    
    and it is also easy to see that \([t_1] \cf [t_3]\), as pictured in \autoref{fig:example-2}.
\end{example}

\begin{example}%
	\label{ex-conflict-ind}
	
	Reusing the process $(a.b\Par \out{b}.c) \bs  b$ from \autoref{def-ind-dep-ccsk-causal}, pictured in \autoref{fig:example-1}, one can further observe that $[t_1]\leq [t_3]$ and also $[t_1] \tdep [t_3]$:
	\[	
		\lblof{t_1} =  \lmidl a[k] \sdep \cpair{b[m]}{\out{b}[m]} \sdep \lmidr c[n] = \lblof{t_3}
	\]
\end{example}

%% file: figures/event-stable-bij.tex
    \begin{tikzpicture}[
    ]

    \node (P) {\(P\)};
    \node [above right= of P] (Q) {\(Q\)};
    \node [below right= of P] (R) {\(R\)};
    \node (S) at ($(P)!2!($(Q)!0.5!(R)$)$)  {\(S\)}; %
    \node (S') [right= .4cm of S] {\(S'\)};
    \draw[fb] (P) -- node (PQ) [above = .2, pos = .2]{\(\alpha[k]\)} (Q);
    \draw[fb] (P) -- node (PR) [below = .3, pos = .2]{\(\beta[m]\)} (R);
    \draw[fb] (R) -- node (RS) [above = .2, pos = .2]{\(\alpha[k]\)} (S);
    \draw[fb] (Q) -- node (QS) [below = .3, pos = .2]{\(\beta[m]\)} (S);
    
    \draw[fb, dashed, shorten >=8pt] (R) -- node (RS'')[below = .1, pos = .8]{\(\alpha[k]\)} (S');
    \draw[fb, dashed, shorten >=8pt] (Q) -- node (QS'')[above = .15, pos = .8]{\(\beta[m]\)} (S');
    
    \node (equal) at ($(S)!.5!(S')%
    $) {\(=\)}; %
        \begin{scope}[on background layer]
     \draw[draw=ord,fill=tsdep!30,rounded corners=0.2cm]
    (S.north west)  -- (S'.north east) -- node [above right = .1 and -.5] {\emph{\bel{code/ccsk/axioms/defs-and-properties.bel}[Auxiliary lemma][250][1990]}} (S'.south east) -- (S.south west) -- cycle;
        \end{scope}
    
    \begin{scope}[xshift=7cm]
    \node (P) {\(P\)};
    \node [above right= of P] (Q) {\(Q\)};
    \node [below right= of P] (R) {\(R\)};
    \node at ($(P)!2!($(Q)!0.5!(R)$)$) (S) {\(S\)}; %
    \node [right= .7cm of S] (S') {\(S'\)};
    \draw[pfb] (P) -- node (PQ') [above = .05, pos = .25]{\(\theta_1\)} (Q);
    \draw[pfb] (P) -- node (PR') [below = .05, pos = .25]{\(\theta_2\)} (R);
    \draw[pfb] (R) -- node (RS') [above = .1, pos = .4]{\(\theta_1'\)} (S);
    \draw[pfb] (Q) -- node (QS') [below = .1, pos = .4]{\(\theta_2'\)} (S);
    
    \draw[pfb, dashed] ($(R) + (.4, 0)$) -- node (RS''') [below = .0, pos = .8]{\(\theta_1\)} (S'); %
    \draw[pfb, dashed] ($(Q) + (.4, 0)$) -- node (QS''')[above = .0, pos = .8]{\(\theta_2\)} (S');  %
    
    \begin{scope}[on background layer]
     \draw[draw=tsdep,fill=tsdep!30,rounded corners=0.2cm]
    (Q.north)  --  node[right = .3, yshift=4]  {\emph{Obtained by \SP}} ($(S'.north east)+ (0, .5)$) -- ($(S'.south east)+(0, -.5)$) -- (R.south)  -- cycle;
    
	    \node [%
        draw=ord, fill=ord!30, %
         fit= (P) (Q) (S) (R), rotate fit=45,
        rectangle, label={[xshift=4, yshift=-4, align=left]above:\emph{Obtained by}\\\emph{proof enrichment}},
        inner xsep = -6, inner ysep = -9, 
        yshift=4pt,
        xshift=-5pt,
        rounded corners=0.2cm
        ] {};        
    \end{scope}
    \end{scope}
    
    \draw[f, draw=ord, shorten >=-2pt] (PQ) edge[out=50,in=100] %
    node[above] {\(\prov{\cdot}\)} (PQ');

    \draw[f, draw=ord, shorten <=3pt, shorten >=3pt] (QS) edge[out=50,in=80] %
    node[above, pos=.35] {\(\prov{\cdot}\)} (QS');

    \draw[f, draw=ord, shorten >=-2pt] (PR) edge[out=-40,in=-100] %
    node[above] {\(\prov{\cdot}\)} (PR');

    \draw[f, draw=ord, shorten <=1pt, shorten >=3pt] (RS) edge[out=-45,in=-80] %
    node[above, pos=.35] {\(\prov{\cdot}\)} (RS');

    \draw[f, draw=tsdep] (RS''') edge[out=-130,in=-30] %
    node[above, pos=.73] {\(\base{\cdot}\)} (RS'');
    
    \draw[f, draw=tsdep] (QS''') edge[out=130,in=30] %
    node[above, pos=.75] {\(\base{\cdot}\)} (QS'');   
\end{tikzpicture}

%% file: figures/example-ltsi-1.tex
\begin{tikzpicture}[
	x={(1.1, 0)},  %
	y={(0, 1))}, %
    baseline,
	baseline=3.9em,
	]
	
	\node[inner sep=0, anchor=east, minimum height=8em] %
		 (a) at (0, 0) {};
	\node (b) at (1.5, 0){};
    \node (c) at (3, 0){};
    \node (d) at (4.5, 0){};

	\draw[pf] (a) -- node[above, pos=.5]{\(t_1\)} (b);
	\draw[pf] (b) -- node[above, pos=.5]{\(t_2\)} (c);
	\draw[pf] (c) -- node[above, pos=.5]{\(t_3\)} (d);
\end{tikzpicture}

%% file: figures/example-ltsi-3.tex
\begin{tikzpicture}[
	x={(1, 0)},  %
	y={(0, 1)}, %
	baseline,
	anchor=base
	]
	\node (z) at (-1, -1){};	
	\node (a) at (0, -1){};
	\node (b) at (1, -2){};
	\node (c) at (1, 0){};
	\node (d) at (2, -1){};
	\node (e) at (3.41, -1){};
	\node (f) at (2.41, -2){};
	
	\draw[pf] (z) -- node[above, pos=.5]{\(t_0\)} (a);
	\draw[pf] (a) -- node[left, pos=.5, yshift=-2 %
			]{\(t_2'\)} (b);
	\draw[pf] (a) -- node[left, pos=.5]{\(t_1\)} (c);
	\draw[pf] (b) -- node[left, pos=.5, yshift=2 %
			]{\(t_1'\)} (d);
	\draw[pf] (c) -- node[above, pos=.5]{\(t_2\)} (d);
	\draw[pf] (d) -- node[above, pos=.5]{\(t_3\)} (e);
	\draw[pf] (b) -- node[below, pos=.5]{\(t_3'\)} (f);
	\draw[pf] (f) -- node[right, pos=.5, xshift=2 %
			]{\(t_1''\)} (e);	
\end{tikzpicture}

%% file: figures/example-ltsi-2.tex
\begin{tikzpicture}[
	x={(1, 0)},  %
	y={(0, 1)}, %
	baseline=2.8em,
	anchor=base,
	]
	\node[minimum height=8em] (a) at (0, 0){};
	\node (b) at (1, .6){};
	\node (c) at (1, -.6){};
	\node (d) at (2, -.6){};
	
	\draw[pf] ($(a)+(0, .1)$) -- node[above, pos=.4]{\(t_1\)} (b);
	\draw[pf] ($(a)+(0, -.1)$) -- node[below, pos=.4]{\(t_2\)} (c);
	\draw[pf] (c) -- node[above, pos=.4]{\(t_3\)} (d);
\end{tikzpicture}

%% file: sections/key-based.tex
\section{Causality and Core Independence from Key-based Events}%
\label{sec:key-based}

The results in this section rely on the keys manipulated by transitions and processes of \pccsk and \ccsk.
A basic design decision in both systems is that fresh keys are chosen when computing forwards, so past events will have different keys.
We shall see that we can decide whether transitions belong to the same event simply using keys (\autoref{prop-event_coincide})---a continuation of \autoref{prop:eveqt evleqt}.
Moreover, it is possible to decide whether an event was caused by another, or if they are core independent, using this information statically stored in the processes (Theorems~\ref{thm-ordering-event-key} and \ref{thm-independent-event-key}).
As in \autoref{ssec:true-concu-relation-pccsk}, all the results below apply to \ccsk, but we focus on \pccsk to ease the presentation: as usual, either the same argument applies, or the bijection from \autoref{def-ccsk-pccsk-bijection} allows us to transfer the results between the systems.

By definition, for any LTSI, if $t_1 \eveqt t_2$ then $t_1$ and $t_2$ must have the same labels (keys included).
However, the converse is false: consider \eg $t_1: a[m].b \pr{f}[b][k] a[m].b[k]$ and $t_2: a[n].b \pr{f}[b][k] a[n].b[k]$.
These transitions have the same proof labels, but are not the same event, since they are caused by different events.
An alternative way of reaching the same conclusion, using the label-based equivalence (\autoref{def-event-label-equivalence}), is to observe that their targets cannot be connected by a path in which \(b[k]\) does not occur.
We will refine this observation by showing that any two transitions with the same key are equivalent, 
provided that there is a path connecting their target processes without using their key.
This sharpens event label equivalence by focusing on keys only, while still being more direct
than than the general (\autoref{def-event-general}) or simplified (\autoref{def-event-simplified}) definition of events, as it neither relies on a notion of independence nor requires constructing \enquote{ladders} of commuting squares~\cite[Figure 5]{LPU24}.

\begin{definition}[Event key equivalence]%
	\label{def-event-key-equivalence}
	Two forward transitions $t_1:X_1 \pr{f}[\theta_1] X'_1$ and $t_2:X_2 \pr{f}[\theta_2] X'_2$ with \(\kay{\theta_1} = \kay{\theta_2}\) are \emph{event key equivalent} ($t_1 \evkeqt t_2$) if there is a path $r: X'_1 \pr{bf}^* X'_2$ such that $\kay{\theta_1}$ does not occur in any transition of \(r\).
	We extend to backward transitions by letting $t_1 \evkeqt t_2$ iff $\rev{t_1} \evkeqt \rev{t_2}$.
\end{definition}

\begin{figure}
    \input{figures/example-transition2.tex}
    \Description[short description]{long description}%
    \caption{Illustrating event key equivalence using processes reachable from \(a \Par b[k]\) (\autoref{ex:key-ev-equiv}).}%
    \label{fig:example-transitions2}
\end{figure}
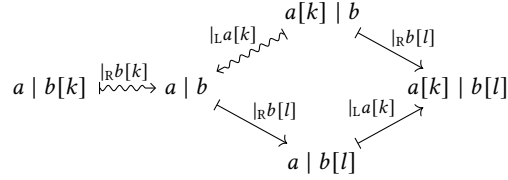

\begin{example}%
    \label{ex:key-ev-equiv}
    Consider the process $a\Par b[k]$ and some of its transitions in \autoref{fig:example-transitions2}.
    The two transitions with the key \(l\) are event key equivalent, as there is a path from $a \Par b[l]$ to $a[k] \Par b[l]$ which does not involve $l$.	
    However, all paths from $a \Par b$ to $a[k] \Par b[l]$ involve at least one transition with key \(k\), and hence the transition with proof label \(\lmidr b[k]\) and the forward transition with proof label \(\lmidl a[k]\) are not event key equivalent.
\end{example}

It is clear in the previous example that the transitions with proof label \(\lmidr b[k]\) and \(\lmidl a[k]\) cannot be equivalent according to the general definition of events either, since they have different labels.
We now set the stage to prove that both equivalences indeed coincide,
first by stating formally a comment from \autoref{rem:ind-irreflexivity}:

\begin{lemma}[Independence implies different keys]%
	\label{lem-ind-diff-keys}
	 For all \(t_1\), \(t_2\), $t_1 \ind t_2$ implies $\key{t_1} \neq \key{t_2}$.%
\end{lemma}

\begin{proof}
	From the rules for $\lind$ in \autoref{fig:pccsk-dep-ind}, if $\theta_1 \lind \theta_2$ then $\kay{\theta_1} \neq \kay{\theta_2}$.
	Hence, if $t_1 \ind t_2$ then $\key {t_1} \neq \key {t_2}$.
\end{proof}

\begin{lemma}[Backward key determinism]%
	\label{lem-bwd_key_det}
	If \(t_1\) and \(t_2\) are coinitial backward transitions with $\key{t_1} = \key{t_2}$, then $t_1 = t_2$.
\end{lemma}

\begin{proof}
	Suppose that \(t_1\) and \(t_2\) are two coinitial backward transitions with the same key %
	and that $t_1 \neq t_2$.
	By \BTI %
	(\autoref{prop-pre-reversible-hold-pccsk})
	we have $t_1 \ind t_2$.
	By \autoref{lem-ind-diff-keys} we have $\key {t_1} \neq \key {t_2}$, which is a contradiction.
\end{proof}

\begin{proposition}[Event equivalences coincide]\label{prop-event_coincide}
	For all \(t_1\), \(t_2\), 
	$t_1 \eveqt t_2$ iff $t_1 \evkeqt t_2$.
\end{proposition}

The proof, that we omit, follows very closely the proof of \autoref{prop:eveqt evleqt}, simply using backward key determinism in place of \BLD.

\autoref{thm-ordering-event-key} proves that the causal ordering on events $\leq$ from \autoref{def-event relations} can be computed by a static ordering on keys, whose definition follows.
A crucial lemma (\autoref{lem-ord}) has its proof facilitated by using event key equivalence.

\begin{definition}[Partial order on keys~\protect{\cite[Definition 3.1]{LaneseP21}}]%
	\label{def-partial-order-on-keys}
	The \emph{partial order \(\leq_{X}\) on \(\keys{X}\)} of a process \(X\) is the reflexive and transitive closure of \(\ord{X}\):%
	\begin{align*}
		\ord{\nil} & = \emptyset  & \ord{X + Y} &= \ord{X} \cup \ord{Y} \\
		\ord{\alpha.X} & = \ord{X} & \ord{X \Par  Y} &= \ord{X} \cup \ord{Y} \\
		\ord{X \bs a} & = \ord{X} & \ord{\alpha[n].X} & = \ord{X} \cup \{n < k \setst k \in \keys{X}\}
	\end{align*}
\end{definition}

The development below assumes that processes are reachable (\autoref{definition:reachable}), as in the rest of this paper.

\begin{lemma}%
	\label{lem:fwd keys}
	If \(X \pr{f}[\theta] X'\) then $\kay{\theta} \notin \keys X$ and $\keys{X'} = \keys X \union \{\kay{\theta}\}$.
\end{lemma}

\begin{proof}
	By induction on derivations using the forward transition rules in \autoref{fig:provedltsrulesccskfw}.
\end{proof}

\begin{lemma}%
	\label{lem:fwd key maximal}
	If \(t: X \pr{f}[\theta] X'\) then $\kay{\theta}$ is maximal under $\leq_{X'}$.
\end{lemma}

\begin{proof}
	By \autoref{lem:fwd keys}, $\kay{\theta}$ is introduced by the forward transition $t$.
	According to the forward transition rules of \autoref{fig:provedltsrulesccskfw}, keys are introduced only on prefixes of standard sub-processes: this is obtained by the side condition of act, and maintained through the derivation by the other rules.
	Hence we have $(\kay{\theta}, n) \notin \ord{X'}$ for any key $n$, and this is still true for \(n \neq \kay{\theta}\) when we take the reflexive and transitive closure to obtain $\leq_{X'}$.
\end{proof}

\begin{lemma}\label{lem:fwd ord}
	If \(X \pr{f}[\theta] X'\) and $m,n \in \keys X$ then $(m,n) \in \ord X$ iff $(m,n) \in \ord {X'}$.
\end{lemma}

\begin{proof}
	Letting \(\kay{\theta} = k\), by \autoref{lem:fwd keys}, the only difference between $X$ and $X'$ is that $k$ occurs in $X'$. %
	By the proof of \autoref{lem:fwd key maximal}, there are no pairs $(k,n)$ in $\ord {X'}$. By inspection of the clauses of \autoref{def-partial-order-on-keys}, in particular the clause for $\ord{\alpha[n].X}$, we see that $\ord{X'}$ is the same as $\ord X$ except that some pairs of the form $(n,k)$ may be added.
\end{proof}

\begin{proposition}%
	\label{prop:order-trans-pres}
	If \(t: X \pr{f}[\theta] X'\) and $m,n \in \keys{X}$ then $m \leq_X n$ iff $m \leq_{X'} n$.
\end{proposition}

\begin{proof}
     Take $m,n \in \keys X$.  If $m = n$ then there is nothing to show. 
     
	(\(\Rightarrow\)) 
    Suppose $m <_X n$.  Then there is a chain 
    \(\{(m = m_0, m_1), \cdots, (m_{i-1}, m_i = n)\} \subseteq \ord{X}\) for $i \geq 1$.
	By \autoref{lem:fwd ord}, all the elements of this chain are also included in $\ord{X'}$, so that $m <_{X'} n$ as required.
	
    (\(\Leftarrow\)) 
    Suppose $m <_{X'} n$.
    Then there is a chain %
    \(\{(m = m_0, m_1), \cdots, (m_{i-1}, m_i = n)\} \subseteq \ord{X'}\) for $i \geq 1$.
     But $m_0,\ldots,m_i$ are all different from $\kay{\theta}$, since $\kay{\theta}$ is maximal in $\leq_{X'}$ (\autoref{lem:fwd key maximal}) and $m,n \neq \kay{\theta}$ since %
     \(\kay{\theta} \notin \keys{X}\).
     Hence by \autoref{lem:fwd ord}, all the elements of the chain are also included in $\ord{X}$, and $m <_{X} n$ as required.
\end{proof}

\begin{definition}[Events in a process]%
	\label{def-events process alt1}
	Given a process \(X\), we let the \emph{events of $X$} be
	$\ev{X} = \{e \setst \cte(r,e) > 0\}$, where $r$ is any rooted forward-only path to $X$.
\end{definition}
There is always such a rooted forward-only path by \PL; and the following guarantees that \(\ev{X}\) is well-defined:

\begin{lemma}[\protect{\cite[Lemma 4.12]{LPU24}}]%
	\label{lem:event-count}
	Assume an LTSI is pre-reversible, and let $r$ and $s$ be coinitial and cofinal paths.
	Then for each event $e$ we have that $\cte(r,e) = \cte(s,e)$.
\end{lemma}

Note that, by definition, all the events in a process are \emph{forward} events.
Recall also that transitions in the same event \(e\) have the same label, and hence the same key, that we denote by \(\key{e}\).
We now prove the following properties, which will be useful later on for \autoref{def-event-key} and \autoref{lem-ord}.

\begin{lemma}[Event keys properties]\label{lem-event-key}
	For any 
	process $X$:
	\begin{enumerate}
		\item if $e_1,e_2 \in \ev{X}$ and $e_1 \neq e_2$, then $\key {e_1} \neq \key {e_2}$; \label{lem-event-key-1}
		\item $\{\key e \setst e \in \ev{X}\} = \keys X$. \label{lem-event-key-2}
	\end{enumerate}
\end{lemma}

\begin{proof}%
	\begin{enumerate}
		\item We consider any forward-only rooted path to \(X\), %
		all keys of transitions must be distinct by \autoref{lem:fwd keys}.
		\item 
		By induction on the cardinality of \(\ev{X}\), which we denote by \(\card{\ev{X}}\).
		If $\card{\ev{X}} = 0$ then $\ev X = \emptyset$ and $X$ is standard, \ie $\keys X = \emptyset$.
		If $\card{\ev{X}} > 0$, then consider a rooted forward-only path $r$ to $X$.
        It must be of the form \(r = r't\) with \(t : X' \pr{f}[\theta] X\).
		Then we have:
		
		\noindent
        \begin{minipage}[b]{0.83\textwidth}
		\begin{align*}
			\keys{X} & = \keys{X'} \union \{\kay{\theta}\} \tag{By \autoref{lem:fwd keys}}\\
					 & = \{\key e \setst e \in \ev{X'}\} \union \{\kay{\theta}\} \tag{By inductive hypothesis, since \(\card{\ev{X'}} < \card{\ev{X}}\)} \\
					 & = \{\key e \setst e \in \ev{X'}\} \union \{\key{[t]}\} \tag{By definition of \(\key{e}\)}\\	 
					 & = \{\key e \setst e \in \ev{X} \} \tag{Since \(\ev{X} = \ev{X'} \cup \{[t]\}\)}
		\end{align*}
		\end{minipage}%
		\hfill \begin{minipage}[b]{0.1\textwidth}
		\qedhere
		\end{minipage}

	\end{enumerate}
\end{proof}

One can also easily observe from the definition of conflict (\autoref{def-event relations}) that for all \(e_1, e_2 \in \ev{X}\), \(e_1 \ncf e_2\). 
However other relations have interesting connections to the events in a process and their keys, as we illustrate now, starting with the immediate predecessor relation \(\ip\) (\autoref{def-immed pred}), which can be used to infer partial order between keys:

\begin{lemma}[Order from events to keys]%
	\label{lem-composable sdep}
	Suppose %
	 $e_1,e_2 \in \ev{X}$. 
     If $e_1 \ip e_2$, then $(\key{e_1}, \key {e_2}) \in \ord{X}$.
\end{lemma}

\begin{proof}
	We apply \autoref{lem-immed pred composable} to deduce that $e_1$ is composable with $e_2$ and $e_1 \dep e_2$.
	Let $k_1 = \key {e_1}$, $k_2 = \key {e_2}$.
	Composability of $e_1$ and $e_2$ implies that $k_1 \neq k_2$.
	Let $rt_1t_2r'$ be a forward-only path with target \(X\) and $t_1 \in e_1$, $t_2 \in e_2$.
	We now consider the path \(rt_1t_2\) with target \(X'\): as \(r'\) is forward-only, \autoref{lem:fwd ord} will allow us to conclude $(\key{e_1}, \key {e_2}) \in \ord{X}$ from $(\key{e_1}, \key {e_2}) \in \ord{X'}$, which we now prove.

	We have $\lblof {t_1} \sdep \lblof {t_2}$ as $e_1 \dep e_2$.
	From there, we proceed by structural induction on $X'$, using the rules for $\sdep$ and $\lind$ as pictured in \autoref{fig:pccsk-dep-ind}.
	\begin{description}
		\item[$P$.]
		This case cannot arise, since $\ev P = \emptyset$.
		\item[${\alpha[k].Y}$.]
		There are various sub-cases:
		\begin{enumerate}
			\item $k_1 = k$.
			Then $k_2 \in \keys Y$, and so $(k_1,k_2) \in \ord {\alpha[k].Y}$.
			\item $k_2 = k$.
			This cannot arise, since $t_1$ occurs before $t_2$ in the path $rt_1t_2$.
			\item $k_1,k_2 \neq k$.
			Then $k_1,k_2 \in \keys Y$.
			There is a forward-only path $r't'_1t'_2$ obtained by omitting the first transition of $r$ and removing prefixes $\alpha[k]$ from processes.
			Clearly $\lblof{t'_1} \sdep \lblof{t'_2}$.
			By induction we have $(k_1,k_2) \in \ord{Y}$, and so $(k_1,k_2) \in \ord {\alpha[k].Y}$.
		\end{enumerate}
		\item[$Y + Q$.]
		Then there is a forward-only path $r^{\L}t_1^{\L}t_2^{\L}$ with target \(Y\) obtained by projecting onto the left-hand component.
		We use rule C$^1$ to deduce $\lblof{t_1^{\L}} \sdep \lblof{t_2^{\L}}$.
		So $\evtkof {Y} {k_1} \sdep \evtkof {Y} {k_2}$ and these events are composable.
		By induction we have $(k_1, k_2) \in \ord{Y}$, and so $(k_1, k_2) \in \ord {Y + Q}$.
		
		\item[$P + Y$.]
		Similar to the preceding case.
		
		\item[$Y \bs \lambda$.]
		Then there is a forward-only path $r't'_1t'_2$ with target \(Y\) obtained by removing the restriction.
		Clearly $\evtkof {Y} {k_1} \sdep \evtkof {Y} {k_2}$ and these events are composable.
		By induction we have $(k_1, k_2) \in \ord{Y}$, and so $(k_1, k_2) \in \ord {Y \bs \lambda}$.
		
		\item[$Y_{\L} \Par Y_{\R}$.]
		There are various sub-cases:
		\begin{enumerate}
			\item $k_1,k_2 \in \keys{Y_{\L}} \inter \keys{Y_{\R}}$.
			Then there are forward-only paths $r^{\D}t_1^{\D}t_2^{\D}$ with target \(Y_{\D}\), for both $\D \in \{\L, \R\}$, obtained by projecting onto the left-hand and right-hand components.
			We use rule S$^3$ to deduce \mbox{$\lblof{t_1^{\D}} \sdep \lblof{t_2^{\D}}$} for either $\D =\L$ or $\D = \R$.
			Without loss of generality, suppose $\lblof{t_1^{\L}} \sdep \lblof{t_2^{\L}}$.
			By induction we have $(k_1, k_2) \in \ord{Y_{\L}}$, and so $(k_1, k_2) \in \ord {Y_{\L} \Par Y_{\R}}$.
			\item $k_1 \in \keys{Y_{\L}} \inter \keys{Y_{\R}}$ and $k_2 \in \keys{Y_{\L}} \setminus \keys{Y_{\R}}$.
			Then there is a forward-only path $r^{\L}t_1^{\L}t_2^{\L}$ to \(Y_{\L}\) obtained by projecting onto the left-hand component.
			We use rule S$^2$ to deduce $\lblof{t_1^{\L}} \sdep \lblof{t_2^{\L}}$.
			By induction we have $(k_1, k_2) \in \ord{Y_{\L}}$, and so $(k_1, k_2) \in \ord {Y_{\L} \Par Y_{\R}}$.
			\item $k_1 \in \keys{Y_{\L}} \inter \keys{Y_{\R}}$ and $k_2 \in \keys{Y_{\R}} \setminus \keys{Y_{\L}}$.
			Similar to the preceding case.
			\item $k_1 \in \keys{Y_{\L}} \setminus \keys{Y_{\R}}$ and $k_2 \in \keys{Y_{\L}} \inter \keys{Y_{\R}}$.
			Then there is a forward-only path $r^{\L}t_1^{\L}t_2^{\L}$ to \(Y_{\L}\) obtained by projecting onto the left-hand component.
			We use rule S$^1$ to deduce $\lblof{t_1^{\L}} \sdep \lblof{t_2^{\L}}$.
			By induction we have $(k_1, k_2) \in \ord{Y_{\L}}$, and so $(k_1, k_2) \in \ord {Y_{\L} \Par Y_{\R}}$.
			\item $k_1 \in \keys{Y_{\R}} \setminus \keys{Y_{\L}}$ and $k_2 \in \keys{Y_{\L}} \inter \keys{Y_{\R}}$.
			Similar to the preceding case.
			\item $k_1, k_2 \in \keys{Y_{\L}} \setminus \keys{Y_{\R}}$.
			Then there is a forward-only path $r^{\L}t_1^{\L}t_2^{\L}$ to \(Y_{\L}\) obtained by projecting onto the left-hand component.
			We use rule P$^1$ to deduce $\lblof{t_1^{\L}} \sdep \lblof{t_2^{\L}}$.
			By induction we have $(k_1, k_2) \in \ord{Y_{\L}}$, and so $(k_1, k_2) \in \ord {Y_{\L} \Par Y_{\R}}$.
			\item $k_1, k_2 \in \keys{Y_{\R}} \setminus \keys{Y_{\L}}$.
			Similar to the preceding case.
			\item $k_1 \in \keys{Y_{\L}} \setminus \keys{Y_{\R}}$ and $k_2 \in \keys{Y_{\R}} \setminus \keys{Y_{\L}}$.
			This case cannot arise, since we would have $\lblof {t_1} \ind \lblof {t_2}$ by rule P$^2_k$, a contradiction.
			\item $k_1 \in \keys{Y_{\R}} \setminus \keys{Y_{\L}}$ and $k_2 \in \keys{Y_{\L}} \setminus \keys{Y_{\R}}$.
			Similar to the preceding case.
			\qedhere
		\end{enumerate}
	\end{description}
\end{proof}

In view of \autoref{lem-event-key} we can make the following definition: \itemref{lem-event-key-1} guarantees that no two events in \(\ev{X}\) have the same key; \itemref{lem-event-key-2} guarantees that for every key in \(\keys{X}\), there is an event in \(\ev{X}\) with that key.

\begin{definition}[Event key]%
	\label{def-event-key}
	Let $X$ be reachable and let $k \in \keys X$.  Define $\evtkof X k$ to be the unique event $e \in \ev{X}$ such that $\key e = k$.
\end{definition}

\begin{lemma}[Order from keys to events]%
	\label{lem-ord}
	Suppose $X$ is reachable and $(k_1,k_2) \in \ord{X}$.  Then $\evtkof X {k_1} < \evtkof X {k_2}$.
\end{lemma}

\begin{proof}
	By structural induction on $X$.
	In the definition of $e_1 < e_2$ we can restrict to forward-only paths using~\cite[Lemma 4.26]{LPU24}, since the LTSI of \pccsk is pre-reversible.
	It is convenient to use event key equivalence (\autoref{def-event-key-equivalence}), which is equivalent to equivalence (\autoref{def-event-general}) by \autoref{prop-event_coincide}.
	There are various cases:
	\begin{description}
		\item[$P$.]
		This case cannot arise, since $\keys P = \emptyset$.
		\item[${\alpha[k].X}$.]
		There are two sub-cases:
		\begin{enumerate}
			\item \label{item:ord base}
			Suppose that $(k_1,k_2) \in \ord {\alpha[k].X}$, and this is derived from $k_1 = k$ and $k_2 \in \keys X$.
			Suppose that $r$ is any rooted forward-only path with $\cte(r,\evtkof {\alpha[k].X} {k_2}) = 1$.
			Then the first transition of $r$ is labelled with $\alpha[k]$, so that $\cte(r,\evtkof {\alpha[k].X} {k_1}) = 1$ also.
			This shows that $\evtkof {\alpha[k].X} {k_1} < \evtkof {\alpha[k].X} {k_2}$.
			
			\item \label{item:ord ind}
			Suppose that $(k_1,k_2) \in \ord {\alpha[k].X}$, and this is derived from $(k_1,k_2) \in \ord{X}$.
			To reach a contradiction, suppose also that $\evtkof {\alpha[k].X} {k_1} \not < \evtkof {\alpha[k].X} {k_2}$.
			Then there is a rooted forward-only path $r$ with \begin{align*}
				\cte(r,\evtkof {\alpha[k].X} {k_2}) = 1 && \text{ and } && \cte(r,\evtkof {\alpha[k].X} {k_1}) = 0\text{.}
			\end{align*}
			Suppose that $t$ is the transition in $r$ which belongs to $\evtkof {\alpha[k].X} {k_2}$.
			Then $t \evkeqt t'$ where $t'$ belongs to a rooted forward-only path $r'$ with target $\alpha[k].X$.
			Let $s$ be a path not containing $k_2$ from $\tgtof {t}$ to $\tgtof {t'}$.
			
			We can omit the initial $\alpha[k]$ transition in $r$ and $r'$, and delete all $\alpha[k]$ prefixes in \(r\), \(r'\) and \(s\), yielding paths \(r_0\), \(r'_0\) and \(s_0\) and transitions $t_0 \in r_0$ and $t'_0 \in r'_0$. Using $s_0$ we see that $t_0 \evkeqt t'_0$.  Using $r'_0$ we see that $t'_0 \in \evtkof X {k_2}$, and so $t_0 \in \evtkof X {k_2}$.
			Now $\cte(r_0, \evtkof X {k_2}) = 1$ and $\cte(r_0, \evtkof X {k_1}) = 0$,
			showing that $\evtkof X {k_1} \not < \evtkof X {k_2}$, contradicting the inductive hypothesis.
		\end{enumerate}
		\item[$X + Q$] Similar to sub-case \itemref{item:ord ind} for $\alpha[k].X$.
		\item[$P + X$] Similar to the preceding case.
		\item[$P \bs \lambda$] Similar to sub-case \itemref{item:ord ind} for $\alpha[k].X$.
		\item[$X \Par Y$]
		Suppose that $(k_1,k_2) \in \ord {X \Par Y}$, which is derived from $(k_1,k_2) \in \ord{X}$.
		To reach a contradiction, suppose also $\evtkof {X \Par Y} {k_1} \not < \evtkof {X \Par Y} {k_2}$.
		Then there is a rooted forward-only path $r$ with $\cte(r,\evtkof {X \Par Y} {k_2}) = 1$ and $\cte(r,\evtkof {X \Par Y} {k_1}) = 0$.
		Suppose that $t$ is the transition in $r$ which belongs to $\evtkof {X \Par Y} {k_2}$.  Then $t \evkeqt t'$ where $t'$ belongs 
		to a rooted forward-only path $r'$ with target $X \Par Y$.
		Let $s$ be a path not containing $k_2$ from $\tgtof {t}$ to $\tgtof {t'}$.
		We can project \(r\), \(r'\) and \(s\) onto their left-hand components (by omitting any moves made solely on the right-hand component, \ie still executing the left-hand component of synchronisations as isolated actions), yielding paths \(r_{\Left}\), \(r'_{\Left}\), and \(s_\Left\) and transitions $t_\Left \in r_\Left$ and $t'_\Left \in r'_\Left$. Using $s_\Left$ we see that $t_\Left \evkeqt t'_\Left$.  Using $r'_\Left$ we see that $t'_\Left \in \evtkof X {k_2}$, and so $t_\Left \in \evtkof X {k_2}$.
		Now $\cte(r_\Left, \evtkof X {k_2}) = 1$ and $\cte(r_\Left, \evtkof X {k_1}) = 0$,
		showing that $\evtkof X {k_1} \not < \evtkof X {k_2}$, contradicting the inductive hypothesis.
		
		The case where $(k_1,k_2) \in \ord {X \Par Y}$ is derived from $(k_1,k_2) \in \ord Y$ is similar. %
		\qedhere
	\end{description}
\end{proof}

\begin{theorem}[Orderings coincide]%
	\label{thm-ordering-event-key}
	For any process \(X\), if $e_1, e_2 \in \ev{X}$ we have: $e_1 \leq e_2$ iff $\key{e_1} \leq_{X} \key {e_2}$.
\end{theorem}

\begin{proof}
    Note that \(\key{e_1} = \key{e_2}\) iff \(e_1 = e_2\) by \autoref{def-event-key}, hence we only have to prove for the strict orders.
    
	($\Rightarrow$)
	Let \(e_1 < e_2\).
	Since \(\leq_{X}\) is the transitive reflexive closure of \(\ord{X}\) (\autoref{def-partial-order-on-keys}), and  \(\leq\) is the transitive reflexive closure of \({\ip}\) (\autoref{prop-chain ip}), we obtain the result by \autoref{lem-composable sdep}.
	
	($\Leftarrow$)
	Suppose \(\key{e_1} <_X \key{e_2}\), then by \autoref{def-partial-order-on-keys} there exists a chain of keys \(k_1, \hdots, k_n \in \keys{X}\) such that 
	\[\{(\key{e_1}, k_1), (k_1, k_2), \hdots, (k_n, \key{e_2}) \} \subseteq \ord{X}\text{,}\]
	with \(n \geq 0\) since \(\key{e_1} \neq \key{e_2}\). 
	By \autoref{lem-ord} we obtain that 
	\[\evtkof{X}{\key{e_1}} < \evtkof{X}{k_1} < \evtkof{X}{k_2} < \hdots < \evtkof{X}{k_n} < \evtkof{X}{\key{e_2}}\text{,}\]
	hence \(\evtkof{X}{\key{e_1}} = e_1 < e_2 = \evtkof{X}{\key{e_2}}\) as desired.
\end{proof}

We now conclude this section by showing that the core independence on past events from \autoref{def-event relations} can also be equivalently computed using static ordering on keys (\autoref{thm-ordering-event-key}).

\begin{theorem}[Core independence from keys]%
	\label{thm-independent-event-key}
	For any process \(X\), if $e_1, e_2 \in \ev{X}$, we have: $e_1 \coind  e_2$ iff neither \(\key{e_1} \leq_{X} \key{e_2}\) nor \(\key{e_2} \leq_{X} \key{e_1}\) holds.
\end{theorem}

\begin{proof}
	($\Rightarrow$)	
	Since $e_1 \coind e_2$, by polychotomy we have $e_1 \not\leq e_2$ and $e_2 \not\leq e_1$.
	Then the result follows from \autoref{thm-ordering-event-key}.

	($\Leftarrow$)
	By polychotomy for events (\autoref{prop-poly}), as \(e_1 \ci e_2\) is the only option left:
	\begin{itemize}
		\item \(e_1 = e_2\) cannot hold since neither \(\key{e_1} \leq_{X} \key{e_2}\) nor \(\key{e_2} \leq_{X} \key{e_1}\) hold by assumption; hence \(\key{e_1} \neq \key{e_2}\) and \(e_1 \neq e_2\) follows.
		\item \(e_1 < e_2\) and \(e_1 > e_2\) do not hold by using \autoref{thm-ordering-event-key} with our hypothesis.
		\item \(e_1 \cf e_2\) does not hold: since $e_1, e_2 \in \ev{X}$, by defitinion of \(\ev{X}\) (\autoref{def-events process alt1}), there exists a rooted forward-only path \(r\) to \(X\) such that \(\cte(r, e_1) > 0\) and \(\cte(r, e_2) > 0\), which would contradict \(e_1 \cf e_2\) by \autoref{def-event relations}, remembering that all events in \(\ev{X}\) are forward events. \qedhere
	\end{itemize}
\end{proof}

We conclude with a new proof of the following statement, which is of interest, for example, when establishing the operational correspondence between configuration structures and process algebras~\cite[Lemma 19]{Aubert2020b}\cite[Lemma 3.5]{AubertPU26}:

\begin{corollary}%
	\label{cor:bwd maximal}
	For any process \(X\), $X \pr{b}[\theta] X'$ with \(\kay{\theta} = k\) iff $k$ is maximal under $\leq_X$.
\end{corollary}

\begin{proof}
	($\Rightarrow$)
	This direction follows from \autoref{lem:fwd key maximal}, using the Loop Lemma (\autoref{lem-loop_proved}).

	($\Leftarrow$)
	Suppose $k$ is maximal under $\leq_X$, and consider any rooted forward-only path $r: \orig X \r{f}^* X$.
	Then there is some transition $t : X_1 \pr{f}[\theta] X_2$ with $\kay{\theta}=k$ in $r$, so that we can let $r = r_1tr_2$.
	By \autoref{thm-ordering-event-key}, $[t]$ is maximal in $\ev{X}$, since $k$ is maximal under $\leq_X$.
	By polychotomy for transitions (\autoref{prop-poly-trans}), $t \coind u$ for all $u$ in $r_2$: \(t \sim u\) cannot hold since they both occur in the same path as forward transitions, \(t < u\) cannot hold since \([t]\) is maximal under \(<\) in \(\ev{X}\), \(t > u\) cannot hold since \(t\) occurred before \(u\), and finally \(t \cf u\) cannot hold since they both occur in the same path.
	Using \RPI and \SP, we can swap $t$ with all transitions in $r_2$ to obtain $X \pr{b}[\theta] X'$ with \(\kay{\theta} = k\).
\end{proof}

%% file: figures/example-transition2.tex
	\tikzset{
		lbl/.style = {
			scale=.8,
		}
	}
	
	\begin{tikzpicture}[
		x={(1.5, 0)},  %
		y={(0, 1.1)}, %
		baseline,
		anchor=base
		]
		\node (aab) at (0, 0){\(a \Par b\)};
		\node (amamb) at (1.2, .9){\(a[k] \Par b\)};
		\node (aamb) at (1.2, -.9){\(a \Par b[l]\)};
		\node (akabm) at (2.4, 0){\(a[k] \Par b[l]\)};
		\node (neg) at (-1.2, 0){\(a \Par b[k]\)};
		
		\draw[pb, reverse] (amamb) -- node[lbl, left, yshift=8, xshift=6]{\(\lmidl a[k]\)} (aab);
		\draw[pf] (aab) -- node[lbl, right, xshift=-3, yshift=6]{\(\lmidr b[l]\)} (aamb);
		\draw[pf] (amamb) -- node[lbl, right, xshift=-3, yshift=6]{\(\lmidr b[l]\)} (akabm);
		\draw[pf] (aamb) -- node[lbl, left, yshift=8, xshift=6]{\(\lmidl a[k]\)} (akabm);
		\draw[pb, reverse] (neg)  node[lbl, above, xshift=35, yshift=0]{\(\lmidr b[k]\)} -- (aab);
	\end{tikzpicture}

%% file: sections/newconclusion.tex
\section{Conclusion}%
\label{sec:conclusion}

\subsection*{Scope and Applicability of our Contributions}

We have presented a detailed and complete study of how independence relates to concurrency, causality and conflict, introducing along the way additional useful relations on transitions and events.
Our work focused on reversible systems, %
but its applicability goes beyond them: arbitrary (forward-only) systems can be automatically extended to a pre-reversible semantics~\cite[Proposition 20]{LaneseM20}, where our theorems can then be applied.
Our results can also be read \enquote{backwards}: for example, since our independence on \ccsk conservatively extends the existing definition of independence on \ccs~\cite[Sect. 3]{BC94}, it strongly grounds it retrospectively as correct.
Finally, while our uniqueness results (\autoref{prop-prerev-coinitial-unique} and \autoref{thm-uniqueness}) are focused on adjacent transitions (following \autoref{remark-on-adjacent-indep}), it is important to note that independence is routinely defined only on such transitions,
and that independence on coinitial transitions can in some systems fully determine independence between all transitions%
        \footnote{This is the case for example for systems satisfying \IRE and the \enquote{Independence of Events is Coinitial} property from the axiomatic approach~\cite[Definition 5.10]{LaneseM20}: if \(t_1 \ind t_2\) then \([t_1] \ci [t_2]\), which in turns implies that there are coinitial \(t_1' \eveqt t_1\), \(t_2' \eveqt t_2\) such that \(t_1' \ind t_2'\).}.

\subsection*{Beluga Formalisation: Evaluation and Reflections}

Mechanising a formal system comes with a range of choices that can significantly affect the formalisation effort.
We chose to treat restriction as a binding operator, as is customary in mechanisations of concurrent calculi~\cite{HonsellMS01,SanoKP23}, even though some pen-and-paper developments of \ccs and reversible variants~\cite{milner80lncs,PU07,aubert2023c} do not explicitly identify processes up to \(\alpha\)-equivalence.
We followed the original \ccsk definition~\cite{PU07} in considering all keys as free, even if the distinction between free and bound keys is gaining traction~\cite{LaneseP21,Vallee2026,lanese_et_al:LIPIcs.CONCUR.2026.41}.
The presence of binders in the target system, together with the existence of a previous Beluga formalisation of \pccsk~\cite{cecilia_2025_15660907}, motivated our choice of Beluga. %

Beluga's support for HOAS relieved us from the need to explicitly manage indices or prove results up to \(\alpha\)-equivalence when dealing with restriction, while its proof-term-oriented style facilitated a close correspondence between formal and informal proofs.
On the other hand, its contextual, two-level approach appears less suitable for the axiomatic development, requiring one to design appropriate context schemas to represent relations such as independence, and to establish the invariance of notions and proofs under the choice of representative of an event.
Therefore, we regard Beluga as a strong choice for languages with binding constructs, while a proof assistant such as Lean~\cite{Lean4} may be better suited to approach the axiomatic material in Sections~\ref{sec:axiomatic}, \ref{sec:CCC}, \ref{ssec:true-concu-relation-pccsk} and~\ref{sec:key-based}, given its support for abstract mathematical structures and quotient types.

The formalisation helped us identify and resolve issues, revise some of our proof strategies, and certify the correctness of our results.
In particular, it revealed that viewing restriction as binding makes its treatment in \pccsk significantly more complex than anticipated:
handling bound names in proof labels required a duplication of types and proofs, accounting for a significant portion of the overall \bkloc{10} of the development.
Despite this, our confidence in the correctness of the encoding of \pccsk rests on its bijection with the more conventional encoding of \ccsk, that establishes their \emph{formal adequacy}~\cite[p.~211]{CheneyNV12}.
Overall, we hope that the richness of our example bank and the central status of \ccsk in the world of reversible concurrent systems will facilitate and inspire future implementation efforts~\cite{DBLP:conf/rc/AubertB23}.

\subsection*{Extending and Strengthening our Contributions}
Our contributions sharpen the understanding that the tools used to make systems reversible do more than enabling backtracking.
For example, marking past actions with keys further enables a static detection of causality and core independence (Theorems~\ref{thm-ordering-event-key} and \ref{thm-independent-event-key}) that could be exploited independently of backward transitions.
Indeed, keys essentially play the role of abstracted proof labels stored inside the process, 
and exploring other unstructured and anonymised ways of recording this information 
could allow other approaches to concurrency to benefit from our results.
Such an effort would be on a per-system basis, as the axiomatic approach refrains from assuming structure in processes.

Another contribution of reversibility to concurrency theory~\cite{AubertPU26,10.1007/3-540-48340-3_32,Aubert2020b,BEM25} revolves around the study of %
history-preserving bisimulation, seen as \enquote{(just) a bisimulation for causality}~\cite{10.1007/3-540-48340-3_32} %
and \emph{hereditary} history-preserving bisimulation, known to preserve concurrency.
In light of our results, one could imagine defining those bisimulations in terms of true-concurrency relations~\cite{bisim-applied-pi}, %
\eg requiring %
 transitions to play the bisimulation game only when their true-concurrency behaviour matches.
Last but not least, a possibly fruitful extension of our work would be to investigate whether the qualitative changes introduced by our characterisations of true-concurrency relations (\eg moving from universally-quantified definitions to existentially-quantified ones) enable us to reduce the complexity of deciding those relations.

%% file: main.bbl

\begin{thebibliography}{53}


\ifx \showCODEN    \undefined \def \showCODEN     #1{\unskip}     \fi
\ifx \showISBNx    \undefined \def \showISBNx     #1{\unskip}     \fi
\ifx \showISBNxiii \undefined \def \showISBNxiii  #1{\unskip}     \fi
\ifx \showISSN     \undefined \def \showISSN      #1{\unskip}     \fi
\ifx \showLCCN     \undefined \def \showLCCN      #1{\unskip}     \fi
\ifx \shownote     \undefined \def \shownote      #1{#1}          \fi
\ifx \showarticletitle \undefined \def \showarticletitle #1{#1}   \fi
\ifx \showURL      \undefined \def \showURL       {\relax}        \fi
\providecommand\bibfield[2]{#2}
\providecommand\bibinfo[2]{#2}
\providecommand\natexlab[1]{#1}
\providecommand\showeprint[2][]{arXiv:#2}
\makeatletter
\@ifundefined{NAT@parse@date}{}{\let\NAT@parse@date@orig\NAT@parse@date}
\@ifundefined{NAT@parse@date}{}{\def\NAT@parse@date#1#2#3#4#5#6@@{\NAT@parse@date@orig#1#2#3#4#5#6@@\def\NAT@tempyear{0000}\def\NAT@tempexlab{{?}}\ifx\NAT@year\NAT@tempyear\ifx\NAT@exlab\NAT@tempexlab\def\NAT@date{[n.\,d.]}\else\edef\NAT@date{[n.\,d.]\NAT@exlab}\fi\fi}}
\makeatother

\bibitem[Abel et~al\mbox{.}(2019)]%
        {AAH19}
\bibfield{author}{\bibinfo{person}{Andreas Abel}, \bibinfo{person}{Guillaume Allais}, \bibinfo{person}{Aliya Hameer}, \bibinfo{person}{Brigitte Pientka}, \bibinfo{person}{Alberto Momigliano}, \bibinfo{person}{Steven Sch{\"{a}}fer}, {and} \bibinfo{person}{Kathrin Stark}.} \bibinfo{year}{2019}\natexlab{}.
\newblock \showarticletitle{POPLMark reloaded: Mechanizing proofs by logical relations}.
\newblock \bibinfo{journal}{\emph{Journal of Functional Programming}}  \bibinfo{volume}{29} (\bibinfo{year}{2019}), \bibinfo{pages}{e19}.
\newblock
\href{https://doi.org/10.1017/S0956796819000170}{doi:\nolinkurl{10.1017/S0956796819000170}}


\bibitem[Aman et~al\mbox{.}(2020)]%
        {Aman2020}
\bibfield{author}{\bibinfo{person}{Bogdan Aman}, \bibinfo{person}{Gabriel Ciobanu}, \bibinfo{person}{Robert Gl{\"{u}}ck}, \bibinfo{person}{Robin Kaarsgaard}, \bibinfo{person}{Jarkko Kari}, \bibinfo{person}{Martin Kutrib}, \bibinfo{person}{Ivan Lanese}, \bibinfo{person}{Claudio~Antares Mezzina}, \bibinfo{person}{Lukasz Mikulski}, \bibinfo{person}{Rajagopal Nagarajan}, \bibinfo{person}{Iain C.~C. Phillips}, \bibinfo{person}{G.~Michele Pinna}, \bibinfo{person}{Luca Prigioniero}, \bibinfo{person}{Irek Ulidowski}, {and} \bibinfo{person}{Germ{\'{a}}n Vidal}.} \bibinfo{year}{2020}\natexlab{}.
\newblock \showarticletitle{Foundations of Reversible Computation}.
\newblock In \bibinfo{booktitle}{\emph{Reversible Computation: Extending Horizons of Computing - Selected Results of the {COST} Action {IC1405}}}, \bibfield{editor}{\bibinfo{person}{Irek Ulidowski}, \bibinfo{person}{Ivan Lanese}, \bibinfo{person}{Ulrik~Pagh Schultz}, {and} \bibinfo{person}{Carla Ferreira}} (Eds.). \bibinfo{series}{LNCS}, Vol.~\bibinfo{volume}{12070}. \bibinfo{publisher}{Springer}, \bibinfo{pages}{1--40}.
\newblock
\href{https://doi.org/10.1007/978-3-030-47361-7_1}{doi:\nolinkurl{10.1007/978-3-030-47361-7_1}}


\bibitem[Aubert(2023)]%
        {aubert2023c}
\bibfield{author}{\bibinfo{person}{Cl\'{e}ment Aubert}.} \bibinfo{year}{2023}\natexlab{}.
\newblock \showarticletitle{{The Correctness of Concurrencies in (Reversible) Concurrent Calculi}}.
\newblock \bibinfo{journal}{\emph{Journal of Logical and Algebraic Methods in Programming}}  \bibinfo{volume}{136} (\bibinfo{year}{2023}), \bibinfo{pages}{100924}.
\newblock
\showISSN{2352-2208}
\href{https://doi.org/10.1016/j.jlamp.2023.100924}{doi:\nolinkurl{10.1016/j.jlamp.2023.100924}}


\bibitem[Aubert and Browning(2023)]%
        {DBLP:conf/rc/AubertB23}
\bibfield{author}{\bibinfo{person}{Cl\'{e}ment Aubert} {and} \bibinfo{person}{Peter Browning}.} \bibinfo{year}{2023}\natexlab{}.
\newblock \showarticletitle{Implementation of a Reversible Distributed Calculus}. In \bibinfo{booktitle}{\emph{Reversible Computation - 15th International Conference, {RC} 2023, Giessen, Germany, July 18-19, 2023, Proceedings}} \emph{(\bibinfo{series}{LNCS}, Vol.~\bibinfo{volume}{13960})}, \bibfield{editor}{\bibinfo{person}{Martin Kutrib} {and} \bibinfo{person}{Uwe Meyer}} (Eds.). \bibinfo{publisher}{Springer}, \bibinfo{pages}{210--217}.
\newblock
\href{https://doi.org/10.1007/978-3-031-38100-3_13}{doi:\nolinkurl{10.1007/978-3-031-38100-3_13}}


\bibitem[Aubert and Cristescu(2020)]%
        {Aubert2020b}
\bibfield{author}{\bibinfo{person}{Cl\'{e}ment Aubert} {and} \bibinfo{person}{Ioana Cristescu}.} \bibinfo{year}{2020}\natexlab{}.
\newblock \showarticletitle{How Reversibility Can Solve Traditional Questions: The Example of Hereditary History-Preserving Bisimulation}. In \bibinfo{booktitle}{\emph{{CONCUR}}} \emph{(\bibinfo{series}{LIPICS}, Vol.~\bibinfo{volume}{171})}, \bibfield{editor}{\bibinfo{person}{Igor Konnov} {and} \bibinfo{person}{Laura Kov\'{a}cs}} (Eds.). \bibinfo{publisher}{Schloss Dagstuhl - Leibniz-Zentrum f{\"{u}}r Informatik}, \bibinfo{pages}{13:1--13:24}.
\newblock
\href{https://doi.org/10.4230/LIPIcs.CONCUR.2020.7}{doi:\nolinkurl{10.4230/LIPIcs.CONCUR.2020.7}}


\bibitem[Aubert et~al\mbox{.}(2022)]%
        {bisim-applied-pi}
\bibfield{author}{\bibinfo{person}{Cl\'{e}ment Aubert}, \bibinfo{person}{Ross Horne}, {and} \bibinfo{person}{Christian Johansen}.} \bibinfo{year}{2022}\natexlab{}.
\newblock \showarticletitle{Bisimulations Respecting Duration and Causality for the Non-interleaving Applied {\(\pi\)}-Calculus}, In \bibinfo{booktitle}{Proceedings Combined 29th International Workshop on Expressiveness in Concurrency and 19th Workshop on Structural Operational Semantics, {EXPRESS/SOS} 2022, Warsaw, Poland, 12th September 2022}, \bibfield{editor}{\bibinfo{person}{Valentina Castiglioni} {and} \bibinfo{person}{Claudio~Antares Mezzina}} (Eds.).
\newblock \bibinfo{journal}{\emph{EPTCS}}  \bibinfo{volume}{368}, \bibinfo{pages}{3--22}.
\newblock
\href{https://doi.org/10.4204/EPTCS.368.1}{doi:\nolinkurl{10.4204/EPTCS.368.1}}


\bibitem[Aubert et~al\mbox{.}(2024)]%
        {APU24}
\bibfield{author}{\bibinfo{person}{Cl\'{e}ment Aubert}, \bibinfo{person}{Iain C.~C. Phillips}, {and} \bibinfo{person}{Irek Ulidowski}.} \bibinfo{year}{2024}\natexlab{}.
\newblock \bibinfo{title}{Dependence and Independence for Reversible Process Calculi}.
\newblock
\showeprint[arxiv]{2410.14699}~[cs.LO]
\href{https://doi.org/10.48550/arXiv.2410.14699}{doi:\nolinkurl{10.48550/arXiv.2410.14699}}


\bibitem[Aubert et~al\mbox{.}(2025)]%
        {APU25}
\bibfield{author}{\bibinfo{person}{Cl\'{e}ment Aubert}, \bibinfo{person}{Iain C.~C. Phillips}, {and} \bibinfo{person}{Irek Ulidowski}.} \bibinfo{year}{2025}\natexlab{}.
\newblock \showarticletitle{Independence and Causality in the Reversible Concurrent Setting}. In \bibinfo{booktitle}{\emph{Reversible Computation - 17th International Conference, {RC} 2025, Odense, Denmark, July 3-4, 2025, Proceedings}} \emph{(\bibinfo{series}{LNCS}, Vol.~\bibinfo{volume}{15716})}, \bibfield{editor}{\bibinfo{person}{Robert Gl{\"{u}}ck} {and} \bibinfo{person}{Robin Kaarsgaard}} (Eds.). \bibinfo{publisher}{Springer}, \bibinfo{pages}{9--26}.
\newblock
\href{https://doi.org/10.1007/978-3-031-97063-4_2}{doi:\nolinkurl{10.1007/978-3-031-97063-4_2}}


\bibitem[Aubert et~al\mbox{.}(2026)]%
        {AubertPU26}
\bibfield{author}{\bibinfo{person}{Cl\'{e}ment Aubert}, \bibinfo{person}{Iain C.~C. Phillips}, {and} \bibinfo{person}{Irek Ulidowski}.} \bibinfo{year}{2026}\natexlab{}.
\newblock \showarticletitle{Bisimulations and Reversibility}. In \bibinfo{booktitle}{\emph{Components Operationally: Reversibility and System Engineering: Essays Dedicated to Jean-Bernard Stefani on the Occasion of His 65th Birthday}} \emph{(\bibinfo{series}{LNCS}, Vol.~\bibinfo{volume}{16065})}, \bibfield{editor}{\bibinfo{person}{Claudio~Antares Mezzina} {and} \bibinfo{person}{Alan Schmitt}} (Eds.). \bibinfo{publisher}{Springer}, \bibinfo{pages}{46--67}.
\newblock
\href{https://doi.org/10.1007/978-3-031-99717-4_3}{doi:\nolinkurl{10.1007/978-3-031-99717-4_3}}


\bibitem[Baeten(1992)]%
        {TOA1992}
\bibfield{author}{\bibinfo{person}{Jos C.~M. Baeten}.} \bibinfo{year}{1992}\natexlab{}.
\newblock \showarticletitle{The Total Order Assumption}. In \bibinfo{booktitle}{\emph{{NAPAW} 92, Proceedings of the First North American Process Algebra Workshop, Stony Brook, New York, USA, 28 Agust 1992}} \emph{(\bibinfo{series}{Workshops in Computing})}, \bibfield{editor}{\bibinfo{person}{S.~Purushothaman} {and} \bibinfo{person}{Amy~E. Zwarico}} (Eds.). \bibinfo{publisher}{Springer}, \bibinfo{pages}{231--240}.
\newblock
\href{https://doi.org/10.1007/978-1-4471-3217-2\_14}{doi:\nolinkurl{10.1007/978-1-4471-3217-2\_14}}


\bibitem[Baeten(2005)]%
        {BAETEN2005131}
\bibfield{author}{\bibinfo{person}{Jos C.~M. Baeten}.} \bibinfo{year}{2005}\natexlab{}.
\newblock \showarticletitle{A brief history of process algebra}.
\newblock \bibinfo{journal}{\emph{Theoretical Computer Science}} \bibinfo{volume}{335}, \bibinfo{number}{2} (\bibinfo{year}{2005}), \bibinfo{pages}{131--146}.
\newblock
\showISSN{0304-3975}
\href{https://doi.org/10.1016/j.tcs.2004.07.036}{doi:\nolinkurl{10.1016/j.tcs.2004.07.036}}


\bibitem[Bednarczyk(1991)]%
        {Bednarczyk1991}
\bibfield{author}{\bibinfo{person}{Marek~A. Bednarczyk}.} \bibinfo{year}{1991}\natexlab{}.
\newblock \bibinfo{booktitle}{\emph{Hereditary History Preserving Bisimulations or What is the Power of the Future Perfect in Program Logics}}.
\newblock \bibinfo{type}{{T}echnical {R}eport}. \bibinfo{institution}{Instytut Podstaw Informatyki PAN filia w Gdańsku}.
\newblock


\bibitem[Bernardo et~al\mbox{.}(2024)]%
        {BEM24}
\bibfield{author}{\bibinfo{person}{Marco Bernardo}, \bibinfo{person}{Andrea Esposito}, {and} \bibinfo{person}{Claudio~Antares Mezzina}.} \bibinfo{year}{2024}\natexlab{}.
\newblock \showarticletitle{Expansion Laws for Forward-Reverse, Forward, and Reverse Bisimilarities via Proved Encodings}, In \bibinfo{booktitle}{Proceedings Combined 31st International Workshop on Expressiveness in Concurrency and 21st Workshop on Structural Operational Semantics, {EXPRESS/SOS} 2024, Calgary, Canada, 9th September 2024}, \bibfield{editor}{\bibinfo{person}{Georgiana Caltais} {and} \bibinfo{person}{Cinzia~Di Giusto}} (Eds.).
\newblock \bibinfo{journal}{\emph{EPTCS}}  \bibinfo{volume}{412}, \bibinfo{pages}{51--70}.
\newblock
\href{https://doi.org/10.4204/EPTCS.412.5}{doi:\nolinkurl{10.4204/EPTCS.412.5}}


\bibitem[Bernardo et~al\mbox{.}(2025)]%
        {BEM25}
\bibfield{author}{\bibinfo{person}{Marco Bernardo}, \bibinfo{person}{Andrea Esposito}, {and} \bibinfo{person}{Claudio~Antares Mezzina}.} \bibinfo{year}{2025}\natexlab{}.
\newblock \showarticletitle{Alternative Characterizations of Hereditary History-Preserving Bisimilarity via Backward Ready Multisets}. In \bibinfo{booktitle}{\emph{Foundations of Software Science and Computation Structures - 28th International Conference, FoSSaCS 2025}} \emph{(\bibinfo{series}{LNCS}, Vol.~\bibinfo{volume}{15691})}, \bibfield{editor}{\bibinfo{person}{Parosh~Aziz Abdulla} {and} \bibinfo{person}{Delia Kesner}} (Eds.). \bibinfo{publisher}{Springer}, \bibinfo{pages}{67--87}.
\newblock
\href{https://doi.org/10.1007/978-3-031-90897-2_4}{doi:\nolinkurl{10.1007/978-3-031-90897-2_4}}


\bibitem[Boudol and Castellani(1988a)]%
        {BC88}
\bibfield{author}{\bibinfo{person}{G\'{e}rard Boudol} {and} \bibinfo{person}{Ilaria Castellani}.} \bibinfo{year}{1988}\natexlab{a}.
\newblock \showarticletitle{A non-interleaving semantics for {CCS} based on proved transitions}.
\newblock \bibinfo{journal}{\emph{Fundamenta Informaticae}} \bibinfo{volume}{11}, \bibinfo{number}{4} (\bibinfo{year}{1988}), \bibinfo{pages}{433--452}.
\newblock
\href{https://doi.org/10.3233/FI-1988-11406}{doi:\nolinkurl{10.3233/FI-1988-11406}}


\bibitem[Boudol and Castellani(1988b)]%
        {BoudolC88Rex}
\bibfield{author}{\bibinfo{person}{G{\'{e}}rard Boudol} {and} \bibinfo{person}{Ilaria Castellani}.} \bibinfo{year}{1988}\natexlab{b}.
\newblock \showarticletitle{Permutation of transitions: An event structure semantics for {CCS} and {SCCS}}. In \bibinfo{booktitle}{\emph{Linear Time, Branching Time and Partial Order in Logics and Models for Concurrency, School/Workshop, Noordwijkerhout, The Netherlands, May 30 - June 3, 1988, Proceedings}} \emph{(\bibinfo{series}{LNCS})}, \bibfield{editor}{\bibinfo{person}{J.~W. de~Bakker}, \bibinfo{person}{Willem~P. de~Roever}, {and} \bibinfo{person}{Grzegorz Rozenberg}} (Eds.). \bibinfo{publisher}{Springer}, \bibinfo{pages}{411--427}.
\newblock
\href{https://doi.org/10.1007/BFB0013028}{doi:\nolinkurl{10.1007/BFB0013028}}


\bibitem[Boudol and Castellani(1994)]%
        {BC94}
\bibfield{author}{\bibinfo{person}{G\'{e}rard Boudol} {and} \bibinfo{person}{Ilaria Castellani}.} \bibinfo{year}{1994}\natexlab{}.
\newblock \showarticletitle{Flow Models of Distributed Computations: Three Equivalent Semantics for {CCS}}.
\newblock \bibinfo{journal}{\emph{Information and Computation}} \bibinfo{volume}{114}, \bibinfo{number}{2} (\bibinfo{year}{1994}), \bibinfo{pages}{247--314}.
\newblock
\href{https://doi.org/10.1006/inco.1994.1088}{doi:\nolinkurl{10.1006/inco.1994.1088}}


\bibitem[Cecilia(2025a)]%
        {Cec25}
\bibfield{author}{\bibinfo{person}{Gabriele Cecilia}.} \bibinfo{year}{2025}\natexlab{a}.
\newblock \showarticletitle{A Formalization of the Reversible Concurrent Calculus {CCSKP} in Beluga}, In \bibinfo{booktitle}{Proceedings 18th Interaction and Concurrency Experience, {ICE} 2025, Lille, France, 20th June 2025}, \bibfield{editor}{\bibinfo{person}{Cl\'{e}ment Aubert}, \bibinfo{person}{Cinzia Di~Giusto}, \bibinfo{person}{Simon Fowler}, {and} \bibinfo{person}{Violet Ka~I Pun}} (Eds.).
\newblock \bibinfo{journal}{\emph{EPTCS}}  \bibinfo{volume}{425}, \bibinfo{pages}{55--72}.
\newblock
\href{https://doi.org/10.4204/EPTCS.425.5}{doi:\nolinkurl{10.4204/EPTCS.425.5}}


\bibitem[Cecilia(2025b)]%
        {cecilia_2025_15660907}
\bibfield{author}{\bibinfo{person}{Gabriele Cecilia}.} \bibinfo{year}{2025}\natexlab{b}.
\newblock \bibinfo{booktitle}{\emph{A Formalization of the Reversible Concurrent Calculus CCSKP in Beluga (artifact)}}.
\newblock
\href{https://doi.org/10.5281/zenodo.16179366}{doi:\nolinkurl{10.5281/zenodo.16179366}}


\bibitem[Cecilia and Momigliano(2024)]%
        {momigliano24}
\bibfield{author}{\bibinfo{person}{Gabriele Cecilia} {and} \bibinfo{person}{Alberto Momigliano}.} \bibinfo{year}{2024}\natexlab{}.
\newblock \showarticletitle{A Beluga Formalization of the Harmony Lemma in the {\(\pi\)}-Calculus}. In \bibinfo{booktitle}{\emph{Proceedings Workshop on Logical Frameworks and Meta-Languages: Theory and Practice, Tallinn, Estonia, 8th July 2024}} \emph{(\bibinfo{series}{EPTCS}, Vol.~\bibinfo{volume}{404})}, \bibfield{editor}{\bibinfo{person}{Florian Rabe} {and} \bibinfo{person}{Claudio Sacerdoti~Coen}} (Eds.). \bibinfo{publisher}{Open Publishing Association}, \bibinfo{pages}{1--17}.
\newblock
\href{https://doi.org/10.4204/EPTCS.404.1}{doi:\nolinkurl{10.4204/EPTCS.404.1}}


\bibitem[Cheney et~al\mbox{.}(2012)]%
        {CheneyNV12}
\bibfield{author}{\bibinfo{person}{James Cheney}, \bibinfo{person}{Michael Norrish}, {and} \bibinfo{person}{Ren{\'{e}} Vestergaard}.} \bibinfo{year}{2012}\natexlab{}.
\newblock \showarticletitle{Formalizing Adequacy: {A} Case Study for Higher-order Abstract Syntax}.
\newblock \bibinfo{journal}{\emph{J. Autom. Reason.}} \bibinfo{volume}{49}, \bibinfo{number}{2} (\bibinfo{year}{2012}), \bibinfo{pages}{209--239}.
\newblock
\href{https://doi.org/10.1007/S10817-011-9221-6}{doi:\nolinkurl{10.1007/S10817-011-9221-6}}


\bibitem[Danos and Krivine(2004)]%
        {DK04}
\bibfield{author}{\bibinfo{person}{Vincent Danos} {and} \bibinfo{person}{Jean Krivine}.} \bibinfo{year}{2004}\natexlab{}.
\newblock \showarticletitle{Reversible Communicating Systems}. In \bibinfo{booktitle}{\emph{{CONCUR} 2004 - Concurrency Theory, 15th International Conference, London, UK, August 31 - September 3, 2004, Proceedings}} \emph{(\bibinfo{series}{LNCS}, Vol.~\bibinfo{volume}{3170})}, \bibfield{editor}{\bibinfo{person}{Philippa Gardner} {and} \bibinfo{person}{Nobuko Yoshida}} (Eds.). \bibinfo{publisher}{Springer}, \bibinfo{pages}{292--307}.
\newblock
\showISBNx{3-540-22940-X}
\href{https://doi.org/10.1007/978-3-540-28644-8_19}{doi:\nolinkurl{10.1007/978-3-540-28644-8_19}}


\bibitem[D{\'{a}}valos and Melgratti(2026)]%
        {Davalos2026}
\bibfield{author}{\bibinfo{person}{Daniel D{\'{a}}valos} {and} \bibinfo{person}{Hern{\'{a}}n~C. Melgratti}.} \bibinfo{year}{2026}\natexlab{}.
\newblock \showarticletitle{A Lean Mechanization of Reversible Occurrence Nets}. In \bibinfo{booktitle}{\emph{Reversible Computation - 18th International Conference, {RC} 2026, Turin, Italy, July 9-10, 2026, Proceedings}} \emph{(\bibinfo{series}{LNCS}, Vol.~\bibinfo{volume}{16626})}, \bibfield{editor}{\bibinfo{person}{Cl\'{e}ment Aubert} {and} \bibinfo{person}{Luca Roversi}} (Eds.). \bibinfo{publisher}{Springer}, \bibinfo{pages}{95--111}.
\newblock
\href{https://doi.org/10.1007/978-3-032-30839-9_5}{doi:\nolinkurl{10.1007/978-3-032-30839-9_5}}


\bibitem[de~Moura and Ullrich(2021)]%
        {Lean4}
\bibfield{author}{\bibinfo{person}{Leonardo de Moura} {and} \bibinfo{person}{Sebastian Ullrich}.} \bibinfo{year}{2021}\natexlab{}.
\newblock \showarticletitle{The Lean 4 Theorem Prover and Programming Language}. In \bibinfo{booktitle}{\emph{Automated Deduction – CADE 28: 28th International Conference on Automated Deduction, Virtual Event, July 12–15, 2021, Proceedings}}. \bibinfo{publisher}{Springer-Verlag}, \bibinfo{address}{Berlin, Heidelberg}, \bibinfo{pages}{625–635}.
\newblock
\showISBNx{978-3-030-79875-8}
\href{https://doi.org/10.1007/978-3-030-79876-5_37}{doi:\nolinkurl{10.1007/978-3-030-79876-5_37}}


\bibitem[Degano et~al\mbox{.}(2003)]%
        {DeganoGP03}
\bibfield{author}{\bibinfo{person}{Pierpaolo Degano}, \bibinfo{person}{Fabio Gadducci}, {and} \bibinfo{person}{Corrado Priami}.} \bibinfo{year}{2003}\natexlab{}.
\newblock \showarticletitle{Causality and Replication in Concurrent Processes}. In \bibinfo{booktitle}{\emph{{PSI}}} \emph{(\bibinfo{series}{LNCS}, Vol.~\bibinfo{volume}{2890})}, \bibfield{editor}{\bibinfo{person}{Manfred Broy} {and} \bibinfo{person}{Alexandre~V. Zamulin}} (Eds.). \bibinfo{publisher}{Springer}, \bibinfo{pages}{307--318}.
\newblock
\showISBNx{3-540-20813-5}
\href{https://doi.org/10.1007/978-3-540-39866-0_30}{doi:\nolinkurl{10.1007/978-3-540-39866-0_30}}


\bibitem[Fr{\"o}schle and Hildebrandt(1999)]%
        {10.1007/3-540-48340-3_32}
\bibfield{author}{\bibinfo{person}{Sibylle~B. Fr{\"o}schle} {and} \bibinfo{person}{Thomas~T. Hildebrandt}.} \bibinfo{year}{1999}\natexlab{}.
\newblock \showarticletitle{On Plain and Hereditary History-Preserving Bisimulation}. In \bibinfo{booktitle}{\emph{Mathematical Foundations of Computer Science 1999}}, \bibfield{editor}{\bibinfo{person}{Miros{\l}aw Kuty{\l}owski}, \bibinfo{person}{Leszek Pacholski}, {and} \bibinfo{person}{Tomasz Wierzbicki}} (Eds.). \bibinfo{publisher}{Springer Berlin Heidelberg}, \bibinfo{address}{Berlin, Heidelberg}, \bibinfo{pages}{354--365}.
\newblock
\showISBNx{978-3-540-48340-3}
\href{https://doi.org/10.1007/3-540-48340-3_32}{doi:\nolinkurl{10.1007/3-540-48340-3_32}}


\bibitem[Glabbeek and Goltz(2001)]%
        {Glabbeek2001}
\bibfield{author}{\bibinfo{person}{Robert J.~van Glabbeek} {and} \bibinfo{person}{Ursula Goltz}.} \bibinfo{year}{2001}\natexlab{}.
\newblock \showarticletitle{Refinement of actions and equivalence notions for concurrent systems}.
\newblock \bibinfo{journal}{\emph{Acta Informatica}} \bibinfo{volume}{37}, \bibinfo{number}{4/5} (\bibinfo{year}{2001}), \bibinfo{pages}{229--327}.
\newblock
\href{https://doi.org/10.1007/s002360000041}{doi:\nolinkurl{10.1007/s002360000041}}


\bibitem[Glabbeek and Plotkin(2009)]%
        {Glabbeek2009}
\bibfield{author}{\bibinfo{person}{Robert J.~van Glabbeek} {and} \bibinfo{person}{Gordon~D. Plotkin}.} \bibinfo{year}{2009}\natexlab{}.
\newblock \showarticletitle{Configuration structures, event structures and {P}etri nets}.
\newblock \bibinfo{journal}{\emph{Theoretical Computer Science}} \bibinfo{volume}{410}, \bibinfo{number}{41} (\bibinfo{year}{2009}), \bibinfo{pages}{4111--4159}.
\newblock
\href{https://doi.org/10.1016/j.tcs.2009.06.014}{doi:\nolinkurl{10.1016/j.tcs.2009.06.014}}


\bibitem[Glabbeek and Vaandrager(1997)]%
        {vGV97}
\bibfield{author}{\bibinfo{person}{Robert J.~van Glabbeek} {and} \bibinfo{person}{Frits~W. Vaandrager}.} \bibinfo{year}{1997}\natexlab{}.
\newblock \showarticletitle{The difference between splitting in $n$ and $n+1$}.
\newblock \bibinfo{journal}{\emph{Information and Computation}} \bibinfo{volume}{136}, \bibinfo{number}{2} (\bibinfo{year}{1997}), \bibinfo{pages}{109--142}.
\newblock
\href{https://doi.org/10.1006/inco.1997.2634}{doi:\nolinkurl{10.1006/inco.1997.2634}}


\bibitem[Gl{\"{u}}ck et~al\mbox{.}(2023)]%
        {DBLP:conf/rc/GluckLMMPUV23}
\bibfield{author}{\bibinfo{person}{Robert Gl{\"{u}}ck}, \bibinfo{person}{Ivan Lanese}, \bibinfo{person}{Claudio~Antares Mezzina}, \bibinfo{person}{Jaroslaw~Adam Miszczak}, \bibinfo{person}{Iain C.~C. Phillips}, \bibinfo{person}{Irek Ulidowski}, {and} \bibinfo{person}{Germ{\'{a}}n Vidal}.} \bibinfo{year}{2023}\natexlab{}.
\newblock \showarticletitle{Towards a Taxonomy for Reversible Computation Approaches}. In \bibinfo{booktitle}{\emph{Reversible Computation - 15th International Conference, {RC} 2023, Giessen, Germany, July 18-19, 2023, Proceedings}} \emph{(\bibinfo{series}{LNCS}, Vol.~\bibinfo{volume}{13960})}, \bibfield{editor}{\bibinfo{person}{Martin Kutrib} {and} \bibinfo{person}{Uwe Meyer}} (Eds.). \bibinfo{publisher}{Springer}, \bibinfo{pages}{24--39}.
\newblock
\href{https://doi.org/10.1007/978-3-031-38100-3_3}{doi:\nolinkurl{10.1007/978-3-031-38100-3_3}}


\bibitem[Harper et~al\mbox{.}(1993)]%
        {Harper93}
\bibfield{author}{\bibinfo{person}{Robert Harper}, \bibinfo{person}{Furio Honsell}, {and} \bibinfo{person}{Gordon~D. Plotkin}.} \bibinfo{year}{1993}\natexlab{}.
\newblock \showarticletitle{A Framework for Defining Logics}.
\newblock \bibinfo{journal}{\emph{J. {ACM}}} \bibinfo{volume}{40}, \bibinfo{number}{1} (\bibinfo{year}{1993}), \bibinfo{pages}{143--184}.
\newblock
\href{https://doi.org/10.1145/138027.138060}{doi:\nolinkurl{10.1145/138027.138060}}


\bibitem[Honsell et~al\mbox{.}(2001)]%
        {HonsellMS01}
\bibfield{author}{\bibinfo{person}{Furio Honsell}, \bibinfo{person}{Marino Miculan}, {and} \bibinfo{person}{Ivan Scagnetto}.} \bibinfo{year}{2001}\natexlab{}.
\newblock \showarticletitle{{\(\pi\)}-calculus in (Co)inductive-type theory}.
\newblock \bibinfo{journal}{\emph{Theoretical Computer Science}} \bibinfo{volume}{253}, \bibinfo{number}{2} (\bibinfo{year}{2001}), \bibinfo{pages}{239--285}.
\newblock
\href{https://doi.org/10.1016/S0304-3975(00)00095-5}{doi:\nolinkurl{10.1016/S0304-3975(00)00095-5}}


\bibitem[Lanese and Medic(2020)]%
        {LaneseM20}
\bibfield{author}{\bibinfo{person}{Ivan Lanese} {and} \bibinfo{person}{Doriana Medic}.} \bibinfo{year}{2020}\natexlab{}.
\newblock \showarticletitle{A General Approach to Derive Uncontrolled Reversible Semantics}. In \bibinfo{booktitle}{\emph{{CONCUR}}} \emph{(\bibinfo{series}{LIPICS}, Vol.~\bibinfo{volume}{171})}, \bibfield{editor}{\bibinfo{person}{Igor Konnov} {and} \bibinfo{person}{Laura Kov\'{a}cs}} (Eds.). \bibinfo{publisher}{Schloss Dagstuhl - Leibniz-Zentrum f{\"{u}}r Informatik}, \bibinfo{pages}{33:1--33:24}.
\newblock
\href{https://doi.org/10.4230/LIPIcs.CONCUR.2020.33}{doi:\nolinkurl{10.4230/LIPIcs.CONCUR.2020.33}}


\bibitem[Lanese et~al\mbox{.}(2026)]%
        {lanese_et_al:LIPIcs.CONCUR.2026.41}
\bibfield{author}{\bibinfo{person}{Ivan Lanese}, \bibinfo{person}{Claudio~Antares Mezzina}, \bibinfo{person}{Iain C.~C. Phillips}, \bibinfo{person}{Irek Ulidowski}, {and} \bibinfo{person}{Shoji Yuen}.} \bibinfo{year}{2026}\natexlab{}.
\newblock \showarticletitle{On the Encodability of Reversible Process Calculi}. In \bibinfo{booktitle}{\emph{37th International Conference on Concurrency Theory (CONCUR 2026)}} \emph{(\bibinfo{series}{LIPICS}, Vol.~\bibinfo{volume}{391})}, \bibfield{editor}{\bibinfo{person}{Ana Sokolova} {and} \bibinfo{person}{Patrick Totzke}} (Eds.). \bibinfo{publisher}{Schloss Dagstuhl - Leibniz-Zentrum f{\"{u}}r Informatik}, \bibinfo{address}{Dagstuhl, Germany}, \bibinfo{pages}{41:1--41:20}.
\newblock
\showISBNx{978-3-95977-447-5}
\showISSN{1868-8969}
\href{https://doi.org/10.4230/LIPIcs.CONCUR.2026.41}{doi:\nolinkurl{10.4230/LIPIcs.CONCUR.2026.41}}


\bibitem[Lanese and Phillips(2021)]%
        {LaneseP21}
\bibfield{author}{\bibinfo{person}{Ivan Lanese} {and} \bibinfo{person}{Iain C.~C. Phillips}.} \bibinfo{year}{2021}\natexlab{}.
\newblock \showarticletitle{Forward-Reverse Observational Equivalences in {CCSK}}. In \bibinfo{booktitle}{\emph{Reversible Computation - 13th International Conference, {RC} 2021, Virtual Event, July 7-8, 2021, Proceedings}} \emph{(\bibinfo{series}{LNCS}, Vol.~\bibinfo{volume}{12805})}, \bibfield{editor}{\bibinfo{person}{Shigeru Yamashita} {and} \bibinfo{person}{Tetsuo Yokoyama}} (Eds.). \bibinfo{publisher}{Springer}, \bibinfo{pages}{126--143}.
\newblock
\href{https://doi.org/10.1007/978-3-030-79837-6_8}{doi:\nolinkurl{10.1007/978-3-030-79837-6_8}}


\bibitem[Lanese et~al\mbox{.}(2024)]%
        {LPU24}
\bibfield{author}{\bibinfo{person}{Ivan Lanese}, \bibinfo{person}{Iain C.~C. Phillips}, {and} \bibinfo{person}{Irek Ulidowski}.} \bibinfo{year}{2024}\natexlab{}.
\newblock \showarticletitle{An Axiomatic Theory for Reversible Computation}.
\newblock \bibinfo{journal}{\emph{ACM Transactions on Computational Logic}} \bibinfo{volume}{25}, \bibinfo{number}{2} (\bibinfo{year}{2024}), \bibinfo{pages}{1--40}.
\newblock
\href{https://doi.org/10.1145/3648474}{doi:\nolinkurl{10.1145/3648474}}


\bibitem[Maletto and Roversi(2024)]%
        {MalettoR24}
\bibfield{author}{\bibinfo{person}{Giacomo Maletto} {and} \bibinfo{person}{Luca Roversi}.} \bibinfo{year}{2024}\natexlab{}.
\newblock \showarticletitle{Certifying expressive power and algorithms of reversible primitive permutations with Lean}.
\newblock \bibinfo{journal}{\emph{Journal of Logical and Algebraic Methods in Programming}}  \bibinfo{volume}{136} (\bibinfo{year}{2024}), \bibinfo{pages}{100923}.
\newblock
\href{https://doi.org/10.1016/J.JLAMP.2023.100923}{doi:\nolinkurl{10.1016/J.JLAMP.2023.100923}}


\bibitem[Melgratti et~al\mbox{.}(2020)]%
        {MMU20}
\bibfield{author}{\bibinfo{person}{Hern{\'{a}}n~C. Melgratti}, \bibinfo{person}{Claudio~Antares Mezzina}, {and} \bibinfo{person}{Irek Ulidowski}.} \bibinfo{year}{2020}\natexlab{}.
\newblock \showarticletitle{Reversing Place Transition Nets}.
\newblock \bibinfo{journal}{\emph{LMCS}} \bibinfo{volume}{16}, \bibinfo{number}{4} (\bibinfo{year}{2020}).
\newblock
\href{https://doi.org/10.23638/LMCS-16(4:5)2020}{doi:\nolinkurl{10.23638/LMCS-16(4:5)2020}}


\bibitem[Milner(1980)]%
        {milner80lncs}
\bibfield{author}{\bibinfo{person}{Robin Milner}.} \bibinfo{year}{1980}\natexlab{}.
\newblock \bibinfo{booktitle}{\emph{A Calculus of Communicating Systems}}. \bibinfo{series}{LNCS}, Vol.~\bibinfo{volume}{92}.
\newblock \bibinfo{publisher}{Springer-Verlag}.
\newblock
\showISBNx{9783540102359}
\href{https://doi.org/10.1007/3-540-10235-3}{doi:\nolinkurl{10.1007/3-540-10235-3}}


\bibitem[Nielsen et~al\mbox{.}(1981)]%
        {NPW81}
\bibfield{author}{\bibinfo{person}{Mogens Nielsen}, \bibinfo{person}{Gordon~D. Plotkin}, {and} \bibinfo{person}{Glynn Winskel}.} \bibinfo{year}{1981}\natexlab{}.
\newblock \showarticletitle{Petri nets, event structures and domains, part {I}}.
\newblock \bibinfo{journal}{\emph{Theoretical Computer Science}}  \bibinfo{volume}{13} (\bibinfo{year}{1981}), \bibinfo{pages}{85--108}.
\newblock
\href{https://doi.org/10.1016/0304-3975(81)90112-2}{doi:\nolinkurl{10.1016/0304-3975(81)90112-2}}


\bibitem[Paolini et~al\mbox{.}(2015)]%
        {PaoliniPR15}
\bibfield{author}{\bibinfo{person}{Luca Paolini}, \bibinfo{person}{Mauro Piccolo}, {and} \bibinfo{person}{Luca Roversi}.} \bibinfo{year}{2015}\natexlab{}.
\newblock \showarticletitle{A Certified Study of a Reversible Programming Language}. In \bibinfo{booktitle}{\emph{21st International Conference on Types for Proofs and Programs, {TYPES} 2015, Tallinn, Estonia, May 18-21, 2015}} \emph{(\bibinfo{series}{LIPICS})}, \bibfield{editor}{\bibinfo{person}{Tarmo Uustalu}} (Ed.). \bibinfo{publisher}{Schloss Dagstuhl - Leibniz-Zentrum f{\"{u}}r Informatik}, \bibinfo{pages}{7:1--7:21}.
\newblock
\href{https://doi.org/10.4230/LIPICS.TYPES.2015.7}{doi:\nolinkurl{10.4230/LIPICS.TYPES.2015.7}}


\bibitem[Pfenning and Elliott(1988)]%
        {pfenning88pldi}
\bibfield{author}{\bibinfo{person}{Frank Pfenning} {and} \bibinfo{person}{Conal Elliott}.} \bibinfo{year}{1988}\natexlab{}.
\newblock \showarticletitle{Higher-Order Abstract Syntax}. In \bibinfo{booktitle}{\emph{Proceedings of the {ACM}-{SIGPLAN} Conference on Programming Language Design and Implementation}}. \bibinfo{publisher}{ACM Press}, \bibinfo{pages}{199--208}.
\newblock
\href{https://doi.org/10.1145/53990.54010}{doi:\nolinkurl{10.1145/53990.54010}}


\bibitem[Phillips and Ulidowski(2007a)]%
        {PU07a}
\bibfield{author}{\bibinfo{person}{Iain C.~C. Phillips} {and} \bibinfo{person}{Irek Ulidowski}.} \bibinfo{year}{2007}\natexlab{a}.
\newblock \showarticletitle{Reversibility and Models for Concurrency}. In \bibinfo{booktitle}{\emph{Proceedings of the Fourth Workshop on Structural Operational Semantics, SOS@LICS/ICALP 2007, Wroclaw, Poland, July 9, 2007}} \emph{(\bibinfo{series}{ENTCS}, Vol.~\bibinfo{volume}{192(1)})}, \bibfield{editor}{\bibinfo{person}{Robert J.~van Glabbeek} {and} \bibinfo{person}{Matthew Hennessy}} (Eds.). \bibinfo{publisher}{Elsevier}, \bibinfo{pages}{93--108}.
\newblock
\href{https://doi.org/10.1016/j.entcs.2007.08.018}{doi:\nolinkurl{10.1016/j.entcs.2007.08.018}}


\bibitem[Phillips and Ulidowski(2007b)]%
        {PU07}
\bibfield{author}{\bibinfo{person}{Iain C.~C. Phillips} {and} \bibinfo{person}{Irek Ulidowski}.} \bibinfo{year}{2007}\natexlab{b}.
\newblock \showarticletitle{Reversing algebraic process calculi}.
\newblock \bibinfo{journal}{\emph{Journal of Logic and Algebraic Programming}} \bibinfo{volume}{73}, \bibinfo{number}{1-2} (\bibinfo{year}{2007}), \bibinfo{pages}{70--96}.
\newblock
\href{https://doi.org/10.1016/j.jlap.2006.11.002}{doi:\nolinkurl{10.1016/j.jlap.2006.11.002}}


\bibitem[Pientka and Dunfield(2010)]%
        {PientkaD10}
\bibfield{author}{\bibinfo{person}{Brigitte Pientka} {and} \bibinfo{person}{Jana Dunfield}.} \bibinfo{year}{2010}\natexlab{}.
\newblock \showarticletitle{Beluga: {A} Framework for Programming and Reasoning with Deductive Systems (System Description)}. In \bibinfo{booktitle}{\emph{Automated Reasoning, 5th International Joint Conference, {IJCAR} 2010, Edinburgh, UK, July 16-19, 2010. Proceedings}} \emph{(\bibinfo{series}{LNCS})}, \bibfield{editor}{\bibinfo{person}{J{\"{u}}rgen Giesl} {and} \bibinfo{person}{Reiner H{\"{a}}hnle}} (Eds.). \bibinfo{publisher}{Springer}, \bibinfo{pages}{15--21}.
\newblock
\href{https://doi.org/10.1007/978-3-642-14203-1_2}{doi:\nolinkurl{10.1007/978-3-642-14203-1_2}}


\bibitem[Pizzo and Sacerdoti~Coen(2026)]%
        {Nicolo2026}
\bibfield{author}{\bibinfo{person}{Nicolò Pizzo} {and} \bibinfo{person}{Claudio Sacerdoti~Coen}.} \bibinfo{year}{2026}\natexlab{}.
\newblock \showarticletitle{A Reversible Crumbling Abstract Machine for Plotkin's Call-by-Value}. In \bibinfo{booktitle}{\emph{Reversible Computation - 18th International Conference, {RC} 2026, Turin, Italy, July 9-10, 2026, Proceedings}} \emph{(\bibinfo{series}{LNCS}, Vol.~\bibinfo{volume}{16626})}, \bibfield{editor}{\bibinfo{person}{Cl\'{e}ment Aubert} {and} \bibinfo{person}{Luca Roversi}} (Eds.). \bibinfo{publisher}{Springer}, \bibinfo{pages}{41--58}.
\newblock
\href{https://doi.org/10.1007/978-3-032-30839-9_2}{doi:\nolinkurl{10.1007/978-3-032-30839-9_2}}


\bibitem[Reisig(1985)]%
        {Rei85}
\bibfield{author}{\bibinfo{person}{Wolfgang Reisig}.} \bibinfo{year}{1985}\natexlab{}.
\newblock \bibinfo{booktitle}{\emph{Petri Nets, An Introduction}}.
\newblock \bibinfo{publisher}{Springer}.
\newblock
\showISBNx{978-3-642-69970-2}
\href{https://doi.org/10.1007/978-3-642-69968-9}{doi:\nolinkurl{10.1007/978-3-642-69968-9}}


\bibitem[Sano et~al\mbox{.}(2023)]%
        {SanoKP23}
\bibfield{author}{\bibinfo{person}{Chuta Sano}, \bibinfo{person}{Ryan Kavanagh}, {and} \bibinfo{person}{Brigitte Pientka}.} \bibinfo{year}{2023}\natexlab{}.
\newblock \showarticletitle{Mechanizing Session-Types using a Structural View: Enforcing Linearity without Linearity}.
\newblock \bibinfo{journal}{\emph{Proceedings of the ACM on Programming Languages}} \bibinfo{volume}{7}, \bibinfo{number}{{OOPSLA2}} (\bibinfo{year}{2023}), \bibinfo{pages}{374--399}.
\newblock
\href{https://doi.org/10.1145/3622810}{doi:\nolinkurl{10.1145/3622810}}


\bibitem[Sassone et~al\mbox{.}(1996)]%
        {Sassone1996}
\bibfield{author}{\bibinfo{person}{Vladimiro Sassone}, \bibinfo{person}{Mogens Nielsen}, {and} \bibinfo{person}{Glynn Winskel}.} \bibinfo{year}{1996}\natexlab{}.
\newblock \showarticletitle{Models for Concurrency: Towards a Classification}.
\newblock \bibinfo{journal}{\emph{Theoretical Computer Science}} \bibinfo{volume}{170}, \bibinfo{number}{1-2} (\bibinfo{year}{1996}), \bibinfo{pages}{297--348}.
\newblock
\href{https://doi.org/10.1016/S0304-3975(96)80710-9}{doi:\nolinkurl{10.1016/S0304-3975(96)80710-9}}


\bibitem[Trogni et~al\mbox{.}(2026)]%
        {EPTCS448.1}
\bibfield{author}{\bibinfo{person}{Lea Trogni}, \bibinfo{person}{Gabriele Cecilia}, {and} \bibinfo{person}{Alberto Momigliano}.} \bibinfo{year}{2026}\natexlab{}.
\newblock \showarticletitle{Barbed Similarity for the $\pi$-Calculus in Beluga: A Case Study in Coinductive Reasoning}. In \bibinfo{booktitle}{\emph{Proceedings of the 21st Workshop on Logical Frameworks and Meta Languages: Theory and Practice, Lisbon, Portugal, 24th July 2026}} \emph{(\bibinfo{series}{EPTCS}, Vol.~\bibinfo{volume}{448})}, \bibfield{editor}{\bibinfo{person}{Sophie Tourret} {and} \bibinfo{person}{Olivier Hermant}} (Eds.). \bibinfo{publisher}{Open Publishing Association}, \bibinfo{pages}{1--17}.
\newblock
\href{https://doi.org/10.4204/EPTCS.448.1}{doi:\nolinkurl{10.4204/EPTCS.448.1}}


\bibitem[Ulidowski et~al\mbox{.}(2014)]%
        {Ulidowski2014}
\bibfield{author}{\bibinfo{person}{Irek Ulidowski}, \bibinfo{person}{Iain C.~C. Phillips}, {and} \bibinfo{person}{Shoji Yuen}.} \bibinfo{year}{2014}\natexlab{}.
\newblock \showarticletitle{Concurrency and Reversibility}. In \bibinfo{booktitle}{\emph{Reversible Computation - 6th International Conference, {RC} 2014, Kyoto, Japan, July 10-11, 2014. Proceedings}} \emph{(\bibinfo{series}{LNCS}, Vol.~\bibinfo{volume}{8507})}, \bibfield{editor}{\bibinfo{person}{Shigeru Yamashita} {and} \bibinfo{person}{Shin{-}ichi Minato}} (Eds.). \bibinfo{publisher}{Springer}, \bibinfo{pages}{1--14}.
\newblock
\showISBNx{978-3-319-08493-0}
\href{https://doi.org/10.1007/978-3-319-08494-7_1}{doi:\nolinkurl{10.1007/978-3-319-08494-7_1}}


\bibitem[Vall\'ee and Lanese(2026)]%
        {Vallee2026}
\bibfield{author}{\bibinfo{person}{Baptiste Vall\'ee} {and} \bibinfo{person}{Ivan Lanese}.} \bibinfo{year}{2026}\natexlab{}.
\newblock \showarticletitle{On Weak Bisimilarities in CCSK}. In \bibinfo{booktitle}{\emph{Reversible Computation - 18th International Conference, {RC} 2026, Turin, Italy, July 9-10, 2026, Proceedings}} \emph{(\bibinfo{series}{LNCS}, Vol.~\bibinfo{volume}{16626})}, \bibfield{editor}{\bibinfo{person}{Cl\'{e}ment Aubert} {and} \bibinfo{person}{Luca Roversi}} (Eds.). \bibinfo{publisher}{Springer}, \bibinfo{pages}{59--74}.
\newblock
\href{https://doi.org/10.1007/978-3-032-30839-9_3}{doi:\nolinkurl{10.1007/978-3-032-30839-9_3}}


\bibitem[Winskel(1986)]%
        {Winskel1986}
\bibfield{author}{\bibinfo{person}{Glynn Winskel}.} \bibinfo{year}{1986}\natexlab{}.
\newblock \showarticletitle{Event Structures}. In \bibinfo{booktitle}{\emph{Petri Nets: Central Models and Their Properties, Advances in Petri Nets 1986, Part II, Proceedings of an Advanced Course, Bad Honnef, 8.-19. September 1986}} \emph{(\bibinfo{series}{LNCS}, Vol.~\bibinfo{volume}{255})}, \bibfield{editor}{\bibinfo{person}{Wilfried Brauer}, \bibinfo{person}{Wolfgang Reisig}, {and} \bibinfo{person}{Grzegorz Rozenberg}} (Eds.). \bibinfo{publisher}{Springer}, \bibinfo{pages}{325--392}.
\newblock
\showISBNx{3-540-17906-2}
\href{https://doi.org/10.1007/3-540-17906-2_31}{doi:\nolinkurl{10.1007/3-540-17906-2_31}}


\end{thebibliography}
